\documentclass[reqno,a4paper]{amsart}

\usepackage[utf8]{inputenc}
\usepackage[T1]{fontenc}
\usepackage[english]{babel}

\usepackage[pdftex]{graphicx}
\usepackage{latexsym}
\usepackage{amsmath,amssymb,amsthm}
\usepackage{bbm}
\usepackage{mathtools}
\usepackage{color}
\usepackage[normalem]{ulem} 
\usepackage{hyperref}
\usepackage{thmtools}
\usepackage{thm-restate}

\usepackage{hyperref}

\usepackage{cleveref}
\newtheorem{prop}{Proposition}[section]
\newtheorem{corr}[prop]{Corollary}
\newtheorem{theorem}[prop]{Theorem}
\newtheorem{lem}[prop]{Lemma}
\newtheorem{defi}[prop]{Definition}
\newtheorem{rema}[prop]{Remark}
\newtheorem{remas}[prop]{Remarks}

\numberwithin{equation}{section} 

\newcommand{\C}{\mathbb{C}} 
\newcommand{\R}{\mathbb{R}} 
\newcommand{\N}{\mathbb{N}} 

\newcommand{\veps}{\varepsilon} 
\newcommand{\id}{\mathbf{1}}

\newcommand{\jbs}{J_{\bmin}}

\newcommand{\loc}{\mathcal{L}[ \psi]}

\newcommand{\pca}{\mathcal{P}_{\cc_1}}
\newcommand{\pcb}{\mathcal{P}_{\cc_2}}

\newcommand{\pcn}{\mathcal{P}_{\cc_1\cc_2}}
\newcommand{\ol}{\overline}

\newcommand{\lip}{\mathrm{Lip}}  
\newcommand{\dist}{\mathrm{dist}} 

\newcommand{\calL}{\mathcal{L}}

\newcommand{\Ht}{\tilde{H}}
\newcommand{\Hf}{\mathcal{H}_{\mathrm{Fermi}}}
\newcommand{\Hfa}{\mathcal{H}^\alpha}

\newcommand{\ltn}{L^2(\mathbb{R}^{3N})}

\newcommand{\lc}{L^2(\mathbb{R}(\mathcal{C}))}

\newcommand{\Dn}{\mathcal{D}_N^3}

\newcommand{\Da}{\mathcal{D}^{at}}
\newcommand{\cc}{\mathcal{C}}

\newcommand{\resbt}{(\tilde H_\bmin-\mu^\alpha)^{-1}}
\newcommand{\resm}{(\tilde H_{\beta_M}-\mu^\alpha_M)^{-1}}

\newcommand{\Pa}{P^\alpha}

\newcommand{\T}{\mathcal{T}}

\newcommand{\ia}{{\alpha'}}

\newcommand{\ias}{{\alpha^*}}

\newcommand{\U}{\mathcal{U}_\beta}
\newcommand{\Um}{\mathcal{U}_{\beta_M}}

\newcommand{\Wt}{\tilde{\mathcal{W}}_\bmin^\alpha}

\newcommand{\fA}{f_2}
\newcommand{\fB}{f_3}

\newcommand{\fkl}{f^{(k,l)}}
\newcommand{\phiA}{\phi}
\newcommand{\phiB}{\phi_2}
\newcommand{\phiC}{\phi_3}
\newcommand{\gamA}{\gamma_1}
\newcommand{\gamB}{\gamma_2}
\newcommand{\gamC}{\gamma_3}

\newcommand{\las}{\langle}
\newcommand{\ras}{\rangle}
\newcommand{\lk}{\Big(}
\newcommand{\rk}{\Big)}

\newcommand{\sumi}{\sum_{i\in \cc_1}}
\newcommand{\sumj}{\sum_{j\in \cc_2}}
\newcommand{\sumij}{\sum_{\substack{i\in \cc_1\\j\in \cc_2}}}
\newcommand{\re}{\operatorname{Re}}

\newcommand{\gagas}{2 \re \gamA \overline{\gamC}}
\newcommand{\supp}{\operatorname{supp}}
\newcommand{\zerz}{\zeta_{R,Z}}

\newcommand{\tpsi}{\hat \psi}

\newcommand{\bsp}{{{\beta}^0}}

\newcommand{\bmin}{{\beta}}

\newcommand{\D}{|D|}

\newcommand{\nF}{F}  
\newcommand{\nG}{\Upsilon}  
\newcommand{\ga}{\Gamma} 

\title[Van der Waals-London interaction of atoms with
pseudo-relativistic kinetic energy]{Van der Waals-London interaction of atoms with
pseudo-relativistic kinetic energy}

\author{Jean-Marie Barbaroux}
\address{Jean-Marie Barbaroux\\
 Aix Marseille Univ, Universit\'e de Toulon, CNRS, CPT, Marseille, France.}
\email{jean-marie.barbaroux@univ-tln.fr}
\author{Michael C. Hartig}
\address{
Michael C. Hartig\\
 Aix Marseille Univ, Universit\'e de Toulon, CNRS, CPT, Marseille, France.}
\email{michael.hartig@univ-tln.fr}
\author{Dirk Hundertmark}
\address{
Dirk Hundertmark\\
Institute for Analysis, Karlsruhe Institute of Technology (KIT), Englerstraße 2,
76131 Karlsruhe, Germany, and
Department of Mathematics
University of Illinois at Urbana-Champaign
1409 W. Green Street
Urbana, Illinois 61801-2975 
}
\email{dirk.hundertmark@kit.edu}
\author{Semjon Vugalter}
\address{
Semjon Vugalter\\
Institute for Analysis, Karlsruhe Institute of Technology (KIT), Englerstraße 2,
76131 Karlsruhe, Germany.
}
\email{semjon.wugalter@kit.edu}
\thanks{\today, version \jobname}
\begin{document}
\begin{abstract}
We consider a multiatomic system where the nuclei are assumed to be point charges at fixed positions. Particles interact via Coulomb potential and electrons have pseudo--relativistic kinetic energy. We prove the van der Waals-London law, which states that the interaction energy between neutral atoms decays as the sixth power of the distance $\D$ between the atoms. 
In the many atom case, we rigorously compute all the terms in the binding energy up to the order $\D^{-9}$ with error term of order $\mathcal{O}(\D^{-10})$. This yields the first proof of the famous Axilrod--Teller--Muto three--body correction to the van der Waals--London interaction, which plays an important role in 
	atom physics.  As intermediate steps we prove exponential decay of eigenfunctions of multiparticle Schrödinger operators with permutation symmetry imposed by the Pauli principle, and new estimates of the localization error.
\end{abstract}
\maketitle
\tableofcontents
\section{Introduction}\label{sec_int}
The van der Waals--London force plays a vital role in many natural phenomena. Its importance for the structure, stability and function of molecules and materials can hardly be overemphasized.  To give a few examples, the van der Waals--London force is needed to explain the condensation of water from vapor, the structural stability of DNA, and the binding between several layers of graphene to form graphite. \par
The importance of the van der Waals--London force is not restricted to the microscopic scale. The van der Waals--London forces are used to explain some biological processes and there are efforts in nanotechnology to take advantage of this attractive force. For further examples, see the introductory discussion in \cite{ioannis} or \cite{DiStasio-2014} and the references therein. 
\par 
Surprisingly enough there are only few mathematically rigorous results concerning the van der Waals--London force. In \cite{morgansimon1980}, J.~D.~Morgan and B.~Simon proved the existence of an asymptotic expansion of the interaction energy using perturbation theory. They note that this asymptotic series neither converges nor is Borel summable.  Moreover, under the \emph{assumption that individual atoms have no dipole nor  quadrupole moments}, the leading behaviour of their asymptotic series is $\mathcal O(\D^{-6})$, where $\D$ is the distance between two nuclei, but they do not give an explicit expression for the coefficient of the leading order term nor do they prove that the asymptotic starts with the term of order $\D^{-6}$. We will compare their method with ours in more detail later in the introduction.
\par
Another result concerning van der Waals--London interaction is obtained in \cite{lieb86} by E.~H.~Lieb and W.~E.~Thirring where they constructed a trial function to show that attractive energy between two atoms without permanent polarity is at least $-C\D^{-6}$ for some positive constant $C$.
\par 
This result was improved by I.~Anapolitanos and I.~M.~Sigal in \cite{ioannis}, who used the Feshbach--Schur method to obtain under some restrictions, which we will discuss later, the leading term of order $\D^{-6}$ for the intercluster energy in the nonrelativistic case with an error $\mathcal{O}(\D^{-7})$. 
Later the remainder term in the van der Waals--London force was analyzed in \cite{Anapolitanos} using again the Feshbach Schur method.  As our work shows, the bound for the remainder in \cite{Anapolitanos} is far from being optimal, see the discussion after Theorem \ref{thm_pre_03}. In a recent work, I.\ Anapolitanos and M.\ Lewin \cite{Anapolitanos-Lewin} considered the van der Waals--London interaction for molecules. 
	The difference between atoms and molecules is that it is easier for molecules to have a permanent dipole moment in their ground state. This leads to the possibility 
	of interactions decaying slower than the van der Waals--London interaction, or decaying with the same rate but having a different physical origin, see section C in \cite{Anapolitanos-Lewin}.
\par 
Note that all previous rigorous results were only proven for non--relativistic kinetic 
energies. For heavy atoms one should include relativistic effects for 
the electrons. This is one of the main goals of the work at hand. Our 
approach is purely variational and 
similar to the one used in \cite{VuZhis02,VuZhis03,VuZhis01}  to 
obtain the asymptotics of eigenvalues of multiparticle Schr\"odinger operators near the bottom of the essential spectrum and in 
\cite{BaVu04,BaVu03,BaVu01,BaVu02} to get the asymptotics for the binding energy of the Pauli--Fierz operator.

 Most  importantly, in Theorem \ref{thm_pre_03} below we rigorously prove the famous Axilrod--Teller--Muto $D^{-9}$ three body correction to the van der Waals--London interaction  which is a genuine non--additive three body effect and which plays an important role in the case of three or more interacting atoms in 
	atom physics \cite{Lilienfeld-Tkatchenko, Axilrod-Teller, DiStasio-2014, Muto}. 
	To the best of our knowledge, this has never been rigorously shown before.

In addition, we believe that our variational approach has several  advantages over other approaches using  the Feshbach--Schur map: First, to be able to use the Feshbach--Schur map, the authors in \cite{Anapolitanos, ioannis} need to show that the ground state is isolated before they could apply the Feshbach--Schur map. On the other hand, although we do not need this explicitly in this paper, our method also works when the ground state is not isolated from the continuum, see, for example,   \cite{BaVu03,BaVu01,BaVu02}, where this has been carried out in a different situation. Secondly, it is known from physical heuristics that the reason for the van der Waals--London attraction of neutral atoms is due to induced virtual dipole moments, which show up in high enough orders of perturbation theory. So on a heuristic level the origin of the van der Waals--London attraction is well--understood. These calculations are far from rigorous, however.  
Our variational approach uses a construction of trial function which is motivated by the physical intuition gained from second order perturbation theory to get a precise upper bound for the van der Waals--London attraction. 
To get a matching lower bound, we use geometric methods based on suitable partitions of unity of the configuration space which is an extension of ideas in \cite{VuZhis02, VuZhis03, VuZhis01}.  Thus our variational approach is not only motivated by informal calculations based on perturbation theory but it also justifies these calculations. Moreover, the inherent simplicity of our method -- follow perturbation theory and make it rigorous for \emph{upper and lower} bounds -- enables us to relatively get precise results for some higher order terms, given the complexity of the many--body problem. For example, for two atoms, we show that the terms of order $D^{-7}$ and $D^{-9}$ in the van der Waals--London interaction do not exist, while for three or more atoms we rigorously establish the Axilrod--Teller--Muto correction. 

	Comparing our method with the early work  of J.\ D.\ Morgan and B.\ Simon, it is important to mention that they also use trial functions for the upper bound and geometrical methods for the lower bound on the interaction energy. 
	The difference with the approach of Morgan and Simon and our work is that in \cite{morgansimon1980} the ground state energy was estimated with an error of order $D^{-1}$, to show that eigenvalues of a Schr\"odinger operator with inter--cluster interaction  converge to the eigenvalues of the cluster operators without inter--cluster interaction at large distances. They use  this then later to justify a clever perturbation theory approach. 
	In the work at hand, we estimate the ground state much more precisely using corrections terms motivated from formal second order perturbation theory. This allows us to not only obtain the leading order van der Waals--London term  but also rigorously establish higher order corrections, including the Axilrod--Teller--Muto correction.

We consider a molecule with $N$ electrons of charge $-e$ and spin $\frac12$, and $M$ pointwise nuclei with charges $eZ_l$ located at positions $X_l$ in $\R^3$, which we suppose to be fixed (Born-Oppenheimer approximation). We assume that the system is neutral, which means that $\sum_{l=1}^M Z_l=N$. The corresponding Hamiltonian is 
\begin{equation}\label{eq_int_10}
H:=\sum_{i=1}^N\left(T_i-\sum_{l=1}^M\frac{e^2Z_l}{|x_i-X_l|}\right)+\sum_{1\leq i<j\leq N} \frac{e^2}{|x_i-x_j|}+\sum_{1\leq k<l\leq M}\frac{e^2Z_kZ_l}{|X_k-X_l|}
\end{equation}
with $k$-th electron kinetic energy operator
\begin{equation}
T_k:= \left\{\begin{array}{ll}\sqrt{p_k^2+1}-1 &\mbox{ in the pseudo--relativistic case }\\
\quad\ \frac{p_k^2}{2}&\mbox{ in the nonrelativistic case }
\end{array}\right. 
\end{equation}
and form domain $H^{1/2}(\R^{3N})$ in the pseudo--relativistic case and $H^1(\R^{3N})$ in the nonrelativistic case.
As usual $p_k = - i\nabla_{x_k}$ denotes the momentum of the $k$-th electron. If $T_k$ is pseudo--relativistic, we assumed $Z_le^2\leq \frac{2}{\pi},$ which ensures that the Hamiltonian is semi--bounded from below, see  \cite{fefferman1986,lieb1988}.
\par 
In the main part of the paper we will focus on the pseudo--relativistic kinetic energy case $T_k=\sqrt{p_k^2+1}-1$ (see \cite{herbst} and references therein) although all the results hold for $T_k=\frac{p_k^2}{2}$ likewise. Here the Hamiltonian is written in atomic units, i.e. $c=\hbar=m=1$.
\par 
The phase space for a system of $N$ electrons, taking into account the Pauli-principle, is the antisymmetric tensor product of $N$ copies of $L^2(\R^3;\C^2)$, namely the space $\bigwedge^N L^2(\R^3;\C^2)$ of functions in $\bigotimes^N L^2(\R^3;\C^2)$ that are antisymmetric with respect to transpositions of pairs of position and spin particle variables $(x_i, s_i)$ and $(x_j,s_j)$, for $i\neq j$.
\par 
The operator $H$ we consider does only depend on the coordinate variables $x_i$, but not on spin variables $s_i$. Hence we consider $H$ to act on the projection of $\bigwedge^N L^2(\R^3;\C^2)$ onto the space of functions depending on coordinates alone, that is, on the space $\Hf$ defined by
\begin{equation}\label{eq_int_20}
\Hf:=\Big\{ \langle\mathfrak{s},\Psi\rangle_{\mathrm{spin}} \vert \Psi\in \bigwedge ^N L^2(\R^3;\C^2), \mathfrak{s}:\Big\{-\frac{1}{2},\frac{1}{2}\Big\}^N\rightarrow\C \Big\}
\end{equation}
where
\begin{equation}\nonumber
\langle \mathfrak{s},\Psi\rangle_\mathrm{spin} :=
	\sum_s\bar{\mathfrak{s}}(s_1,\cdots,s_N)\Psi(x_1,s_1,\cdots,x_N,s_N).
\end{equation}
Note that $\Hf$ is a subspace of $\ltn$.

The condition of antisymmetry with respect to transposition of the particle variables implies certain symmetry properties for permutations of coordinate variables after decoupling of the spin variables. Namely, permutations of electrons transform the functions according to a Young pattern with at most two columns as described in \cite[§ 7.3.]{hamermesh1962}. Note that for more than two particles a function which is completely symmetric under transposition of coordinate variables can never be antisymmetric under transposition of the full particle variables, since the spin can only attain two values.
\par
More precisely, let $S_N$ be the group of permutations of $N$ electrons. For any $\pi\in S_N$ let $\T_\pi: \Hf \rightarrow \Hf$ with
\begin{equation}\label{eq_int_11}
\T_\pi\psi(x_1,\cdots,x_N):=\psi(x_{\pi^{-1}(1)},\cdots,x_{\pi^{-1}(N)})
\end{equation}
be the operator that realizes a permutation on the particle variables.
\par 
Let $\alpha$ be an irreducible representation of the group $S_N$ and $\Pa$ the projection on the subspace of functions transformed under the action of operators $\T_\pi$ according to the representation $\alpha$. These projections decompose the space $\Hf$ into a finite number of orthogonal subspaces $\Hfa:=\Pa\Hf$ such that
\begin{equation}
\Hf=\bigoplus_{\alpha\in\mathcal{A}}\Hfa ,
\end{equation}
where $\mathcal{A}$ is the set of all irreducible representations of the group $S_N$ corresponding to a Young pattern with at most two columns. Note that for such $\alpha$, we have $P^\alpha\Hf=P^\alpha \ltn$.  In fact, studying the operator $H$ on the subspaces $P^\alpha L^2(\R^{3N})$ gives us complete information on the spectrum of the operator on $\Hf$.
To that end let
\begin{equation}\label{eq_int_04}
H^\alpha:=HP^\alpha
\end{equation}
be the operator $H$ restricted to the space $\Hfa$ and
\begin{equation}\label{eq_pre_01}
E_{(X_1,\cdots,X_M)}^\alpha:=\inf\sigma(H^\alpha).
\end{equation}
In the work at hand, we will compute the interaction energy for fixed positions of the nuclei, which is the difference between $E^\alpha_{(X_1,\cdots,X_M)}$ and the sum of ground state energies of atoms. Let us start with the simplest case of a diatomic molecule, i.e. $M=2$.
\subsection{Diatomic molecules}

\par 
Let $\cc\subsetneq \{1,\cdots,N\},\  \cc\neq \emptyset$ be an arbitrary subsystem of a system of $N$ electrons. We define $\R(\cc)$ as the vector space of position vectors $(x_i)_{i\in \cc}$ of particles in $\cc$. Note that this space is isomorphic to $\R^{3\sharp\cc}$, where $\sharp \cc$ is the number of elements in $\cc$.
We let $L^2(\R(\cc))$ be the space of $L^2$-functions with arguments in $\R(\cc)$. Denote by $L^2(\R(\cc))^\perp$ the orthogonal complement in $L^2(\R^{3N})$ of $L^2(\R(\cc))$.\par 
For particles in $\cc$ interacting via Coulomb potential with a nucleus at the origin of charge $eZ$ we define the Hamiltonian 
\begin{equation}\label{eq_int_02}
 \tilde H_\cc^Z:= \sum_{i\in \cc} T_i - \sum_{i\in \cc} \frac{e^2 Z}{|x_i|}+\sum_{\substack{i,j\in \cc\\ i<j}}\frac{e^2}{|x_i-x_j|}
\end{equation}
acting on $L^2(\R(\cc))$. We extend the operator by the identity in $L^2(\R(\cc))^\perp$ to an operator acting on functions in $\ltn$.
In abuse of notation we will write $\tilde H_\cc^Z$ for both, the one acting on $L^2(\R(\cc))$ and the operator acting on $L^2(\R(\cc))\oplus~L^2(\R(\cc))^\perp$.\par 
Let $S(\cc)$ be the group of permutations within $\cc$. Obviously $S(\cc)$ is a subgroup of $S_N$. Consider $\alpha_\cc$ to be an irreducible representation of $S(\cc)$.
\begin{defi}\label{def_01} For $\alpha$ a type of irreducible representation of $S_N$, we say that $\alpha'_\cc$ is \emph{induced} by $\alpha$ and write
 $\alpha'_\cc\prec\alpha$, if $\alpha'_\cc$ is contained in the restriction of $\alpha$ to $S(\cc)$, see \cite[p. 94-98]{hamermesh1962}.\par 
 \end{defi}
In the same way as the space $\Hf$ can be decomposed into the spaces $\Hfa$, the corresponding Fermi subspace of $L^2(\R(\cc))$ can be decomposed into subspaces $$P^{\alpha_\cc}\lc$$ where $\alpha_\cc$ runs over all irreducible representations of $S(\cc)$ corresponding to a Young pattern of at most two columns.\par 
We will consider a cluster decomposition $\beta=(\cc_1,\cc_2)$ of the original system $\{1,\cdots,N\}$ into clusters $\cc_1$ and $\cc_2$ such that $\cc_1\cup\cc_2=\{1,\cdots,N\}$ and $\cc_1\cap \cc_2=\emptyset$. Define $\mathcal{D}_N^2$ as the set of all such decompositions. Decompositions where the number of electrons in $\cc_1$, $\sharp\cc_1=Z_1$ and the number of electrons in $\cc_2$, $\sharp\cc_2=Z_2$ will be called \emph{atomic} decomposition $\Da \subset \mathcal{D}_N^2$.\par 
For the decomposition $\beta =(\cc_1,\cc_2)$ we define the intercluster interaction
\begin{equation}
I_\beta:=\sum_{i\in \cc_1} \frac{-e^2Z_2}{|x_i-X_2|}+\sum_{j\in \cc_2} \frac{-e^2Z_1}{|x_j-X_1|}+\sumij \frac{e^2}{|x_i-x_j|}+\frac{e^2Z_1Z_2}{|X_2-X_1|}
\end{equation}
and set the cluster Hamiltonian $H_\beta$ to be
\begin{equation}
H_\beta := H-I_\beta.
\end{equation}
In other words, $H_\beta$ is the operator where particles from different subsystems do not interact.
Note that for each $\beta \in \mathcal{D}_N^2$ we have $L^2(\R(\cc_1))^\perp=L^2(\R(\cc_2))$ . The symmetry group of this Hamiltonian we consider is $S_\beta:=S(\cc_1)\times S(\cc_2)\subset S_N$, the group of permutations which leave the cluster decomposition $\beta$ intact.
We use the same notion of inducing of representations as above. Since $S_\beta$ is a direct product of two groups, the irreducible representations $\alpha_\beta$ of $S_\beta$ are direct products too. In particular, for any irreducible representation $\alpha'_\beta\prec \alpha$ of $S_\beta$ there exists a unique pair $\alpha'_{\cc_1}\prec \alpha$ and $\alpha'_{\cc_2} \prec\alpha$ such that
\begin{equation} 
 \ia _{\cc_1}\otimes \ia_{\cc_2}\cong \alpha'_\beta,
 \end{equation}
see \cite[p. 110-114]{hamermesh1962}. We take $P^{\alpha'_\beta}$ to be the projection in $\Hf$ onto functions of symmetry type $\alpha'_\beta$.
Letting 
\begin{equation}
H^{\alpha'_ \beta}_\beta:=H_\beta P^{\alpha'_\beta}
\ \text{ and }\ 
H_\beta^\alpha:= \sum_{\alpha'_\beta\prec\alpha} H_\beta^{\alpha'_\beta},
\end{equation}
we define
\begin{equation}\label{eq_int_08}
\mu_\beta^\alpha:=\min_{\alpha'_\beta\prec\alpha} \inf \sigma(H_\beta^{\alpha'_\beta})
\end{equation}
and
\begin{equation}\label{eq_int_09}
\mu^\alpha:=\min_{\beta\in \mathcal{D}_N^2}\mu^\alpha_\beta.
\end{equation}
By translation and rotation invariance of the Hamiltonian for $M=2$, $E^\alpha_{(X_1,X_2)}$ only depends on $|D|$, where $D:=X_2-X_1$. We will write $E_{|D|}^\alpha$ instead of $E^\alpha_{(X_1,X_2)}$. In both, the pseudo--relativistic and the nonrelativistic, cases it is not difficult to see that $\mu^\alpha=\lim_{|D|\rightarrow \infty} E^\alpha_{|D|}$.\par 
For some fixed point $X\in \R^3$, which will be the position of one of the nuclei, and the variable $x\in \R^{3N}$, we define the unitary shift by $X$ in the $i$-th particle variable as
\begin{equation}\label{eq_int_01}
\mathcal{U}^{(i)}_X:\left\{\begin{array}{ll}
   \ltn\rightarrow\ltn \\
   \mathcal{U}^{(i)}_X\varphi(x) \mapsto\varphi(x_1,\cdots,x_{i-1},x_i+X,x_{i+1},\cdots,x_N).
\end{array}\right.
\end{equation}
For $\beta =(\cc_1,\cc_2)\in \mathcal{D}_N^2$ and $X_1, X_2$ being the positions of the nuclei we define the shift operators 
\begin{equation}\label{def_ubeta}
\U:= \prod_{i\in \cc_1} \mathcal{U}_{X_1}^{(i)}\prod_{j\in \cc_2} \mathcal{U}_{X_2}^{(j)}.
\end{equation}
We set
\begin{equation}\label{def_hbeta}
\tilde H_\beta := \U H_\beta \U^*.
\end{equation}
Note that $\tilde H_\beta$ is unitary equivalent to $H_ \beta$ and
\begin{equation}\nonumber
\tilde H_\beta = \tilde H_ {\cc_1}^{Z_1}+\tilde H_{\cc_2}^{Z_2}.
\end{equation}
We define for $\beta\in \mathcal{D}_N^2$ the functions $f_2,f_3\in L^2(\R^{3N})$ as 
\begin{equation}\label{eq:def-fA}
f_2(x):=\sumij -e^2\big(3(x_i \cdot e_D)(x_j\cdot e_D)-x_i\cdot x_j\big),
\end{equation}
\begin{equation}\label{eq:def-fB}
\begin{split}
 f_3(x) :=\sumij & \frac{e^2}{2} \Big(  3(x_i-x_j)\cdot e_D\big[2(x_i\cdot x_j)-5(x_i\cdot e_D)(x_j\cdot e_D)\big]\\
 & \qquad+ 3|x_i|^2(x_j\cdot e_D)-3|x_j|^2(x_i\cdot e_D)\Big),
\end{split}
\end{equation}
where $e_D:=\frac{D}{|D|}$, a unit vector in the direction from $X_1$ to $X_2$.
Note that the functions $\fA, \fB$ depend on the cluster decomposition $\beta$. These functions stem from a Taylor expansion of the Coulomb interactions representing dipole--dipole, respectively dipole--quadropole type interactions.
\par
For now, let us fix any $\beta\in \Da$.  
We will show in Appendix~\ref{sec_exist} that $\mu^\alpha$ is a discrete eigenvalue of $H_\beta^\alpha$. By unitary equivalence $\mu^\alpha$ is also a discrete eigenvalue of 
\begin{equation}
\tilde H_\beta^\alpha := \sum_{\alpha_\beta'\prec\alpha} \tilde H _\beta ^{\alpha'_\beta} :=\sum_{\alpha_\beta'\prec\alpha} \tilde H _\beta P^{\alpha'_\beta}
\end{equation}
where the sum is over all induced irreducible representations $\alpha'_\beta \prec \alpha$.
Denote by $\tilde {\mathcal{W}}_\beta^\alpha\subset \Hfa$ the eigenspace of $\tilde H^\alpha_\beta$ corresponding to $\mu^\alpha$ and let
\begin{equation}\label{eq_int_06}
a_1(\beta):=\max_{\substack{\phi\in \tilde{\mathcal{W}}_\beta^\alpha\\
\Vert \phi\Vert=1}}\Vert (\tilde H_\beta-\mu ^\alpha)^{-\frac{1}{2}}f_2\phi\Vert^2.
\end{equation}
Although $\mu^\alpha$ is an eigenvalue of $\tilde H_\beta$ the value $a_1(\beta)$ is well-defined since $f_2\phi$ is orthogonal to the corresponding eigenspace, see Lemma~\ref{lem_app_03A}. We define $\tilde{\mathcal{V}}_\beta^\alpha\subset \tilde{\mathcal{W}}_\beta^\alpha$ as the subspace of all $\phi$ such that $\Vert (\tilde H_\beta-\mu^\alpha)^{-\frac{1}{2}}f_2\phi\Vert^2=a_1(\beta)$ and
\begin{equation}\label{eq_int_07}
a_2(\beta):=\max _{\substack{\phi\in \tilde{\mathcal{V}}_\beta^\alpha\\\Vert \phi\Vert=1}}\Vert (\tilde H_\beta-\mu^\alpha)^{-\frac{1}{2}}f_3 \phi \Vert^2.
\end{equation}
Similarly, Lemma~\ref{lem_app_03A} ensures that also $a_2(\beta)$ is well-defined. Due to permutational symmetry, for any $\beta_1,\beta_2\in  \Da$ we have $a_1(\beta_1)=a_1(\beta_2)$ and $a_2(\beta_1)=a_2(\beta_2)$. Hence we omit the argument $\beta$ in the definition and write $a_1$ and $a_2$ throughout the paper. For diatomic molecules our main result is
\begin{theorem}[van der Waals--London interaction for diatomic molecules]\label{thm_pre_02}
Assume that $Z_le^2\leq \frac{2}{\pi},$ for all nuclear charges when the kinetic energy of the electrons is taken to be pseudo--relativistic. 
Let $\alpha$  be an irreducible representation of $S_N$  corresponding to a Young pattern with at most two columns and assume that 
\par 
\noindent 1) For all $\beta\in\mathcal{D}_N^2\setminus\Da$
\begin{equation}\nonumber
\mu_\beta^\alpha > \mu^\alpha .
\end{equation}
\par 
\noindent 2) For each $\beta \in \Da$ and each irreducible representation $\alpha^*_\beta$ of the group $S_\beta$ with $\alpha^*_\beta\prec \alpha$ such that $P^{\alpha^*_\beta}\tilde{\mathcal{W}}_\beta^{\alpha}\neq \emptyset$, $$\ \dim (P^{\alpha^*_\beta}\tilde{\mathcal{W}}_\beta^\alpha)=\dim \alpha^*_\beta.$$
\par 
Then
\begin{equation}\label{eq_pre_03}
E^\alpha_{|D|} - \mu^\alpha=-\frac{a_1}{|D|^6}-\frac{a_2}{|D|^8}+\mathcal{O}(|D|^{-10})
\end{equation}
where $a_1>0$ and $a_2>0$ are defined in \eqref{eq_int_06} and \eqref{eq_int_07} respectively.
\end{theorem}  
\begin{remas} 
\begin{itemize}
\item Conditions 1) and 2) of Theorem \ref{thm_pre_02} are the same as in the previous work \cite{ioannis} by I. Anapolitanos and I. M. Sigal, where they obtained an asymptotic expansion of $E^\alpha _{\D}-\mu^\alpha$ in the nonrelativistic case with an error of order $\mathcal{O}(\D^{-7})$.
\item The physical meaning of Condition~1) is that the lowest energy of the non-interacting system occurs when the electrons are allocated neutrally. It is important to mention that if Condition 1) does not hold, then $E_{\D}^\alpha-\mu^\alpha$ is dominated by Coulomb interaction which decays like $\D^{-1}$ and is thus much stronger than the van der Waals-London interaction. Both variants are possible. Experimental data shows that for some molecules Condition 1) is fulfilled and for some it is not, see discussion in the introduction of \cite{ioannis}. 
\item Condition 2) imposes restrictions on the rotational symmetry of the atoms in the diatomic molecule. In particular the ground state space of $\tilde H_\beta^\alpha$ only contains functions which transform according to the irreducible representation of the group $SO(3)$ of degree $\ell=0$. To see this, notice that the Hamiltonian $\tilde H_\beta^\alpha$ is invariant under rotations $R\in SO(3)$. Thus for any eigenfunction $\phi\in P^{\alpha'_\beta} \tilde {\mathcal{W}}_\beta^\alpha$ the rotated function $T_R\phi$ is an eigenfunction corresponding to the same value. Rotation and permutation operators commute, thus $T_R\phi\in P^{\alpha'_\beta}\tilde {\mathcal{W}}_\beta^\alpha$. So by \cite[\S 3.19]{hamermesh1962} the dimension of $P^{\alpha'_\beta} \tilde {\mathcal{W}}_\beta^\alpha$ is an integer multiple of the dimension of $\alpha'_\beta$ and the dimension of a representation of the $SO(3)$ group. By Condition~2) $\dim (P^{\alpha'_\beta}\tilde{\mathcal{W}}_\beta^\alpha)=\dim \alpha'_\beta$ so the dimension of the representation of $SO(3)$ describing the symmetry of $\phi$ is one. So it must be the irreducible representation of degree $\ell=0$. 
\item Our method allows to obtain the expansion of $E_{\D}^\alpha-\mu^\alpha$ up to arbitrary negative power of $|D|$. In particular, for diatomic molecules this expansion does not include odd powers $|D|^{-7}$ and $\D^{-9}$ in both the pseudo--relativistic and nonrelativistic case. 
	There is a correction, the famous Axilrod--Teller--Moto correction 
	to the van der Waals law, which starts with the $\D^{-9}$ term. 
	However, it is well understood in the physics literature that 
	this correction is due to interactions 
	between triplets of atoms, hence it should be absent for diatomic 
	molecules. As our Theorem \ref{thm_pre_02} shows, this is indeed 
	the case.  For three or more atoms, this correction is present, see Theorem \ref{thm_pre_03}. 

\item In the definition of the functions $f_2$, $f_3$ and therefore in the definition of $a_1$ and $a_2$, we use the vector $e_D$. By the $SO(3)$ symmetry of $\tilde H_\beta$ and Condition 2), the values of $a_1$ and $a_2$ will not change if we replace $e_D$ in \eqref{eq:def-fA} and \eqref{eq:def-fB} with an arbitrary normalized vector in $\R^3$.
\item The functions $f_2$, $f_3$ are, respectively, the second- and third-order coefficients in the Taylor expansion of the intercluster interaction (see Appendix~\ref{sec_app_01}). They are invariant under permutations in $S_\beta$ and hence for any irreducible representation $\alpha'_\beta\prec\alpha$ of $S_\beta$, we have $f_l P^{\alpha'_\beta}=P^{\alpha'_\beta} f_l$, for $l=2,3$.
\end{itemize}
\end{remas}
\textbf{Strategy of the proof of Theorem~\ref{thm_pre_02}} To prove the main result, we derive estimates of the difference $E_{\D}^\alpha - \mu ^\alpha$ from above and from below. These bounds coincide up to an order $\mathcal{O}(\D^{-10})$. To get an estimate from below for the interaction energy, we apply a partition of unity to the configuration space, and minimize the functionals in the corresponding regions. To obtain an upper bound, we construct a suitable trial function.
\par 
More precisely, let $\bsp=( \{1,\cdots,Z_1\},$ $\{Z_1+1,\cdots,N\})$. By permutation symmetry of the operator $\tilde H_\bsp^\alpha$, the ground state space $\tilde{\mathcal{W}}_\bsp^\alpha$ of $\tilde H_\bsp^\alpha$ can be written as a direct sum of subspaces transforming according to the induced irreducible representations $\alpha'_\bsp \prec \alpha$, more explicitly
\begin{equation}\nonumber
\tilde{\mathcal{W}}_\bsp^\alpha = \bigoplus_{\alpha'_\bsp \prec \alpha} P^{\alpha'_\bsp} \tilde{\mathcal{W}}_\bsp ^\alpha.
\end{equation}
Thus there is at least one $\alpha^*_\bsp \prec\alpha$ such that there exists $\phi\in P^{\alpha^*_\bsp} \tilde {\mathcal{W}}_\bsp ^\alpha$ that realises the maxima $a_1$ and $a_2$ with $\|\phi\|=1$. For such a $\phi\in P^{\alpha^*_\bsp} \tilde{\mathcal{W}}_\bsp^\alpha$ we define
\begin{equation}\label{eq_int_18}
\begin{split}
 \Upsilon : =
   \chi_{o}(x) \Bigg\{ 
 \phi(x)  
    - (\tilde H_\bsp-\mu^\alpha)^{-1}\Big( \frac{f_2(x)}{|D|^3} 
 + \frac{f_3(x)}{|D|^4} \Big) \phi(x)
 \Bigg\}
\end{split}
\end{equation}
where $\chi_{o}(x)$ is a smooth function which localizes each particle in a ball of radius $ |D|^\frac34$, centered at the origin. As a trial function, which yields the required estimate of $E^\alpha_{\D}-\mu^\alpha$ from above, we define $\Upsilon_{\mathrm{trial}}:=P^\alpha \mathcal{U}^*_\bsp \Upsilon$.\par
To prove the estimate from above, we need to show that applying the cutoff function $\chi_o(x)$ increases the energy only by an exponentially small amount. To this end we need to prove exponential decay of $\phi, (\tilde H_\bsp-\mu^\alpha)^{-1}f_2\phi$, and $(\tilde H_\bsp-\mu^\alpha)^{-1}f_3 \phi$, which is done in Section \ref{sec_exp}. In addition, we need a suitable estimate for the so-called localization error for the pseudo--relativistic kinetic energy. Such an estimate is obtained in Section \ref{sec_loc}. In both cases, the proof of exponential decay and the estimate of the localization error, the main difficulty arises from the non-locality of the pseudo--relativistic kinetic energy operator.\par 
For the estimate from below we consider all possible cluster decompositions into three clusters $\beta = (\cc_0, \cc_1, \cc_2).$ Some of the clusters may be empty. Particles in $\cc_0$ are far from the nucleus. Electrons in $\cc_1$ and $\cc_2$ are close to $X_1$ and $X_2$ respectively. We apply a partition of unity of the configurations space with smooth functions $J_\beta$ cutting the configuration space according to the clusters in $\beta$. If $\cc_0 \neq \emptyset $ or if $\cc_1$ and $\cc_2$ are not neutral atoms, the infimum of the spectrum of the cluster Hamiltonian corresponding to this $\beta$ on the subspace $\mathcal H ^{\alpha'_\beta}$ is, by assumption, strictly greater than $\mu^\alpha$ for all $\alpha'_\beta\prec \alpha$.
For sufficiently large $|D|$, this implies
\begin{equation}\nonumber
		\las  J_\beta \psi ,(H^\alpha - \mu^{\alpha} ) J_\beta \psi \ras  \ge 0.
\end{equation}
Now consider $\beta$ for which $\cc_0 = \emptyset$, and $(\cc_1,\cc_2)\in \Da$. Similar to \cite{BaVu04,BaVu03,BaVu01,BaVu02,VuZhis02,VuZhis03,VuZhis01} we define a bilinear form 
\begin{equation}\nonumber
		\las\varphi , \psi\ras_1 := \las\varphi , (\tilde H_\beta - \mu^{\alpha}) \psi\ras
\end{equation}
		and the corresponding semi--norm
\begin{equation}\nonumber
		\| \varphi\|^2_1 := {\las\varphi, \varphi \ras}_1.
\end{equation}
Then we project the state $\U J_\beta \psi$ onto the ground state subspace $\tilde{\mathcal{W}}_\beta ^\alpha$ of the operator $\tilde H_\beta ^\alpha$ to get
\begin{equation}\nonumber
\U J_\beta \psi = \gamma _1 \phi + \mathcal{R}
\end{equation}	
for a normalized state $\phi \in \tilde {\mathcal{W}}_\beta^\alpha$. We proceed by projecting the rest term $\mathcal{R}$ onto the functions
		\begin{equation}\label{eq_int_22}
			\begin{array}{lcl}
		\phi _2 & = & (\tilde H_\beta - \mu^{\alpha})^{-1} \ f_2 \phi , \\
		\phi _3 & = &  (\tilde H_\beta - \mu^{\alpha})^{-1} \ f_3 \phi
			\end{array}
\end{equation}
consecutively, with respect to $\las \cdot , \cdot \ras_1$. Note that by Corollary~\ref{cor_app_01} these states are well-defined. For the state $J_\beta \psi$ we arrive at the following representation
\begin{equation} \label{eq_int_21}
		J_\beta \psi  =\U ^*\big( \gamma _1 \phi + |D|^{-3} \gamma _2 \phi _2 + |D|^{-4} \gamma _3 \phi _3 + g\big)
\end{equation}
for a suitable function $g$.
We substitute \eqref{eq_int_21} into the quadratic form of
\begin{equation}\nonumber
(H-\mu^\alpha)P^\alpha=(H_\beta -\mu^\alpha + I_\beta\big)P^\alpha.
\end{equation}
Then we expand $I_\beta$ as a Taylor series and do a simple minimization in parameters $\gamma _1, \gamma _2, \gamma _3$, using orthogonality relations proven in Appendix~\ref{sec_app_01}. It turns out that $\Vert g \Vert$  will be very small and $\gamma _1, \gamma _2, \gamma _3$ close to the coefficients of the trial function, which we used to get the upper bound, when $\psi$ is close to a minimizer of the energy.\par 
Finally, in analogy to the estimate from above, the localization error is small on $\phi, \phi _2$, and $\phi _3$ due to their exponential decay.
%
%
\subsection{Extension to M-atomic molecules} We can extend the result of Theorem~\ref{thm_pre_02}, stated for a diatomic molecule, to larger systems.\par 
We will assume that the distances between atoms are simultaneously scaled by a parameter $d>0$. For all $1\leq k<l\leq M$, we write $X_k-X_l=:dD_{k,l}$, where vectors $D_{k,l}$ are assumed to be fixed. The scaling parameter $d$ will tend to infinity.
The operator $H$ can be written as
\begin{equation}
H=\sum_{i=1}^N\left(T_i-\sum_{k=1}^M\frac{e^2Z_k}{|x_i-X_k|}\right)+\sum_{1\leq i<j\leq N} \frac{e^2}{|x_i-x_j|}+\sum_{1\leq k<l\leq M}\frac{e^2Z_kZ_l}{d|D_{k,l}|}.
\end{equation}
We let 
\begin{equation}
E^\alpha_d:=\inf \sigma(H^\alpha)
\end{equation}
denote the infimum of the spectrum of $H$ restricted to the space $\Hfa=P^\alpha \Hf$.\par 
Consider the cluster decomposition $\beta_M:=(\cc_1,\cdots,\cc_M)$ of the original system into $M$ clusters such that $\bigcup_{k=1}^M \cc_k=\{1,\cdots,N\}$ and $ \cc_k \cap \cc_l = \emptyset$ for all $k\neq l$. We define the set $\mathcal{D}^M_N$ as the collection of all such decompositions. Let 
\begin{equation}\label{eq_mat_02}
\tilde H_{\beta_M}:= \sum_{k=1}^M \tilde H_{\cc_k}^{Z_k}
\end{equation}
where $\tilde H_{\cc_k}^{Z_k}$ is defined according to \eqref{eq_int_02}, acting on the space $\ltn$. The symmetry group of this Hamiltonian is $S_{\beta_M}:=S(\cc_1)\times \cdots \times S(\cc_M)\subset S_N$, the group of permutations which leave the cluster decomposition $\beta_M$ intact. Once again, the irreducible representations of $S_{\beta_M}$ can be expressed as direct products of irreducible representations $\alpha'_\cc$ of $S(\cc)$. In particular, for any irreducible representation $\alpha'_{\beta_M}\prec\alpha$ of $S_{\beta_M}$ there exists a unique $M$-tuple of irreducible representations $\alpha'_{\cc_k}\prec\alpha$ such that
\begin{equation}
\bigotimes_{k=1}^M \alpha'_{\cc_k}\cong \alpha'_{\beta_M}.
\end{equation}
We take $P^{\alpha'_{\beta_M}}$ to be the projection in $\Hf$ onto functions belonging to the irreducible representation $\alpha'_{\beta_M}$. Letting $\tilde H_{\beta_M}^{\alpha'_{\beta_M}}:=\tilde H_{\beta_M}P^{\alpha'_{\beta_M}}$ we define
\begin{equation}
\mu^\alpha_{\beta_M}:= \min _{\alpha'_{\beta_M}\prec\alpha}\inf \sigma\big(\tilde H_{\beta_M}^{\alpha'_{\beta_M}}\big)
\end{equation}
and
\begin{equation}\label{eq_int_17}
\mu^\alpha_M:= \min_{\beta_M\in \mathcal{D}_N^M} \mu^\alpha_{\beta_M}.
\end{equation}
Similar to the diatomic case $\mu^\alpha_M=\lim_{d\rightarrow \infty} E^\alpha_d$. We define the functions $\fkl_2,$ $\fkl_3\in \ltn$ as
\begin{equation}\label{eq_int_13}
\fkl_2(x):=\sum_{\substack{i\in \cc_k\\ j\in \cc_l}} -e^2\big(3(x_i \cdot e_{D_{k,l}})(x_j\cdot e_{D_{k,l}})-x_i\cdot x_j\big),
\end{equation}
\begin{equation}\label{eq_int_14}
\begin{split}
 f_3^{(k,l)}(x) :=\sum_{\substack{i\in \cc_k\\ j\in \cc_l}} & \frac{e^2}{2} \Big(  3(x_i-x_j)\cdot e_{D_{k,l}}\big[2(x_i\cdot x_j)-5(x_i\cdot e_{D_{k,l}})(x_j\cdot e_{D_{k,l}})\big]\\
 & \qquad+ 3|x_i|^2(x_j\cdot e_{D_{k,l}})-3|x_j|^2(x_i\cdot e_{D_{k,l}})\Big),
\end{split}
\end{equation}
where $e_{D_{k,l}}:=\frac{D_{k,l}}{|D_{k,l}|}$ is the unit vector in the direction from nucleus $k$ to nucleus $l$. The functions $\fkl_2$ and $\fkl_3$ are related to the second- and third-order coefficients in the Taylor expansion of the intercluster interaction of cluster $k$ with cluster $l$, see Appendix \ref{sec_app_01} for details.
\par 
The value $\mu^\alpha_M$ defined in \eqref{eq_int_17} is a discrete eigenvalue of the operator $\tilde H^\alpha_{\beta_M}$, see Theorem~\ref{thm_app_01}. Denote by $\tilde {\mathcal{W}}^\alpha_{\beta_M}\subset \Hfa$ the eigenspace of $\tilde H^\alpha_{\beta_M}$ corresponding to $\mu^\alpha_M$ and let
\begin{equation}\label{eq_int_15}
a_1^M:=\max_{\substack{\phi\in \tilde {\mathcal{W}}_{\beta_M}^\alpha\\
\Vert \phi\Vert=1}}\Vert (\tilde H_{\beta_M}-\mu_M ^\alpha)^{-\frac{1}{2}}\hspace{-4mm}\sum_{1\leq k<l\leq M} |D_{k,l}|^{-3}\fkl_2\phi\Vert^2.
\end{equation}
We define $\tilde {\mathcal{V}}_ {\beta_M}^\alpha\subset \tilde{\mathcal{W}}_{\beta_M}^\alpha$ the subspace of all $\phi$ such that $$\Vert (\tilde H_{\beta_M}-\mu_M^\alpha)^{-\frac{1}{2}}\hspace{-4mm}\sum_{1\leq k<l\leq M}|D_{k,l}|^{-3}\fkl_2\phi\Vert^2=a_1^M$$ and
\begin{equation}\label{eq_int_16}
a_2^M:=\max _{\substack{\phi\in \tilde{ \mathcal{V}}^\alpha_{\beta_M}\\\Vert \phi\Vert=1}}\Vert (\tilde H_{\beta_M}-\mu_M^\alpha)^{-\frac{1}{2}}\hspace{-4mm}\sum_{1\leq k<l\leq M}|D_{k,l}|^{-4}\fkl_3 \phi \Vert^2.
\end{equation}
Slightly abusing notation, for $\beta\in \mathcal D_N^M$ we write $\beta\in\Da$ iff for all $k\in \{1,\cdots,M\}$ one hase $\sharp\cc_k=Z_k$.
\begin{theorem}[The Axilrod--Teller--Muto three body correction to the van der Waals--London interaction]\label{thm_pre_03}
Assume that $Z_le^2\leq \frac{2}{\pi},$ for all nuclear charges when the kinetic energy of the electrons is taken to be pseudo--relativistic. 
Let $\alpha$ be an irreducible representation of $S_N$ corresponding to a Young pattern with at most two columns and let the following conditions hold:\\
1') For all $\beta\in \mathcal D_N^M \setminus \Da$ 
\begin{equation}
\mu_{\beta_M}^\alpha> \mu^\alpha_M.
\end{equation}
2') For each induced irredducible representation $\alpha^*_{\beta_M}\prec \alpha$ of the group $S_N$ such that $P^{\alpha^*_{\beta_M}} \tilde{\mathcal{W}}_{\beta_M}^\alpha\neq \emptyset$,
\begin{equation}
\dim (P^{\alpha^*_{\beta_M}} \tilde{\mathcal{W}}_{\beta_M}^\alpha)=\dim \alpha^*_{\beta_M}.
\end{equation}
Then 
\begin{equation}\nonumber
\begin{split}
&E^\alpha_d-\mu^\alpha_M=-\frac{a_1^M}{d^6}-\frac{a_2^M}{d^8}+\frac{a_3^M}{d^{9}}+\mathcal{O}(d^{-10}).
\end{split}
\end{equation}
where
\begin{equation}
a_3^M=\sum_{\substack{k\neq l\\l\neq n,n\neq k}}\frac{\las (\tilde H_{\beta_M}-\mu_M^\alpha)^{-1}f_2^{(k,l)}\phi,f_2^{(l,n)} (\tilde H_{\beta_M}-\mu_M^\alpha)^{-1}f_2^{(n,k)}\phi\ras}{8|D_{k,l}|^{3}|D_{l,n}|^{3}|D_{n,k}|^{3}}.
\end{equation}
\end{theorem}
\begin{remas}\noindent \begin{itemize}
\item The term of order $d^{-6}$ is a sum of the corresponding terms in Theorem~\ref{thm_pre_02}. Again no term of order $d^{-7}$ appears. The main difference to the diatomic case is the appearance of the term of order $d^{-9}$. This term, 
	a non-additive many body effect, is the 
	 famous Axilrod--Teller--Muto three-body correction, which plays an important role for
	atom physics \cite{Lilienfeld-Tkatchenko, Axilrod-Teller, DiStasio-2014, Muto}. It  
  	stems from an interaction of three atoms, each of the atoms induces dipole momenta in the other two atoms of this triplet. Their interaction is proportional to $d^{-9}$.  
  	To the best of our knowledge our 
  	result is 
  	the first proof of this famous  
  	conjecture in atom physics. 
\item  Recently I.\ Anapolitanos \cite{Anapolitanos} studied the error term in the van der Waals--London estimate and proved that under the same conditions as in Theorem \ref{thm_pre_03} the difference between the van der Waals--London term and the term $a_1^Md^{-6}$ is bounded by 
  \begin{align}
  	c_2 M^2 d^{-7} + c_3 \frac{M^4}{d^9}\left( 1+ N^Ze^{-c_3d} \right)
  \end{align}
  for $d\gtrsim N^{4/3}$, where $M$ is the number of atoms, $Z$ is the maximal charge of the nuclei, $N$ the number of electrons and $c_1,c_2,c_3$ are some non--specified constants. 
  Our Theorem \ref{thm_pre_03} shows that such a bound on the error is far from being optimal. The term of order $d^{-7}$ is, in fact, absent in the expansion, the first correction term should have power $d^{-8}$. Moreover, the term of order $d^{-9}$ is a three--body effect, thus it should grow as $M^3$ and not as $M^4$, since it describes interactions of tripels of atoms whose combinatorial factor is given by $M(M-1)(M-2)$. 
   A term in the expansion with a factor growing like $M^4$ should come with a much higher power than $d^{-9}$.  
\item In the diatomic case the result will not change if we replace the vector $e_D$ in the definition $f_2,f_3$, \eqref{eq:def-fA} and \eqref{eq:def-fB} by an arbitrary normalized vector. In contrast to that, in the multi-atomic case  the term of order $d^{-9}$ depends on the angles between vectors $D_{k,l}, D_{l,n}$ and $D_{n,k}$, which confirms the prediction of Axilrod--Teller and Muto. .
\end{itemize}
\end{remas}
The paper is organized as follows.  
In Section~\ref{sec_exp} we prove exponential decay of functions $\phi,(\tilde H_\beta -\mu^\alpha)^{-1} f_2 \phi,$ and $(\tilde H_\beta -\mu^\alpha)^{-1} f_3\phi$, which play a crucial role in the proof of Theorems~\ref{thm_pre_02} and \ref{thm_pre_03}.\par 
In Section~\ref{sec_loc} we prove a localization error estimate for the pseudo--relativistic kinetic energy, which shows that outside the region, where the derivative of the cutoff function is non-zero, the localization error is exponentially small.\par 
In Sections~\ref{sec_diat} and \ref{sec_mat} we prove Theorems~\ref{thm_pre_02} and \ref{thm_pre_03} respectively.\par 
In Appendix~\ref{sec_hvz} and \ref{sec_exist} we prove the HVZ theorem for atoms and atomic ions and the existence of a ground state for pseudo--relativistic atoms and positive ions on spaces with fixed permutation symmetry. This result was announced by G. Zhislin in \cite{zhislin2006}. For convenience of the reader we give a complete proof of these statements.\par 
In Appendix~\ref{sec_app_04} and \ref{sec_app_01} we prove several technical estimates, which we use in Sections~\ref{sec_exp} and \ref{sec_diat}, respectively.\par 
Finally, in Appendix~\ref{sec_app2} we prove orthogonality relations, which are due to the symmetry of functions $\phi$ and $I_\beta$.
\section{Exponential decay of eigenfunctions}\label{sec_exp}
In the nonrelativistic case, exponential decay of eigenfunctions with given permutation symmetry is well-known (see e.g. \cite{agmonlecture}). 
The exponential decay of eigenfunctions of a Hamiltonian with pseudo--relativistic kinetic energy proved by Carmona, Masters and Simon in \cite{carmona1990relativistic} does not apply for Coulomb potentials, however. Although being motivated by the question of exponential decay estimates for multi--particle pseudo--relativistic Schr\"odinger with Coulomb interactions, the class of potentials they use, the so--called relativistic Kato--class, does not contain any potential with a Coulomb singularity. 
For pseudo--relativistic kinetic energy and Coulomb potentials, exponential decay of eigenfunctions was shown by Nardini in \cite{Nardini1986} for the two body case. He extended his results to the $N$-body case in \cite{nardini1988asymptotic}. However, in the proof he uses a method which destroys permutational symmetries.
To prove Theorem~\ref{thm_pre_02} we need exponential decay of ground states $\phi\in \tilde{\mathcal{W}}^\alpha_\beta$ of $\tilde H_\beta$ and exponential decay of functions of the form $(\tilde H_\beta-\mu^\alpha)^{-1}f_l\phi$, $l=2,3$ where $\phi\in \tilde{\mathcal{W}}^\alpha_\beta$ is a ground state.
To this end we will apply a modification of Agmon's method (see \cite{agmonlecture}), adapted to the nonlocal pseudo--relativistic kinetic energy, which preserves symmetry. \par 
Let $\alpha_\cc$ be an irreducible representation of $S(\cc)$. We define
\begin{equation}\label{eq:def-Sigma}
\Sigma^{\alpha_\cc}:= \lim_{R\rightarrow\infty} \inf_{\substack{\psi\in P^{\alpha_\cc} H^{1/2}(\R(\cc))\\ \supp(\psi)\cap B_R(0)=\emptyset}}\Vert \psi\Vert^{-2}\las \psi, \tilde H_\cc^Z \psi\ras ,
\end{equation}
where $B_R(0)$ is the ball in $\R(\cc)$ of radius $R$ centered at $0$ and $\tilde H_\cc^Z$ was defined in \eqref{eq_int_02}.
Everywhere in this section we treat the pseudo--relativistic kinetic energy operator $T_i=\sqrt{p_i^2+1}-1$ only.
%
%
\begin{theorem}\label{thm_int_04}
For any fixed $\mu<\Sigma^{\alpha_\cc}$, assume that $\nG\in H^{1/2}(\R(\cc))$ satisfies $P^{\alpha_\cc}\nG = \nG$ 
and $(\tilde H_\cc^Z-\mu)\nG =\ga$, where $\ga$ is a function with $e^{a |\cdot|} \ga\in L^2(\R(\cc))$ for some $a>0$. Then there exists $b>0$ such that 
\begin{equation}
e^{b|\cdot|}\nG\in L^2(\R(\cc)).
\end{equation}
\end{theorem}
\begin{rema}Choosing $\ga=0$ in the above theorem implies that any eigenfunction $\nG$ of $\tilde H^Z_\cc	$ with associated eigenvalue $\mu<\Sigma^{\alpha_\cc}$ is exponentially decaying.  
\end{rema}
In addition to Theorem~\ref{thm_int_04} we will need a similar statement for cluster Hamiltonians $\tilde H_\beta$ corresponding to a cluster decomposition $\beta$ into two clusters. 
%
%
\begin{prop}\label{prop_int_02}
Let $\alpha_\beta$ be an irreducible representation of $S_\beta$ and let
\begin{equation}
\Sigma^{\alpha_\beta}:= \lim_{R\to\infty}\ \inf_{\substack{\varphi\in P^{\alpha_\beta}H^{1/2}(\R^{3N})\\
\mathrm{supp}(\varphi) \cap B_R(0) =\emptyset }} \ \|\varphi\|^{-2}\las \varphi, \tilde H_\beta\varphi \ras ,
\end{equation}
where $B_R(0)$ is the ball in $\R^{3N}$ with radius $R$ centered at $0$. 
For any fixed $\tilde \mu<\Sigma^{\alpha_\beta}$, assume that $\tilde\nG\in H^{1/2}(\R^{3N})$ satisfies $P^{\alpha_\beta}\tilde \nG =\tilde \nG$ 
and $(\tilde H_\beta-\tilde \mu)\tilde \nG =\tilde \ga$, where $\tilde \ga$ is a function with $e^{a |\cdot|} \tilde \ga\in L^2(\R^{3N})$ for some $a>0$. Then there exists $b>0$ such that 
\begin{equation}
e^{b|\cdot|}\tilde \nG\in L^2(\R^{3N}).
\end{equation}
\end{prop}
\begin{proof}
The proof of Proposition~\ref{prop_int_02} follows immediately from Theorem~\ref{thm_int_04}, since the Hamiltonian $\tilde H_\beta$ describes non-interacting clusters, whose  center has been moved to the origin. Thus the total system is a direct sum of these noninteracting systems to each of which Theorem \ref{thm_int_04} applies. 
\end{proof}
%
%
\begin{corr}\label{cor_exp_01}
Let $\beta\in \Da$ and $\tilde {\mathcal{W}}_\beta^\alpha$ be the ground state subspace of the Hamiltonian $\tilde H_\beta^\alpha = \tilde H_\beta P^\alpha$ corresponding to the energy $\mu^\alpha$. Then for any normalized function $\phi\in \tilde {\mathcal{W}}_\beta^\alpha$ the functions $\phi$, $\phi_2=(\tilde H_\beta-\mu^\alpha)^{-1}f_2\phi$, and $\phi_3=(\tilde H_\beta-\mu^\alpha)^{-1}f_3\phi$ with $f_2, f_3$ defined in \eqref{eq:def-fA},\eqref{eq:def-fB} and some $b_1,b_2,b_3>0$ 
\begin{equation}
 e^{b_1|\cdot|}\phi,\ e^{b_2|\cdot|} \phi_2,\ e^{b_3|\cdot|} \phi_3 \in L^2(\R^{3N}).
\end{equation}
\end{corr}
\begin{proof}
Since $\phi$ is a ground state, the existence of a $b_1>0$ such that 
$ e^{b_1|\cdot|}\phi\in L^2(\R^{3N})$ follows immediately 
from Proposition~\ref{prop_int_02}. Notice that for $l=2,3$ we 
have
\begin{equation}
(\tilde H_\beta-\mu^\alpha)\phi_l=(\tilde H_\beta-\mu^\alpha)(\tilde H_\beta-\mu^\alpha)^{-1}f_l\phi=f_l \phi.
\end{equation}
The functions $f_l$ grow at most polynomially in $|x_i|$ which is controlled by the exponential decay of $\phi$ so we can apply Proposition~\ref{prop_int_02} with $\tilde\Gamma = f_l \phi$ to obtain the result.
\end{proof}
%
%
It remains to give the 
\begin{proof}[Proof of Theorem~\ref{thm_int_04}] Let $k:=\sharp\cc$ be the number of electrons in the cluster $\cc$. To simplify the notation assume $\cc=\{1,\cdots,k\}$. Let $\nG\in H^{1/2}(\R(\cc))$ be a solution of the equation $(\tilde H_\cc^Z-\mu)\nG  =\ga $ and  $\xi=\xi_{\nu,\veps,R} \in C^\infty(\R(\cc);\R)$ be a family of functions with the following properties:
\begin{itemize}
\item $\xi$ is bounded

\item $\xi$ is invariant under all permutations of the variables in the cluster $C$ 

\item $\supp (\xi)\cap B_R(0)=\emptyset$ for some large enough $R>0$, which will be chosen later

\item $|\xi|\leq C e^{a|\cdot|}$ for some constant $C>0$ 
\end{itemize}
By the definition of $\Sigma^{\alpha_\cc}$ in \eqref{eq:def-Sigma} and since 
$\mathrm{supp}(\xi)\cap B_R(0)=\emptyset$, there exists a function $\vartheta(R)$ such that $\lim_{R\to\infty}\vartheta(R) = 0$ and 
\begin{equation}\label{eq_exp_05}
\big(\Sigma^{\alpha_\cc}-\mu-\vartheta(R)\big)\Vert \xi \nG\Vert^2\leq \las \xi \nG,(\tilde H_\cc^Z-\mu)\xi\nG\ras.
\end{equation}
Since $(\tilde H_\cc^Z-\mu)\nG=\ga$, we can write
\begin{equation}\label{eq_exp_06}
\begin{split}
\las \xi \nG,(\tilde H_\cc^Z-\mu)\xi\nG\ras&=\re \las \xi \nG,(\tilde H_\cc^Z-\mu)\xi\nG\ras 
		= \re \las \xi^2 \nG, (\tilde H_\cc^Z-\mu)\nG\ras 
			+ L^C_\xi(\nG,\nG)\\
	&=\re \las \xi^2\nG, \Gamma\ras + L^C_\xi(\nG,\nG)
\end{split}
\end{equation}
with $L^C_\xi$ the quadratic form for the commutation error from Lemma \ref{lem_ims-extended-many-body}. 
Clearly $\re \las \xi^2\nG, \Gamma\ras = \re \las \xi\nG, \xi\Gamma\ras \le \|\xi\nG\|\|\xi\Gamma\|$. 
Together with \eqref{eq_exp_05} we get
\begin{equation}\label{eq_exp_02}
\big(\Sigma^{\alpha_\cc}-\mu-\vartheta(R)\big)\Vert \xi \nG\Vert^2\leq \Vert\xi\nG\Vert\Vert\xi\Gamma\Vert+L^C_\xi(\nG,\nG).
\end{equation}
We now specify the choice of $\xi=\xi_{\nu,\veps,R}$: For $\nu\ge 0$ and for $\veps\ge  0$ we set
\begin{equation}\label{eq_exp_01}
G_{\nu,\veps}(r):=\frac{\nu r}{1+\veps r}.
\end{equation}
and 
\begin{equation}\label{eq_F}
\nF_{\nu,\veps}= \sum_{j\in C} G_{\nu,\veps}(|x_j)|.
\end{equation}
Pick $\chi_0\in C^\infty(\R_+;[0,1])$ such that
\begin{equation}\nonumber
\chi_0 (r) := \left\{
\begin{array}{ll}
   1 & \quad\mbox{if}\quad r < 1  \\
   0 & \quad\mbox{if}\quad r > 2
\end{array}
\right.
\end{equation}
and define for $x\in \R^{3k}$ the function $\chi\in C^\infty(\R^{3k};[0,1])$ with
\begin{equation}\label{eq_exp_03}
\chi(x):=1-\prod_{i=1}^k \chi_0(|x_i|).
\end{equation}
For $R>0$ we set   
\begin{equation} 
	\chi_R(x)= \chi\big( x/R \big) \, 
\end{equation}
and 
\begin{equation}\label{eq_exp_04}
\xi=\xi_{\nu,\veps,R}:=\chi_Re^{\nF_{\nu,\veps}}.
\end{equation}
%
%
Lemma \ref{lem_L_xi-bound-many-body} gives a convenient bound for the second term on the right hand side of \eqref{eq_exp_02}. Using \eqref{eq_commutation-error-bound-many-body} in \eqref{eq_exp_02} yields
\begin{equation}\label{eq_exp-punshline-1}
(\Sigma^{\alpha_\cc} - \mu -\vartheta(R))\Vert \xi\nG\Vert^2
 \leq  \Vert\xi\nG\Vert\Vert\xi\Gamma\Vert
	+kC_{\nu}\left(L_\chi/R +\nu\right)^2 \big\|e^{F}\nG\big\|^2
\end{equation}
where $k$ is the number of particles in the cluster $C$ and, for simplicity of notation,  we abbreviated $F=F_{\nu\veps}$,  

Note that $(1-\chi_R)e^F\le e^{2k\nu R}$, hence $\Vert (1-\chi_R)e^F\nG\Vert \le e^{2k\nu R} \Vert\nG\Vert$. Using this and 
$\Vert e^F\nG \Vert\le \Vert \chi_Re^F \nG \Vert + \Vert (1-\chi_R)e^F \nG \Vert$ in \eqref{eq_exp-punshline-1} and rearranging terms, we get 
\begin{equation}\label{eq_exp_12}
\begin{split}
 \big(\Sigma^{\alpha_\cc} - \mu -\vartheta(R) 
 	-  \delta(R,\nu)\big)&\Vert \xi_\veps\nG\Vert^2 
 		- \big(2 e^{2k\nu R} \delta(R,\nu)\Vert\nG\Vert +  \Vert\xi\Gamma\Vert\big) \Vert\xi\nG\Vert
 		 \\
	 	&\leq \delta(R,\nu)e^{4k\nu R} \Vert\nG\Vert^2 .
\end{split} 
\end{equation}
Where we also abbreviated $\delta(R,\nu)= kC_\nu(\nu+L_\chi/R)^2$. 
Since 
\begin{equation*}
	\lim_{R\to\infty}(\vartheta(R) +\delta(R,\nu)) = kC_\nu \nu^2\, ,
\end{equation*} 
we can find, for any $0<\nu\le a$ with $kC_\nu \nu^2< \Sigma^{\alpha_\cc} - \mu$, a radius $R>0$ such that 
\begin{equation*}
	\gamma\coloneqq \Sigma^{\alpha_\cc} - \mu -\vartheta(R)-\delta(R,\nu)>0.
\end{equation*}
With such a choice for $\nu$ and $R$, setting 
$C= \delta(R,\nu)$, 
we get from \eqref{eq_exp_12} 
\begin{equation}\label{eq_exp_09}
 \gamma\Vert \xi_{\nu,\veps,R}\nG\Vert^2 - \big(C\Vert\nG\Vert +\Vert e^{a|\cdot|}\Gamma \Vert\big)\,\Vert \xi_{\nu,\veps,R}\nG\Vert
 	\le
 		Ce^{4k\nu R} \Vert \nG\Vert . 
\end{equation}
since $\xi=\xi_{\nu,\veps,R}\le e^{\nu|\cdot|}\le e^{a|\cdot|}$, which clearly gives 
$\Vert\xi\Gamma\Vert\le \Vert e^{a|\cdot|}\Gamma\Vert$. 
 
Note that the r.h.s\ of \eqref{eq_exp_09} is independent of $\veps$. 
Since $\gamma>0$, the map 
$$0\le s\mapsto  \gamma s^2- \big(C\Vert\nG\Vert +\Vert e^{a|\cdot|}\Gamma \Vert\big)\,s$$ 
is unbounded from above. Furthermore, $\xi_{\nu,\veps,R}$ converges monotonically to $\chi_R e^{\nu |\cdot|}$ as $\veps\rightarrow 0$. 
Thus the monotone convergence theorem and the bound \eqref{eq_exp_09} shows
\begin{equation*}
	\Vert \chi_R e^{\nu|\cdot|} \nG \Vert =\lim_{\veps\to 0} \Vert \xi_{\nu,\veps,R} \nG \Vert <\infty\, .
\end{equation*} 
Since $\chi_R$ equals one outside a ball of radius $2R$, this implies $	\Vert e^{\nu|\cdot|} \nG \Vert<\infty$, which 
completes the proof of Theorem~\ref{thm_int_04}.
\end{proof}
%
%
\section{Localization error estimates}\label{sec_loc}
In the proof of Theorem~\ref{thm_pre_02} and Theorem~\ref{thm_pre_03} we will use a partition of unity of the configuration space. In addition to this, we use a cutoff function in our construction of the trial function which we will use to bound the intercluster energy from above (see the introduction in Section~\ref{sec_int}).
To obtain the required upper bound we need to show that cutting the ground states of the subsystems leads to an exponentially small increase in the expectation value of the intercluster energy. Therefore we need a suitable estimate of the so-called localization error.
Note that in contrast to the nonrelativistic kinetic energy operator, the pesudo-relativistic operator is not local.
Consequently the localization error is non-zero everywhere, including the regions where derivatives of the cutoff functions vanish.
Of course, there exist several estimates for the localization error of the pseudo--relativistic kinetic energy. However none of them are precise enough for the proof of the van der Waals-London law. The main difference between the bound for the localization error given below in Theorem~\ref{thm_loc_01} and most of the previously known results (see for example \cite{Benguria,lenzmann2010,Lewis1997,lieb2001analysis,lieb1988}) is, that the localization error is confined to a region which is close to the support of the derivatives of the cutoff functions with a remainder which \emph{decays exponentially with the distance} to the support of the derivatives of the cutoff functions. 
	A similar bound was given in \cite{SSS}, however, our bound is simpler, with a simpler proof, and more suitable for our application.

Take any Lipschitz continuous cut--off functions $w_0,w_1,w_2$ on $\R^3$ and assume that $\sum_{l=0}^2 w_l(z)^2=1$. 
We will choose them later such that $w_1$, respectively $w_2$, localizes near the nucleus at $X_1$, respectively $X_2$. 
Then 
\begin{align}
	1& = \prod_{j=1}^N \Big( \sum_{l=0}^2 w_l(x_j)^2 \Big)
		= \sum_{l_1,\ldots,l_N=0}^2 \prod_{j=1}^N w_{l_j}(x_j)^2 \nonumber\\
		&= \sum_{(\cc_0,\cc_1,\cc_2)\in \Dn} \Big(\prod_{i\in \cc_1} w_1(x_i)\prod_{j\in \cc_2} w_2(x_j)
\prod_{k\in \cc_0} w_0(x_k)\Big)^2 
\label{eq:partition-1}
\end{align}
gathering the indices with the same $l_j$ 
into clusters 
$\cc_l=\{j=1,\ldots,N: l_j=l\}$, $l=0,1,2$, which form a partition of $\{1,\ldots,N\}$.  
We also denote by $\Dn$ the collection of decompositions $\beta=(\cc_0,\cc_1,\cc_2)$ of 
$\{1,\cdots,N\}$ into three clusters $(\cc_0,\cc_1,\cc_2)$, with $\cc_k\cap \cc_l=\emptyset$ for all $k\neq l$ and $\bigcup_{k=0}^2 \cc_k=\{1,\cdots,N\}$. In this way the cluster $\cc_0$ contains particles far from both nuclei while 
clusters $\cc_1$ and $\cc_2$ contain electrons localized near $X_1$ and $X_2$, respectively. 

 For $x\in \R^{3N}$ and 
 $\beta= (\cc_0,\cc_1,\cc_2)$ we define a family of bounded Lipschitz continuous cutoff functions $J_\beta\in Lip(\R^{3N};[0,1])$ by
\begin{equation}\label{eq_loc_02}
\begin{split}
J_\beta (x):=&\prod_{i\in \cc_1} w_1(x_i)\prod_{j\in \cc_2} w_2 (x_j) \prod_{k\in \cc_0}w_0(x_k).
\end{split}
\end{equation}
Because of \eqref{eq:partition-1} these 
functions form a partition of unity, i.e.,  
for all $x\in \R^{3N}$
\begin{equation}\label{eq_int_03}
\sum_{\beta \in \Dn}J_\beta ^2(x)= 1.
\end{equation}
A convenient choice of cut--off functions $w_l$ 
is as follows: Let $\chi$ be given by 
$\chi(t)= 1$ for $0\le t\le 1$, $\chi(t)= \cos\left( \frac{\pi}{2} (t-1)\right)$ 
for $1\le t\le 2$ and $\chi(t)=0$ for $t\ge 2$. 
This is a Lipschitz continuous function and 
$\sqrt{1-\chi^2}$ is also Lipschitz continuous. 

Given positions $X_1, X_2$ of the two nuclei define for $z\in\R^3$  
\begin{equation}\label{eq_loc_102}
\begin{split}
w_1(z)&:=\chi\left( \frac{|z-X_1|}{R} \right) \\
w_2(z)&:=\chi\left( \frac{|z-X_2|}{R} \right)\\
w_0(z)&:=\sqrt{1- w_1^2(z)-w_2^2(z)}.
\end{split} 
\end{equation}
Note that under the condition $4R\le |X_2-X_1|$, we have $w_1w_2=0$, hence $w_0= \sqrt{1-w_1^2-w_2^2}= \sqrt{1-w_1^2}\sqrt{1-w_2^2}$ is also Lipschitz continuous. 

The localization error for some state $\psi\in H^{1/2} (\R^{3N})$ and the partition of unity defined by the functions $J_\beta\in \lip(\R^{3N};[0,1])$ is given by
\begin{equation}\label{eq_loc_03}
	\mathcal{L}[\psi] := \sum_{\beta\in \Dn} \las J_\beta\psi, H J_\beta \psi\ras -\las \psi , H \psi\ras .
\end{equation}
For $z\in \R^3$ we set
\begin{equation}
\Theta_R(z):= \id_{[R/2, 5R/2]}(|z|)
\end{equation}
and for $x\in \R^{3N}$ we define
\begin{equation}\label{eq_loc_04}
  \Theta_{1,j,R}(x):= \Theta_R(x_j-X_1), \quad 
  \Theta_{2,j,R}(x):= \Theta_R(x_j-X_2).
\end{equation}
and 
\begin{equation}\label{eq_loc_04-2}
  \Theta_{1,R}:= \sum_{j=1}^N\Theta_{1,j,R}, 
  	\quad \Theta_{2,R}:= \sum_{j=1}^N\Theta_{1,j,R}\, ,
  	\quad \Theta_R =   \Theta_{1,R} +   \Theta_{2,R}
\end{equation}
which count the number of electrons in an annulus around 
the nuclei at $X_1$ or $X_2$, at least, when $R\le2|X_2-X_1|/5 $, when there is no overlap of the two annular regions. 
With this, we can formulate our bound on the localization error.  
\begin{theorem}[N electron localization error estimate]\label{thm_loc_01}
There exists $C>0$ such that for any $\psi\in H^{1/2}(\R^{3N})$ we have
\begin{equation}\label{eq_int_05}
|\mathcal{L}[\psi] |\leq \frac{C}{R^2}\left(
		\las\psi, \Theta_{R} \psi\ras   
		+e^{-R/4}\Vert \psi \Vert^2 \right).
\end{equation}
for all $0<R\le|X_2-X_1|/4$, where the constant $C$ depends only on $N$, the number of electrons.   
\end{theorem}
%
%
For the proof of this theorem we need the following result.
\begin{prop}\label{cor_loc_01}
Let $w_0,w_1,w_2$ be as defined in \eqref {eq_loc_102}. Then there exists a constant $C<\infty$, such that for all $0<R\le |X_2-X_1|/4$ and all $h\in H^{1/2}(\R^3)$
\begin{equation}\label{eq_loc_100}
\begin{split}
&|\sum_{l=0}^2 \las w_l h,T_1 w_l h\ras-\las h,T_1 h \ras  |\\
&\leq \frac{C}{R^2}\lk\Big\Vert\big(\Vert\Theta_R(\cdot -X_1)h\Vert^2 + \Vert\Theta_R(\cdot -X_2)h\Vert^2 \big)  +e^{-R/4}\Vert h\Vert^2\rk
\end{split}
\end{equation}
\end{prop}
%
\begin{proof}
Note that $w_1$ and $\sqrt{1-w_1^2}$ are both bounded Lipschitz continuous functions with Lipschitz constants $R^{-1}$. 
Lemma \ref{lem_H-onehalf-invariance} shows that all the terms 
in the l.h.s.\ of \eqref{eq_loc_100} are well--defined. According to Lemma \ref{lem_ims_extended} and choosing $d=R/2$ in Lemma~\ref{lem_loc_01} we have
\begin{equation}\label{eq_loc_21}
\begin{split}
	&\las h,T_1,h \ras =  \las w_1h,T_1w_1h\ras + \Big\las \sqrt{1-w_1^2}h,T_1\sqrt{1-w_1^2}h\Big\ras +\text{error}_1
\end{split}
\end{equation}
with 
\begin{equation}
	|\text{error}_1| 
	\le \frac{C}{R^2}\Big(\Vert\Theta_R(\cdot -X_1)h\Vert^2  +e^{-R/4}\Vert h\Vert^2\Big)\, ,
\end{equation}
with a slight abuse of notation for $\Theta_R$ (compared to Lemma \ref{lem_loc_01}). 

Iterating this for $\widetilde{h}=\sqrt{1-w_1^2}h\in H^{1/2}(\R^3)$ and the cutoff function $w_2$,  we get
\begin{equation}\label{eq_loc_105}
\begin{split}
	&\las \widetilde{h},T_1,\widetilde{h} \ras =  \las w_2\widetilde{h},T_1w_2\widetilde{h}\ras + \Big\las \sqrt{1-w_2^2}\widetilde{h},T_1\sqrt{1-w_2^2}\widetilde{h}\Big\ras +\text{error}_2
\end{split}
\end{equation}
with 
\begin{equation}
 \begin{split}
	|\text{error}_2| 
	&\le \frac{C}{R^2}\Big(\Vert\Theta_R(\cdot -X_2)\widetilde{h}\Vert^2  +e^{-R/4}\Vert \widetilde{h}\Vert^2\Big)  \\
	&\le  \frac{C}{R^2}\Big(\Vert\Theta_R(\cdot -X_2)h\Vert^2  +e^{-R/4}\Vert h\Vert^2\Big) 
\end{split}
\end{equation}
since $|\widetilde{h}|\le |h|$. Moreover, since $\supp(w_1)\cap \supp(w_2)=\emptyset$ we find
$$w_2\sqrt{1-w_1}=w_2,$$
hence $w_2\widetilde{h}= w_2 h$ and 
$$\sqrt{(1-w_2^2)}\widetilde{h}=\sqrt{(1-w_2^2)(1-w_1^2)}h=w_0h.$$
So from  \eqref{eq_loc_21} and \eqref{eq_loc_105} we get 
\begin{align*}
	\big| \las w_1h, T_1 w_1h \ras &+ \las w_2h, T_1 w_2h \ras  + \las w_0h, T_1 w_0h \ras - \las h, T_1 h \ras \big| = \big| \text{error}_1 +\text{error}_2 \big| \\
	&\le \frac{C}{R^2}\Big(\Vert\Theta_R(\cdot -X_1)h\Vert^2 + \Vert\Theta_R(\cdot -X_2)h\Vert^2  +e^{-R/4}\Vert h\Vert^2\Big) 
	\qedhere
\end{align*}
\end{proof}
\begin{rema}
  Without much change in notation, the above proof easily applies to an arbitrary number of nuclei at positions $X_1, \ldots, X_M$ for all  $0<R\le \min_{k\not=l}|X_k-X_l|/4$.  
\end{rema}
\begin{proof}[Proof of Theorem \ref{thm_loc_01}]
The Coulomb potential, as a multiplicative operator, commutes with the functions $J_\beta$. 
The operator $T_m$ only acts in the $m$-th particle, meaning that it commutes with functions $w_l(x_j)$ for $l=0,1,2$ and $j\neq m$ and we have
\begin{equation}\label{eq_summed out loc error}
\mathcal{L}[ \psi]=\sum_{m=1}^N\left( \sum_{l=0}^2\las w_l(x_m)\psi,T_m w_l(x_m)\psi\ras -\las \psi ,T_m \psi\ras \right).
\end{equation}
 Applying Proposition~\ref{cor_loc_01} 
 on the r.h.s.\ of \eqref{eq_summed out loc error} yields the result, since $(\Theta_{k,j,R})^2= \Theta_{k,j,R}$ for $k=1,2$. 
To see \eqref{eq_summed out loc error} note 
\begin{equation*}
\begin{split}
	\mathcal{L}[\psi] 
		&= \sum_{\beta\in \Dn} \las J_\beta\psi, T J_\beta \psi\ras -\las \psi , T \psi\ras 
		= \sum_{m=1}^N\Big(\sum_{\beta\in \Dn} \las J_\beta\psi, T_m J_\beta \psi\ras -\las \psi , T_m \psi\ras \Big) \, .
\end{split}
\end{equation*}
 Given $m\in \{1,\ldots,N\}$ and a cluster decomposition $\beta=(\cc_0,\cc_1,\cc_2)$ 
  let $\widetilde{\cc}_j= \cc_j\setminus\{m\}$, $j=0,1,2$.
  Then  $\widetilde{\beta}=(\widetilde{\cc}_0,\widetilde{\cc}_1,\widetilde{\cc}_2)$ forms a cluster decomposition of $\{1,\ldots,N\}\setminus\{m\}$, i.e., $N-1$ particles.  
  Furthermore, let $l$ be uniquely determided by 
  $\widetilde{\cc}_{l}\neq \cc_{l}$, i.e, the particle $m$ was removed from the cluster 
    $\cc_{l}$, and denote the corresponding cluster decompositions by 
    $\widetilde{\beta}_{l}$. Define $J_{\widetilde{\beta}_l}(\widehat{x}_m)$ for $\widehat{x}_m=(x_1,\ldots,x_{m-1},x_{m+1}, \ldots, x_N)\in R^{3(N-1)}$ similarly as 
    $J_\beta$ in \eqref{eq:partition-1}.
    Then $J_\beta(x)= J_{\widetilde{\beta}_l}(\widehat{x}_m) w_l(x_m)$    
    and since $T_m$ acts only on the $m$-th particle, one has 
    \begin{align*}
    	 \las J_\beta\psi, T_m J_\beta \psi\ras 
    	 &=  \big\las J_{\widetilde{\beta}_l}(\widehat{x}_m)w_l(x_m)\psi, T_m J_{\widetilde{\beta}_l}(\widehat{x}_m)w_l(x_m)\psi\big\ras\\
    	 &= \big\las J_{\widetilde{\beta}_l}(\widehat{x}_m)^2w_l(x_m)\psi, T_m w_l(x_m)\psi\big\ras\, .
    \end{align*}
  Thus 
    \begin{align*}
    	\sum_{\beta\in \Dn} \las J_\beta\psi, T_m J_\beta \psi\ras 
    		&= \sum_{l=0}^2   
    			\Big\las \sum_{\widetilde{\beta}_l} 
				J_{\widetilde{\beta}_l}(\widehat{x}_m)^2w_l(x_m)\psi, T_m w_l(x_m)\psi\Big\ras \\
			&= \sum_{l=0}^2 \las w_l(x_m)\psi, T_m w_l(x_m)\psi\ras 
    \end{align*}
  since, by the same argument as for \eqref{eq_int_03}, we also have 
  $\sum_{\widetilde{\beta}_l} 
    		J_{\widetilde{\beta}_l}(\widehat{x}_m)^2=1$. 
  This implies \eqref{eq_summed out loc error}. 
\end{proof}
\begin{rema}\label{rema:thm_loc_01} 
  With just minor changes in notation, the above proof can be easily adapted to cluster decomposition with an arbitrary number of clusters. In particular, this allows for an arbitrary finite number of nuclei.
\end{rema}

%
%
\section{Diatomic molecules}\label{sec_diat}
\subsection{Lower bound}\label{sec_below}
Let $\psi \in \Hfa$ with $\Vert \psi\Vert=1$ and $a_1, a_2$ defined in \eqref{eq_int_06} and \eqref{eq_int_07}. We have to show that there exists a constant $0<C<\infty$ such that
\begin{equation}
\las \psi, (H-\mu^\alpha)\psi\ras \geq -\frac{a_1}{|D|^6}-\frac{a_2}{|D|^8}-\frac{C}{\D^{10}}.
\end{equation}
We decompose an arbitrary state $\psi\in \Hfa$ with respect to the partition of unity given by $J_\beta$ defined in \eqref{eq_loc_02} according to the cluster decompositions in $\Dn$ to get
\begin{equation}\label{eq_below_100}
\begin{split}
\las \psi , (H-\mu^\alpha)\psi\ras &=\sum_{\beta \in \Dn}\las J_\beta \psi,(H-\mu^\alpha)J_\beta \psi\ras -\calL[\psi]
\end{split}
\end{equation}
where $\loc$ is the localization error defined in \eqref{eq_loc_03}. By Theorem~\ref{thm_loc_01} there exists a constant $0<C<\infty$ such that
\begin{equation}\label{eq_below_42}
-\loc \ge -\frac{C}{R^2}\left(
		\las\psi, \Theta_{R} \psi\ras  
		+e^{-R/4}\Vert \psi \Vert^2 \right).
%
\end{equation}
where $ \Theta_{1,R} $ and $ \Theta_{2,R} $ are defined in \eqref{eq_loc_04-2}. Let
\begin{equation}\label{eq_below_02}
\begin{split}
	L[J_\beta \psi]:=&\las J_\beta \psi,(H-\mu^\alpha)J_\beta \psi\ras 
		- \frac{C}{R^2}\left(
		\las J_\beta\psi,\Theta_{R}J_\beta\psi\ras
		+e^{-R/4}\Vert J_\beta\psi \Vert^2 \right).
%
\end{split}
\end{equation} 
We will choose $R= R_D= |D|^\frac{3}{4}$ with $D=X_2-X_1$, 
so that for all large enough separations $|D|$ of the nuclei we have  $R<|X_2-X_1|/4$ and, in addition, that 
the support of $ \Theta_{R} $, is far from the nuclei  at $X_1$ and $X_2$.  
According to \eqref{eq_int_03} we have
\begin{equation}
\Vert \psi\Vert^2=\sum_{\beta \in\Dn} \Vert J_\beta \psi\Vert^2  
	\quad \text{and }\Vert  \Theta^k_{j,R}\psi\Vert^2=\sum_{\beta \in\Dn} \Vert  \Theta^k_{j,R}J_\beta \psi\Vert^2\, 
	\quad k=1,2 
\end{equation}
and from \eqref {eq_below_100},\eqref{eq_below_42} and \eqref{eq_below_02} we get
\begin{equation}\label{eq_below_101}
\las \psi ,(H-\mu^\alpha)\psi\ras \geq \sum_{\beta \in \Dn} L[J_\beta \psi].
\end{equation}
Slightly abusing notation, we say $\beta=(\emptyset,\cc_1,\cc_2)\in \Da$ if $\sharp \cc_1=Z_1$ and $ \sharp \cc_2=Z_2$. From \eqref{eq_below_101} we have
\begin{equation}\label{eq_below_102}
\sum_{\beta \in \Dn} L[J_\beta \psi]=\sum_{\beta \in \Da} L[J_\beta \psi]+\sum_{\beta \in \Dn\setminus \Da} L[J_\beta \psi].
\end{equation}
We start with estimating the second sum in the r.h.s. of \eqref{eq_below_102}.\par 
For $\beta=(\cc_0,\cc_1,\cc_2)\in \Dn$ we set
\begin{equation}
\begin{split}
I_\beta &:= \sum_{i\in  \cc_0 \cup\cc_1} \frac{-e^2Z_2}{|x_i-X_2|}+\sum_{j\in \cc_0\cup \cc_2} \frac{-e^2Z_1}{|x_j-X_1|}+\sumij \frac{e^2}{|x_i-x_j|}\\
&\quad +\sum_{\substack{k\in \cc_0\\i\in \cc_1\cup \cc_2}}\frac{e^2}{|x_k-x_i|}+\frac{e^2Z_1Z_2}{|X_2-X_1|}
\end{split}
\end{equation}
the sum of Coulomb interactions between particles belonging to different subsystems and let
\begin{equation}
H_\beta:=H-I_\beta.
\end{equation}
Then we can write
\begin{equation}
\las J_\beta \psi,(H-\mu^\alpha)J_\beta \psi\ras = \las J_\beta \psi,( H_\beta -\mu^\alpha) J_\beta \psi\ras + \las J_\beta \psi,I_\beta J_\beta \psi\ras.
\end{equation}
\subsubsection{Non-neutral decompositions}
If $\beta$ is a non-neutral cluster decomposition, i.e. $\beta \in \Dn\setminus\Da$, on the support of the function $J_\beta\psi$, the distances between a particle in subsystem $1$ to a particle in subsystem $2$ grows in $\D$. The same is true for an electron in $\cc_0$ and both of the nuclei.\par
Hence, since the interaction is small when the clusters are far apart, there exists $\veps_{\D} >0$ with $\veps_{\D}\xrightarrow[\D\rightarrow \infty]{} 0$ such that
\begin{equation}\label{eq_below_104}
\las J_\beta \psi,I_\beta J_\beta \psi\ras \geq -\veps_{\D} \color{black}\Vert J_\beta \psi\Vert^2.
\end{equation} 
As the next step, we find that for $\beta \in \Dn\setminus\Da$, for some $\delta>0$ independent of $\psi$ and $\D$  we have
\begin{equation}\label{eq_below_103}
\las J_\beta \psi,(H_\beta-\mu^\alpha)J_\beta \psi\ras \geq \delta \Vert J_\beta \psi\Vert^2.
\end{equation}
For $\cc_0(\beta)=\emptyset$ the inequality \eqref{eq_below_103} follows from Condition $1$) in Theorem~\ref{thm_pre_02}. 
If $\cc_0(\beta)\neq\emptyset$, the inequality follows from the fact that for all irreducible representations of $S_N$, Hamiltonians of neutral atoms have discrete eigenvalues at the bottom of their spectrum, see Theorem~\ref{thm_app_01}. Removing an electron will increase the energy of the system, according to Theorem~\ref{thm_hvz_01}.
Combining \eqref{eq_below_104} and \eqref{eq_below_103} yields
\begin{equation}\label{eq_below_07}
\begin{split}
L[J_\beta \psi]
	&\geq (\delta-\veps_{\D} )\Vert J_\beta \psi\Vert^2 
		- \frac{C}{R^2}\left(
			\las J_\beta\psi, \Theta_R J_\beta\psi\ras 
			+e^{-R/4}\Vert J_\beta\psi \Vert^2 \right)
%
	\geq 0 
\end{split}
\end{equation}
choosing $R=\D^{3/4}$ and $\D$ big enough.
We can now begin to estimate the functionals $L[J_\beta\psi]$ for $\beta\in \Da$.
%
%
\subsubsection{Neutral decompositions}\label{sec_below_02}
 Let $\bmin\in \Da$, which implies $\sharp \cc_1=Z_1$ and $\sharp\cc_2=Z_2$. For this $\bmin$ and $\varphi,\psi\in \Hfa$ recall that the weighted bilinear form was defined as
\begin{equation}
\las \varphi,\psi\ras _1 := \las  \varphi,(\tilde H_\bmin-\mu^\alpha)\psi\ras 
\end{equation}
and the corresponding semi--norm
\begin{equation}\label{eq_below_01}
\Vert \psi \Vert _1^2:=\las \psi, \psi\ras _1
\end{equation}
where $\tilde H_\beta$ was defined in \eqref{def_hbeta}. 
Let $\Wt\subset\Hfa$ be the ground state space of $\tilde H_\bmin^\alpha$ corresponding to $\mu^\alpha$. Note that $\Wt\neq \emptyset$ by Theorem~\ref{thm_app_01}. We project the function $\U\jbs\psi$ onto the space $\Wt$ with respect to the standard $\ltn$-inner product where $\U$ was defined in \eqref{def_ubeta}. For some $\gamA\in \C$ with $|\gamma_1|\leq 1$ and $\phi\in \Wt$ with $\Vert \phi\Vert=1$ we get 
\begin{equation}\label{eq_below_49}
\U\jbs\psi =\gamA \phi +G.
\end{equation}
As the next step we project $G$ in the sense of the bilinear form $\las \cdot , \cdot \ras _1$  consecutively onto the functions
\begin{equation}\label{eq_below_58}
\phiB:= \resbt  f_2\phi
\end{equation}
and
\begin{equation}\label{eq_below_59}
\phiC:=\resbt  f_3 \phi\, ,
\end{equation}
where $f_2$ is defined in \eqref{eq:def-fA} and $f_3$ in \eqref{eq:def-fB}. 
We will prove in Lemma~\ref{lem_app_03A} that the function $\phi$, because of its rotational symmetry, is orthogonal to $ f_2 \phi$ and $ f_3 \phi$ with respect to the standard $L^2$-inner product, which ensures that the functions $\phi_2$ and $\phi_3$ are well defined. 
Furthermore we show in Corollary~\ref{cor_app_02} that $\phi$, $\phi_2$, $\phi_3$ are mututally orthogonal with respect to the bilinear form $\las \cdot,\cdot\ras_1$.
After this decomposition we have
\begin{equation}\label{eq_below_50}
\jbs\psi=\U^*\big(\gamA \phiA+\D^{-3}\gamB\phiB + \D^{-4}\gamC \phiC +g\big),
\end{equation}
where
\begin{equation}
\las \phi,g\ras=\las g,\phiB\ras_1=\las g,\phiC\ras_1=0.
\end{equation}
By definition of the functions $\phi,\phiB, \phiC$, and $g$ and their orthogonality with respect to $\las \cdot,\cdot\ras_1$ we have
\begin{equation}\label{eq_below_51}
\begin{split}
\las \jbs \psi ,( H_\bmin-\mu^\alpha)\jbs \psi\ras&= \las \U J_\beta \psi,(\U H_\beta \U^*-\mu^\alpha) \U J_\beta \psi\ras\\
&=\las \U J_\beta \psi,(\tilde H_\beta -\mu^\alpha) \U J_\beta \psi\ras\\ &=\frac{|\gamB|^2}{\D^6}\Vert \phiB\Vert_1^2 + \frac{|\gamC|^2}{\D^8}\Vert\phiC\Vert_1^2 +\Vert g\Vert_1^2.
\end{split}
\end{equation}
Now we turn to the term with the intercluster interaction $I_\beta$. In Lemma~\ref{lem_app_04} we prove that for any $\delta>0$ there exist $C>0$  such that for $\D$ sufficiently big
\begin{equation}\label{eq_below_15}
\begin{split}
\las \jbs\psi,I_\beta \jbs\psi\ras&\geq  2|D|^{-6}\re \gamA\overline{\gamB}\Vert \phiB\Vert_1^2+2|D|^{-8}\re \gamA \overline{\gamC}\Vert \phiC\Vert^2_1  \\
&\quad -C\frac{|\gamA|^2+|\gamB|^2+|\gamC|^2}{|D|^{10}} -\delta\Vert g \Vert^2.
\end{split}
\end{equation}
Summing \eqref{eq_below_51} and \eqref{eq_below_15} we arrive at
\begin{equation}\label{eq_below_57}
\begin{split}
\las \jbs\psi , (H-\mu^\alpha)\jbs\psi\ras&\geq  \frac{|\gamB|^2+2\re \gamA\overline{\gamB}}{\D^6}\Vert \phiB\Vert_1^2 +\frac{|\gamC|^2+2\re \gamA\overline{\gamC}}{\D^8}\Vert \phiC\Vert^2_1 \\
&\quad  -C\frac{|\gamA|^2+|\gamB|^2+|\gamC|^2}{\D^{10}}-\delta\Vert g\Vert^2+\Vert g\Vert _1^2.
\end{split}
\end{equation}
Let $\kappa$ be the distance between ground state energy and the next higher eigenvalue of $\tilde H_\bmin$. By Theorem~\ref{thm_app_01} we have $\kappa>0$ and, since $g$ is orthogonal to $\Wt$, also $
\Vert g\Vert_1^2=\las g, (\tilde H_\bmin-\mu^\alpha) g\ras \geq \kappa \Vert g\Vert^2$. Taking $\delta <\frac{\kappa}{2}$ we get
\begin{equation}\label{eq_below_54}
\Vert g\Vert_1^2-\delta \Vert g \Vert ^2\geq \frac{\kappa}{2}\Vert g\Vert ^2.
\end{equation}
Note that
\begin{equation}
|\gamB|^2+ 2\re \gamA \overline{\gamB}=|\gamA+\gamB|^2-|\gamA|^2
\end{equation}
and
\begin{equation}
|\gamC|^2+ 2\re \gamA \overline{\gamC}=|\gamA+\gamC|^2-|\gamA|^2.
\end{equation}
Summing the bound for $\las J_\bmin \psi,( H-\mu^\alpha)J_\bmin \psi\ras$ yields
\begin{equation}\label{eq_below_29}
\begin{split}
\las J_\bmin \psi, (H-\mu^\alpha) J_\bmin\psi\ras &\geq \frac{-|\gamA|^2+|\gamA+\gamB|^2}{|D|^6}\Vert \phiB \Vert_1^2+\frac{-|\gamA|^2+|\gamA+\gamC|^2}{|D|^8}\Vert \phiC \Vert_1^2\\
& \qquad  -\frac{C(|\gamA|^2+|\gamB|^2+|\gamC|^2)}{|D|^{10}}+\frac{\kappa}{2}\Vert g\Vert ^2.
\end{split}
\end{equation}
We now minimize the expression on the r.h.s. of \eqref{eq_below_29} with respect to $\gamma_2$ and $\gamma_3$. We aim to show that for $\D$ large enough, minimization in $\gamB$ yields 
\begin{equation}
\frac{|\gamA+\gamB|^2}{\D^6}\Vert \phi_2\Vert^2_1 - C \frac{|\gamB|^2}{\D^{10}}\geq - \frac{4 C|\gamA|^2}{\D^{10}}.
\end{equation}
Assume that $|\gamB|> 2|\gamA|$, then
\begin{equation}\label{eq_below_09}
\frac{|\gamA+\gamB|^2}{\D^6}\Vert\phiB\Vert^2_1 -\frac{C|\gamB|^2}{\D^{10}}> \frac{\frac14 |\gamB|^2}{\D^6}\Vert \phiB\Vert_1^2 - \frac{C|\gamB|^2}{\D^{10}}
\end{equation}
which is positive for large $\D$.\par 
Whereas for $|\gamB|\leq 2|\gamA|$ we have
\begin{equation}\label{eq_below_11}
\frac{|\gamA+\gamB|^2}{\D^6}\Vert\phiB\Vert^2_1 -\frac{C|\gamB|^2}{\D^{10}} \geq -\frac{4C |\gamA|^2}{\D^{10}}
\end{equation}
which is obviously smaller than the expression on the r.h.s. of \eqref{eq_below_09}. Minimizing similarly in $\gamC$, for $\D$ large enough we get 
\begin{equation}\label{eq_below_12}
\frac{|\gamA+\gamC|^2}{\D^8}\Vert \phiC\Vert^2_1 - \frac{C|\gamC|^2}{\D^{10}}\geq - \frac{4C|\gamA|^2}{\D^{10}}.
\end{equation}
Plugging \eqref{eq_below_11} and \eqref{eq_below_12} into \eqref{eq_below_29}, taking into account that $|\gamA|^2\leq \Vert\jbs \psi\Vert^2$ we arrive at
\begin{equation}\label{eq_below_106}
\las J_\bmin \psi,(H-\mu^\alpha)J_\bmin \psi\ras \geq \left(-\frac{\Vert \phiB\Vert^2_1}{\D^6}-\frac{\Vert \phi_3\Vert^2_1}{\D^8}-C\D^{-10}\right) \Vert \jbs \psi\Vert^2+\frac{\kappa}{2}\Vert g\Vert^2.
\end{equation}
Now we turn to the estimate of the term coming from the localization error, that is,
\begin{equation}\label{eq-loc-error-again}
	\frac{C}{R^2}\left(
			\las J_\beta\psi, \Theta_R J_\beta\psi\ras 
			+e^{-R/4}\Vert J_\beta\psi \Vert^2 \right)
%
\end{equation}
The second term of this expression is exponentially small. For the first term we have
\begin{equation}
\las J_\beta\psi, \Theta_R J_\beta\psi\ras \leq 2\Big \Vert \Theta_R^{1/2}\U^* \Big(\gamma_1\phi +\frac{\gamma_2}{\D^3}\phi_2+\frac{\gamma_3}{\D^4}\phi_3\Big)\Big\Vert^2+2\Vert \Theta_R^{1/2}\U^* g\Vert^2.
\end{equation}
The operator $\Theta_R$ counts the expected number of particles in an annular region close to either of the two nuclei. 
According to Corollary~\ref{cor_exp_01}, $\phiA,\phiB$, and $\phiC$ are exponentially decaying, for normalized $\psi$ we get 
\begin{equation} 
	\las J_\beta\psi, \Theta_R J_\beta\psi\ras 	
		\le  2\Vert \Theta_R \, \U^* g\Vert^2 +\mathcal{O}(e^{-R/2}). 
\end{equation} 
Thus 
\begin{equation}
\begin{split}
  \frac{1}{R^2}\left(
			\las J_\beta\psi, \Theta_R J_\beta\psi\ras 
			+e^{-R/4}\Vert J_\beta\psi \Vert^2 \right) 
  	&\lesssim  \frac{1}{R^2}\left( 
  				\Vert \Theta_R \, \U^* g\Vert^2
				+\mathcal{O}(e^{-R/4})\right) \\
	&\le \frac{1}{R^2}\left( 
  				N\Vert g\Vert^2
				+\mathcal{O}(e^{-R/4})\right) \, . 
\end{split}
\end{equation}
Substituting this into \eqref{eq_below_02}, together with the estimate for $\las J_\bmin \psi,(H-\mu^\alpha)J_\bmin \psi\ras$ in \eqref{eq_below_106} we get
\begin{equation}
L[J_\bmin\psi]\geq \left( -\frac{a_1}{\D^6}-\frac{a_2}{\D^8}-C\D^{-10}\right) \Vert J_\bmin \psi\Vert^2 +\left(-CN R^{-2}+\frac{\kappa}{2} \right)\Vert g \Vert^2 -\mathcal{O}(e^{-R/4}).
\end{equation}
Again choosing $R=\D^{3/4}$ and $\D$ sufficiently large, the second to last term is positive and we arrive at
\begin{equation}\label{eq_below_06}
L[J_\bmin \psi]\geq \left( -\frac{a_1}{\D^6}-\frac{a_2}{\D^8}-C\D^{-10}\right) \Vert J_\bmin \psi\Vert^2 -\mathcal{O}(e^{-\D^{3/2}/4}).
\end{equation}
This inequality is true for any $\beta \in \Da$. Recall from \eqref{eq_below_101} the bound
\begin{equation}
\las \psi,(H-\mu^\alpha)\psi\ras \geq \sum_{\beta \in \Dn}L[J_\beta \psi].
\end{equation}
By \eqref{eq_below_07} for all $\beta \in \Dn\setminus\Da$
\begin{equation}
L[J_\beta \psi]\geq 0.
\end{equation}
Since the number of cluster decompositions $\beta\in \Da$ is finite and $\sum_{\beta\in \Da} \Vert J_\beta \psi\Vert^2 \leq \Vert \psi\Vert^2 = 1$, gathering \eqref{eq_below_101}, \eqref{eq_below_07}, and  \eqref{eq_below_06} we obtain
\begin{equation}
\las \psi ,(H-\mu^\alpha)\psi\ras \geq -\frac{a_1}{\D^6}-\frac{a_2}{\D^8}-\frac{C}{\D^{10}}.
\end{equation}
for some constant $C<\infty$ and all large enough 
$\D$. 
%
%
\subsection{Upper bound}\label{sec_abo_01}
We aim to construct a trial function $\psi_0\in \Hfa$ with $\Vert \psi_0\Vert=1$ such that
\begin{equation}
\las \psi_0,(H-\mu^\alpha)\psi_0\ras \leq  -\frac{a_1}{\D^6}-\frac{a_2}{\D^8}-\frac{C}{\D^{10}}
\end{equation}
where $a_1$ and $a_2$ are defined in \eqref{eq_int_06} and \eqref{eq_int_07}. 

Now we \emph{fix some neutral} cluster decomposition $\bmin \in \Da$ and denote by $\Wt$ the ground state space of $\tilde H_\bmin^\alpha$. Using the permutation symmetry of $\tilde H_\bmin$ we have
\begin{equation}
\Wt=\bigoplus_{\ia(\bmin)\prec\alpha}P^{\ia(\bmin)} \Wt,
\end{equation}
thus there is at least one $\ias(\bmin)\prec \alpha$ such that there exists $\phi \in P^{\ias(\bmin)} \Wt$ that realises the maxima $a_1$ and $a_2$ with $\Vert \phi\Vert=1$. For such a $\phi\in P^{\ias(\bmin)} \Wt$ we set
\begin{equation}\label{def-tpsi_0}
\begin{split}
\tpsi_0:&=\U^* \Big(\phi - \frac{(\tilde H_\bmin-\mu^\alpha)^{-1}  f_2\phi}{\D^3}-\frac{(\tilde H_\bmin-\mu^\alpha)^{-1}  f_3\phi}{\D^4}\Big)\\
&=\U^* \big(\phiA-\frac{\phiB}{\D^3}-\frac{\phiC}{\D^4}\big),
\end{split}
\end{equation}
by definition of $\phi_2$ and $\phi_3$ in \eqref{eq_below_58} and \eqref{eq_below_59}.
With $P^\alpha$ being the projection onto $\Hfa$ and  the cutoff function $J_\bmin$ defined in \eqref{eq_loc_02}, we define the trial state as
\begin{equation}
\psi_0:=\frac{P^\alpha J_\bmin \tpsi_0}{\Vert P^\alpha J_\bmin\tpsi_0 \Vert}.
\end{equation}
\subsubsection{ }
As a first step, we will show
\begin{equation}\label{eq_abo_12}
\las \psi_0, (H-\mu^\alpha)\psi_0\ras = \frac{\las J_\bmin \tpsi_0,(H-\mu^\alpha) J_\bmin \tpsi_0\ras}{\Vert J_\bmin \tpsi_0\Vert ^2}.
\end{equation}
Let $\chi_{\pi^{-1}}^\alpha$ denote the character of the element $\pi^{-1}\in S_N$ in the representation $\alpha$. For $\T_\pi$ defined in \eqref{eq_int_11}, and $|\alpha|$ denoting the dimension of the irreducible representation $\alpha$, by \cite[p. 113]{hamermesh1962} the projection operator onto $\Hfa$ is given by
\begin{equation}\label{eq_abo_08}
P^\alpha= \frac{|\alpha|}{N!}\sum_{\pi\in S_N}\chi^\alpha_{\pi^{-1}}\T_{\pi}.
\end{equation}
Following \cite{SigalovZhislin} we write the r.h.s.\ of \eqref{eq_abo_08} as two sums. In the first sum we collect the permutations which only permute particles within the subsystems of $\bmin$. The second sum contains permutation which change at least one pair of particles belonging to different subsystems of $\bmin$. We get
\begin{equation}\label{eq_abo_09}
P^\alpha=\frac{|\alpha|}{N!}\sum_{\pi\in S_\bmin}\chi^\alpha_{\pi^{-1}}\T_{\pi}+\frac{|\alpha|}{N!}\sum_{\pi\in S_N\setminus S_\bmin}\chi^\alpha_{\pi^{-1}}\T_{\pi}.
\end{equation}
For $\ia(\bmin)\prec \alpha$ we set
\begin{equation}\label{eq_abo_14}
\theta_{\ia(\bmin)}:=\frac{Z_1! Z_2! }{N!}\frac{ |\alpha|}{|\ia(\beta)|}.
\end{equation}
Note that for $\pi \in S_\beta$
\begin{equation}
\chi_{\pi}^\alpha= \sum_{\ia (\beta)\prec \alpha}\chi_{\pi}^{\ia (\beta)}
\end{equation}
and
\begin{equation}
P^{\ia (\beta)} = \frac{|\ia (\beta)|}{Z_1!Z_2!}\sum_{\pi \in S_\beta} \chi_{\pi^{-1}}^{\ia (\beta)} \T_\pi.
\end{equation}
Let us define
\begin{equation}
P_1^\alpha:=\sum_{{\ia(\bmin)}\prec\alpha}\theta_{\ia(\bmin)} P^{\ia(\bmin)}
\end{equation}
and
\begin{equation}
P_2^\alpha:=\frac{|\alpha|}{N!}\sum_{\pi\in S_N\setminus S_\bmin}\chi^\alpha_{\pi^{-1}}\T_{\pi}.
\end{equation}
Then following \cite{SigalovZhislin} we rewrite \eqref{eq_abo_09} as
\begin{equation}\label{eq_abo_01}
P^\alpha=P_1^\alpha+P_2^\alpha.
\end{equation}
To prove \eqref{eq_abo_12} we first compute 
$\Vert P^\alpha J_\bmin \tpsi_0 \Vert^2$. Since $(P^\alpha)^2=P^\alpha$, by \eqref{eq_abo_01} we have
\begin{equation}\label{eq_abo_10}
\begin{split}
&\Vert P^\alpha J_\bmin \tpsi_0 \Vert^2=\big\las (P^\alpha_1+P^\alpha_2) J_\bmin \tpsi_0, J_\bmin \tpsi_0\big\ras\\
&=\sum_{{\ia(\bmin)}\prec\alpha}\theta_{\ia(\bmin)}\las  P^{\ia(\bmin)} J_\bmin \tpsi_0, J_\bmin \tpsi_0\ras+ \frac{|\alpha|}{N!}\sum_{\pi\in S_N\setminus S_\bmin}\hspace{-3mm}\chi^\alpha_{\pi^{-1}}\las \T_{\pi}J_\bmin \tpsi_0, J_\bmin \tpsi_0\ras.
\end{split}
\end{equation}
The function $J_\bmin$ is invariant under permutations in $S_\bmin$, thus $J_\bmin \tpsi_0$ belongs to the same symmetry type $\ias(\bmin)$ as the function $\phi$. The projectors $P^{\ia (\bmin)}$ are mutually orthogonal for different $\ia(\bmin)$. Hence for the first term on the r.h.s. of \eqref{eq_abo_10} we get
\begin{equation}
\sum_{{\ia(\bmin)}\prec\alpha}\theta_{\ia(\bmin)}\las  P^{\ia(\bmin)} J_\bmin \tpsi_0, J_\bmin \tpsi_0\ras=\theta_{\ias(\bmin)}\Vert J_\bmin \tpsi_0\Vert^2.
\end{equation}
The last sum on the r.h.s. of \eqref{eq_abo_10} is zero, as the functions $ \T_\pi J_\beta \tpsi_0$ and $J_\beta \tpsi_0$ are supported on different domains (for details see Appendix~\ref{sec_sym}). Thus
\begin{equation}\label{eq_abo_11}
\Vert P^\alpha J_\bmin \tpsi_0 \Vert^2 =\theta_{\ias(\bmin)}\Vert J_\bmin \tpsi_0\Vert^2.
\end{equation}
Note that \eqref{eq_abo_14} implies $\theta_{\ias(\bmin)}\neq 0$, which yields, in particular, $P^\alpha J_\beta \tpsi_0\neq 0$.\par 
As the next step we would like to show
\begin{equation}
\las P^\alpha J_\beta \tpsi_0,(H-\mu^\alpha) J_\beta \tpsi_0\ras = \theta_{\ias (\beta)} \las J_\beta \tpsi_0,(H-\mu^\alpha) J_\beta \tpsi_0\ras .
\end{equation}
To this end we split $P^\alpha$ as in \eqref{eq_abo_01} and get
\begin{equation}\label{eq_abo_100}
\begin{split}
&\las P^\alpha J_\bmin \tpsi_0,(H-\mu^\alpha)J_\bmin \tpsi_0 \ras \\
&=\las P_1^\alpha J_\bmin \tpsi_0 , (H-\mu^\alpha)J_\bmin \tpsi_0\ras +\las P_2^\alpha J_\bmin \tpsi_0, (H-\mu^\alpha) J_\bmin \tpsi_0\ras.
\end{split}
\end{equation}
Let us show that the second term on the r.h.s. of \eqref{eq_abo_100} is zero. Since for all $\pi\in S_\bmin$, $\T_\pi J_\bmin \tpsi_0$ and $I_\bmin J_\bmin \tpsi_0$ have disjoint support
\begin{equation}\label{eq_abo_15}
\las P_2^\alpha J_\bmin \tpsi_0, I_\bmin J_\bmin \tpsi_0\ras=0.
\end{equation}
Furthermore $H_\beta$ is the sum of two operators
\begin{equation}
H_\beta= \U^* \tilde H_{\cc_1}^{Z_1} \U +\U^* \tilde H_{\cc_2}^{Z_2} \U.
\end{equation}
The first operator acts only on particles in $\cc_1$ and the second operator acts only on particles in $\cc_2$. The localization function $J_\beta$ is supported in the region, where particles in $\cc_1$ are located near $X_1$ and particles in $\cc_2$ are near $X_2$ with distances to the corresponding nucleus $X_1$ and $X_2$ much smaller than $\D=|X_1-X_2|$. We can apply Lemma~\ref{lem_sym} to see
\begin{equation}\label{eq_abo_16}
\las P_2^\alpha J_\bmin \tpsi_0, \U^* \tilde  H_{\cc_1}^{Z_1}\U J_\bmin \tpsi_0\ras=0
\end{equation}
and
\begin{equation}\label{eq_abo_17}
\las P_2^\alpha J_\bmin \tpsi_0, \U^* \tilde H_{\cc_2}^{Z_2}\U J_\bmin \tpsi_0\ras=0
\end{equation}
since the respective functions have disjoint support (see Appendix~\ref{sec_sym}). Equalities \eqref{eq_abo_15}, \eqref{eq_abo_16} and \eqref{eq_abo_17} imply
\begin{equation}
\las P_2^\alpha J_\beta \tpsi_0, (H-\mu^\alpha) J_\beta \tpsi_0\ras = 0.
\end{equation}
Now we turn to the first term on the r.h.s. of \eqref{eq_abo_100}. The operators $( H_\bmin-\mu^\alpha)$ and $I_\bmin$ are invariant under permutations in $S_\bmin$, thus $(H-\mu^\alpha)J_\bmin \tpsi_0$ belongs to the representation $\ias (\bmin)$. By orthogonality of functions belonging to different irreducible representations, we get
\begin{equation}
\las P_1^\alpha J_\bmin \tpsi_0 , (H -\mu^\alpha)J_\bmin \tpsi_0\ras =\theta_{\ias(\bmin)} \las J_\bmin \tpsi_0 , (H -\mu^\alpha)J_\bmin \tpsi_0\ras.
\end{equation}
This proves \eqref{eq_abo_12}.
\subsubsection{} Our next goal is to estimate 
\begin{equation}
\las J_\bmin \tpsi_0 , (H-\mu^\alpha) J_\bmin \tpsi_0\ras.
\end{equation}
We substitute $H= H_\beta+I_\beta $ to get
\begin{equation}\label{eq_abo_06}
\las J_\bmin \tpsi_0 , (H-\mu^\alpha) J_\bmin \tpsi_0\ras=\las J_\beta\tpsi_0, ( H_\beta -\mu^\alpha)J_\beta \tpsi_0\ras + \las J_\beta \tpsi_0,I_\beta J_\beta \tpsi_0\ras.
\end{equation}
For the first term on the r.h.s. of \eqref{eq_abo_06} we write 
\begin{equation}\label{eq_abo_07}
\begin{split}
&\las J_\bmin \tpsi_0 , ( H_\beta-\mu^\alpha) J_\bmin \tpsi_0\ras\\
&\quad=  \las \tpsi_0,( H_\beta-\mu^\alpha)\tpsi_0\ras- \las \sqrt{1-J_\bmin^2}\tpsi_0, (H_\beta-\mu^\alpha) \sqrt{1-J_\bmin^2}\tpsi_0\ras  + \mathcal{LE}_\bmin[\tpsi_0],
\end{split}
\end{equation}
where $\mathcal{L}[\tpsi_0]$ is the localization error coming from the partition of unity with cutoff functions $J_\bmin$ and $\sqrt{1-J_\bmin^2}$. Similar to Theorem~\ref{thm_loc_01}, this can be estimated as
\begin{equation}
|\mathcal{L}_\bmin[\tpsi_0]|\leq \frac{C}{R^2}\left(\las\tpsi_0, \Theta_R \tpsi_0 \ras + e^{-R/4}\Vert \tpsi_0\Vert^2\right).
\end{equation}
By Proposition~\ref{prop_int_02} the function 
$\tpsi_0$ decays exponentially. Choosing 
$R=\D^{3/4}$ then implies 
\begin{equation}
|\mathcal{L}_\bmin[\tpsi_0]|=\mathcal{O} (e^{-\D^\frac{1}{2}}).
\end{equation}
for all large enough $\D$. 
The operator $(H_\beta-\mu^\alpha)$ is semi--bounded from below, thus for some constant $C>0$ we get
\begin{equation}
\las \sqrt{1-J_\bmin^2}\tpsi_0, (H_\beta-\mu^\alpha) \sqrt{1-J_\bmin^2}\tpsi_0\ras\geq -C\Vert \sqrt{1-J_\beta^2}\tpsi_0\Vert ^2  \geq -C e^{-\D^{\frac{1}{2}}}
\end{equation}
 taking into account exponential decay of $\tpsi_0$.
This together with \eqref{eq_abo_07} yields
\begin{equation}\label{eq_abo_05}
\las J_\bmin \tpsi_0 , ( H_\beta-\mu^\alpha) J_\bmin \tpsi_0\ras\leq  \las \tpsi_0,(H_\beta-\mu^\alpha) \tpsi_0\ras+Ce^{-\D^\frac{1}{2}}.
\end{equation}
Once again, by exponential decay of $\tpsi_0$
\begin{equation}\label{eq_abo_18}
\Vert J_\bmin \tpsi_0\Vert^2=\Vert \tpsi_0\Vert^2+\mathcal{O}(e^{-\D^\frac{1}{2}})
\end{equation}
and since $\phi$ is orthogonal to $\phi_2$ and 
$\phi_3$, we get from the definition 
\eqref{def-tpsi_0} 
\begin{equation}\label{eq_abo_03} 
\Vert \tpsi_0\Vert^2=1+\mathcal{O}(\D^{-6}).
\end{equation}
Combining \eqref{eq_abo_05}, \eqref{eq_abo_18} and \eqref{eq_abo_03} yields
\begin{equation}\label{eq_abo_13}
\begin{split}
\Vert J_\beta \tpsi_0\Vert^{-2}\las J_\beta \tpsi_0,( H_\beta-\mu^\alpha) J_ \beta\tpsi_0\ras&=  \las \tpsi_0,(H_\beta-\mu^\alpha) \tpsi_0\ras(1+\mathcal{O}(\D^{-6}) \\
&=(\Vert \phiB\Vert^2_1+\Vert \phiC\Vert^2_1) (1+\mathcal{O}(\D^{-6})).
\end{split}
\end{equation}
Applying \eqref{eq_abo_13} in \eqref{eq_abo_12} we get
\begin{equation}
\las \psi_0,(H-\mu^\alpha)\psi_0\ras = \big(\las J_\beta \tpsi_0,I_\beta J_\beta \tpsi_0\ras+\Vert \phiB\Vert^2_1+\Vert \phiC\Vert^2_1\big ) \big (1+\mathcal{O}(\D^{-6})\big ).
\end{equation}
Similar to the estimates done in Lemma~\ref{lem_app_04}, with simplifications coming from the fact that we have $\gamA=1$, $\gamB=\gamC=-1$ and $g=0$, we obtain
\begin{equation}
\las J_\beta \tpsi_0,I_\beta J_\beta \tpsi_0\ras=-2 \Vert \phiB\Vert^2_1-2 \Vert \phiC\Vert^2_1 +\mathcal O (\D^{-10})
\end{equation}
which completes the proof of Theorem~\ref{thm_pre_02}.
%
%
\section{Multi--atomic case }\label{sec_mat}
The proof of Theorem~\ref{thm_pre_03} is very similar to the proof of Theorem~\ref{thm_pre_02}.  
We start with the lower bound. Define cluster decompositions $\beta_M=(\cc_0,\cdots,\cc_M)$ into $M+1$ clusters, such that particles which are far from all nuclei belong to the subsystem $\cc_0$. As the next step we define the cutoff functions $J_{\beta_M}$ corresponding to the cluster decompositions $\beta_M$. The estimate of the localization error is not different from the diatomic case.\par 
Similar to the proof of Theorem~\ref{thm_pre_02} one can show that if $\beta_M$ is not a decomposition into $M$ neutral atoms, for $\psi\in \Hfa$ we have the inequality
$$
\las J_ {\beta_M} \psi,(H-\mu_M^{\alpha}) J_{\beta_M}\psi\ras >0.
$$
Now we turn to the estimate of the quadratic form 
$\las J_{\beta_M} \psi,(H-\mu^\alpha)J_{\beta_M} \psi\ras
$
for decompositions $\beta_M$ corresponding to $M$ neutral atoms (c.f. Section~\ref{sec_below_02}).\par 
We defined $\tilde H_{\beta_M},\tilde{\mathcal W}^\alpha_{\beta_M}$ and functions $f_2^{(k,l)}, f_3^{(k,l)}$ in equations \eqref{eq_mat_02}- \eqref{eq_int_16}. Let $\mathcal U _{\beta_M}$ be the shift operator defined analogous to $\U$ in \eqref{def_ubeta}. Similar to \eqref{eq_below_50} we write
\begin{equation}
J_{\beta_M} \psi =\mathcal U ^* _{\beta_M}\big(\gamA \phi+d^{-3}\gamB \phi_2 + d^{-4}\gamC  \phi_3 + g\big).
\end{equation}
where $\phi\in \mathcal W_{\beta_M}^\alpha$ and the functions $\phi_2,\phi_3$ are given by
\begin{equation}
\phi_2= \resm\sum_{k<l} |D_{k,l}|^{-3}f_2 ^{(k,l)} \phi
\end{equation}
and
\begin{equation}
\phi_3 = \resm \sum_{k<l} |D_{k,l}|^{-4}f_3 ^{(k,l)} \phi.
\end{equation}
Note that by the same reasons as in the diatomic case we have
\begin{equation}
\las \phi,\phiB\ras= \las \phi,\phi_2\ras _1=\las \phi,\phi_3\ras = \las \phi,\phi_3\ras _1=\las \phi_2,\phi_3\ras = \las \phi_2 ,\phi_3\ras _1=0.
\end{equation}
With the above definitions we get the same epression as \eqref{eq_below_51} for the expected value of $(\tilde H _{\beta_M} -\mu^\alpha_{\beta_M})$.\par 
We now estimate the expectation value of the interaction $I_{\beta_M}$ of particles belonging to dofferent clusters $\las J_{\beta_M} \psi,I_{\beta_M} J_{\beta_M} \psi\ras$. Our goal is to generalize the estimate \eqref{eq_below_15}, which is proven in Lemma~\ref{lem_app_04}, to the case of $M$ atoms.
Let $\chi_{\beta_M}$ be the characteristic function of the support of $J_{\beta_M}$ and let 
\begin{equation}
I_{\beta_M}^o:= I_{\beta_M} \chi_{\beta_M} .
\end{equation}
Note that
\begin{equation}\label{eq_mat_01}
\begin{split}
\las J_{\beta_M}\psi, &I_{\beta_M} J_{\beta_M}\psi\ras=\las J_{\beta_M}\psi, I_{\beta_M}^o J_{\beta_M}\psi\ras\\
 &=  |\gamA|^2\las \Um^*\phiA,I_{\beta_M}^o \Um^*\phiA\ras +\frac{2\re \gamA \overline{\gamB}}{d^3}\las \Um^*\phiB,I_{\beta_M}^o \Um^*\phiA\ras\\
 &\quad+\frac{2\re \gamA\overline{\gamC}}{d^4}\las \Um^*\phiC,I_{\beta_M}^o \Um^*\phiA\ras   +\frac{|\gamB|^2}{d^6}\las \Um^*\phiB,I_{\beta_M}^o \Um^*\phiB\ras  \\
&\quad+ \frac{2\re \gamB \overline{\gamC}}{d^7}\las \Um^*\phiC,I_{\beta_M}^o \Um^*\phiB\ras +\frac{|\gamC|^2}{d^8}\las \Um^*\phiC,I_{\beta_M}^o\Um^* \phiC\ras  \\
&\quad +2\re \gamA\las \Um^*g,I_{\beta_M}^o \Um^*\phiA\ras+\frac{2\re \gamB}{d^3}\las\Um^* g,I_{\beta_M}^o \Um^*\phiB\ras\\
&\quad+\frac{2 \re \gamC}{d^4} \las \Um^*g,I_{\beta_M}^o \Um^*\phiC\ras + \las \Um g,I_{\beta_M}^o \Um g\ras\\
&= B^M_1+B^M_2+B^M_3+\las \Um^*g, I_{\beta_M}^o\Um^* g \ras, 
\end{split}
\end{equation}
where $B^M_1$ contains the first three terms of the r.h.s. of \eqref{eq_mat_01}, $B^M_2$ the second triple and $B^M_3$ third triple on the r.h.s. of \eqref{eq_mat_01}. We define analogously to the diatomic case the functions $f_4^{(k,l)},f_5^{(k,l)}$, see Appendix~\ref{sec_app_01}. Let $x\in \R^{3N}$ and $|\cdot|$ denote the standard norm in this space. On the support of $J_{\beta_M}$ we have $|x|< C(D_0 d)^\frac34$ with $D_0=\min_{k,l} |D_{k,l}|$ and some constant $C$. We can expand $I_{\beta_M}^o$ as a Taylor series for large $d$ arriving at
\begin{equation}\label{eq_mmat_04}
\big| I_{\beta_M}^o- \sum_{k\neq l}\frac{\Um^* f_2^{(k,l)}}{2|D_{k,l}|^3d^3}-\sum_{k\neq l}\frac{\Um^* \fkl_3}{2|D_{k,l}|^4d^4}-\sum_{k\neq l}\frac{\Um^* \fkl_4}{2|D_{k,l}|^5d^5}-\sum_{k\neq l}\frac{\Um^* \fkl_5}{2|D_{k,l}|^6d^6}\big|\leq C\frac{|x|^6}{(D_0d)^7}.
\end{equation}
As the first step, we note that for $B_1^M$, similar to Proposition~\ref{prop_below_07} we have,
\begin{equation}\label{eq_mat_03}
B^M_1\geq \frac{2\re \gamA \overline{\gamB}}{d^6}\Vert \phi_2\Vert_1^2 + \frac{2\re \gamA \overline{\gamC}}{d^8}\Vert \phi_3\Vert ^2_1 - C\frac{|\gamA|^2+|\gamB|^2+|\gamC|^2}{d^{10}}.
\end{equation}
To prove \eqref{eq_mat_03} we substitute \eqref{eq_mmat_04} into the expression for $I_{\beta_M}^o$ in $B_1$ and follow the same steps as in the proof of Proposition~\ref{prop_below_07}, replacing orthogonality relations from Lemma~\ref{lem_app_03B} with the following proposition.
\begin{prop}
Let Condition~2') of Theorem~\ref{thm_pre_03} be fulfilled. Then for $n,m=2,3,4,5,\ n\neq m$ and all $k,l=1,\cdots,M$, $k\neq l$ we have
\begin{equation}
\las \phi_m,f_n^{(k,l)} \phi \ras =0
\end{equation}
where 
\begin{equation}
\phi_m=\resm \sum_{k \neq l} \frac {\fkl_m \phi}{2|D_{k,l}|^{m+1}}
\end{equation}
\end{prop}
\begin{proof}
By Condition~2'), the state $\phi$ belongs to the irreducible representation of the $SO(3)$ group corresponding to the degree $\ell=0$. The functions $f_n^{(k,l)}$ belong to the irreducible representation of the $SO(3)$ group corresponding to the degree $\ell=n$, see the proof of Lemma~\ref{lem_app_02}. Consequently, $\phi_m$ and $f_n^{(k,l)}\phi$ are orthogonal as two functions belonging to different irreducible representations of the $SO(3)$ group. 
\end{proof}
For $B_2^M$ we have
\begin{equation}\label{eq_mat_05}
\begin{split}
B_2^M&\geq \sum_{\substack{k\neq l\\l\neq n,n\neq k}}\frac{\las (\tilde H_{\beta_M}-\mu_M^\alpha)^{-1}f_2^{(k,l)}\phi,f_2^{(l,n)} (\tilde H_{\beta_M}-\mu_M^\alpha)^{-1}f_2^{(n,k)}\phi\ras}{8|D_{k,l}|^{3}|D_{l,n}|^{3}|D_{n,k}|^{3}}\\
&\quad -C\frac{|\gamA|^2+|\gamB|^2}{(D_0d)^{10}}.
 \end{split}
\end{equation}
To prove \eqref{eq_mat_05} we proceed similar to the proof of Proposition~\ref{prop_below_08} except the remark after \eqref{eq_app_101}, which says that for $M=2$ we have $\las \phiB,f_2 \phiB\ras=0$. For $M\geq 3$ the argument of Lemma~\ref{lem_app_02} yields
\begin{equation}
\las (\tilde H_{\beta_M}-\mu_M^\alpha)^{-1}f_2^{(m,m')}\phi,f_2^{(k,k')} (\tilde H_{\beta_M}-\mu_M^\alpha)^{-1}f_2^{(l,l')}\phi\ras=0
\end{equation}
only if at least one of the indices $m,m',k,k',l,l'$ appears an even number of times. Consequently the terms with each of the indices $m,m',k,k',l,l'$ coming twice contribute to the estimate of $B_2^M$. \par
The bound for $B_3^M$ is not different from the one given in Proposition~\ref{prop_below_09} for $M=2$.\par 
To get the upper bound, analogous to the diatomic case let $\alpha^*_{\beta_M}\prec\alpha$ such that there is a function $\phi\in P^{\alpha^*_{\beta_M}}\tilde{\mathcal W}^\alpha_{\beta_M} $ with $\Vert\phi\Vert =1$ that realises the maxima $a_1^M$ and $a_2^M$. We set 
\begin{equation}
\tpsi_0:= \Um^*\big(\phi-\resm\sum_{k<l} \frac{\fkl_2}{|D_{k,l}|^3} \phi - \resm \sum_{k<l} \frac{\fkl_3}{|D_{k,l}|^4}\phi\big)
\end{equation}
and take as a trial function
\begin{equation}
\psi_0:= \frac{P^\alpha J_{\beta_M}\tpsi_0}{\Vert P^\alpha J_{\beta_M}\tpsi_0\Vert},
\end{equation}
and follow the same steps as in the proof of Theorem~\ref{thm_pre_02}.
%
%
%
%
\appendix
\section{The HVZ theorem}\label{sec_hvz}
In Appendix \ref{sec_hvz} and \ref{sec_exist} we prove two fundamental facts regarding the spectra of a pseudo--relativistic Hamiltonian of an atom or positive ion, which are of crucial importance for Theorems~\ref{thm_pre_02} and \ref{thm_pre_03}.
\par 
In Appendix \ref{sec_hvz} we prove a HVZ-type theorem, which gives the location of the essential spectrum for an arbitrary type of permutational symmetry. In Appendix \ref{sec_exist} we prove that Hamiltonians of pseudo--relativistic atoms and positive ions for any type of permutational symmetry have discrete eigenvalues at the bottom of the spectrum. Both results were announced earlier without proof by G.~Zhislin in \cite{zhislin2006}. 
For the convenience of the reader, we give complete proofs in these appendices. In the nonrelativistic case both results are well-known. The first one, which is called HVZ theorem (see \cite{reed1978methods}), was first proven without symmetry considerations in 1960 by G.~Zhislin \cite{Zhi60}, and later generalized by Sigalov and Zhislin to the case of subspaces with fixed permutational symmetry \cite{SigalovZhislin}. The second one, which is known as Zhislin's theorem was proven in the same publications \cite{Zhi60,SigalovZhislin}. Nice discussions of Zhislin's method 
 are given in \cite{joeweid1973}, including the extension to particle  
 symmetries, and Chapter XIII.5 of  \cite{reed1978methods}, which discusses also the proofs of van Winter 
and Hunziker and where one can find further references for extensions of the methods, including symmetry considerations,  in the notes to Chapter XIII.5. 

For multiparticle Schrödinger operators with pseudo--relativistic kinetic energy the HVZ-type theorem  was proven earlier in \cite{Lewis1997}, where systems with finite particle masses and fixed total momentum were considered. The result needed for Theorems~\ref{thm_pre_02} and \ref{thm_pre_03} is different from \cite{Lewis1997}, because on one hand we have a particle with infinite mass, the nuclei, which makes the situation easier. On the other hand we need to include the permutational symmetry.
\par
We follow the ideas in the work by Sigalov 
and Zhislin \cite{zhis67}, with necessary modifications related to the fact that the 
pseudo--relativistic kinetic energy operator is non-local, which also 
requires a different estimate of the localization error. Not only for 
convenience of the reader but also because 
some of the necessary modification are not at all obvious, we give complete proofs. 
\par 
For any $k\in \N$ and $Ze^2 < \frac{2}{\pi}$ we set
\begin{equation}
H_k^Z:= \sum_{i=1}^k T_i - \sum_{i=1}^k\frac{e^2Z}{|x_i|}+\sum_{1\leq i<j\leq k}\frac{e^2}{|x_i-x_j|}
\end{equation}
acting on $L^2(\R^{3k})$, where $T_i$ denotes the pseudo--relativistic kinetic energy operator for the $i$-th electron. Let $\alpha_k$ be an irreducible representation of the group of permutations of $k$ electrons $S_k$.
We set
\begin{equation}\label{eq_hvz_13}
\mu^{\alpha_k}:=\inf \sigma(H_k^ZP^{\alpha_k}).
\end{equation}
Denote by $\alpha_{k-1}'\prec\alpha_k$ an irreducible representation of $S_{k-1}$ induced by $\alpha_k$. We define
 \begin{equation}
 \mu^{\alpha_k}_{k-1}:=\min _{\alpha'_{k-1}\prec\alpha_k} \inf  \sigma (H^Z_{k-1}P^{\alpha'_{k-1}}).
\end{equation}
\begin{theorem}\label{thm_hvz_01}
For subcritical nucleus charge $Ze^2<\frac2\pi$ and for any irreducible representation $\alpha_k$ of $S_k$,
\begin{equation}\nonumber
\sigma_{ess}(H_k^ZP^{\alpha_k})=[\mu^{\alpha_k}_{k-1},+\infty).
\end{equation}
\end{theorem}
\begin{proof}The proof is split into two parts.
\setcounter{subsection}{1}
\subsubsection{"Easy part":} Let us first show that 
\begin{equation}\label{eq_hvz_11}
\sigma_{ess}(H_k^ZP^{\alpha_k})\supseteq[ \mu^{\alpha_k}_{k-1},+\infty).
\end{equation}
To do so, for arbitrary $\lambda\geq\mu^{\alpha_k}_{k-1}$, we give the construction of a Weyl sequence $(\psi_m)_{m\in \N}\subset P^{\alpha_k}L^2(\R^{3k})$ with $\Vert \psi_m\Vert=1$, $\psi_m\rightharpoonup 0$ and
$$
\lim_{m\rightarrow\infty} \Vert (H_k^Z-\lambda)\psi_m\Vert=0.
$$
Let $\alpha^*_{k-1}\prec \alpha_k$ be an irreducible representation of $S_{k-1}$ such that
\begin{equation}\nonumber
\inf \sigma(H^Z_{k-1} P^{\alpha^*_{k-1}})=\mu^{\alpha_k}_{k-1}.
\end{equation}
Since $C_0^\infty(\R^{3(k-1)})$ is dense in the domain of $H^Z_{k-1}P^{\alpha^*_{k-1}}$, for any $\veps>0$ there exists a function $ \phi_{\veps}\in P^{\alpha^*_{k-1}}C_0^\infty(\R^{3(k-1)})$ with $\Vert  \phi_{\veps}\Vert=1$ such that
\begin{equation}\label{eq_hvz_02}
\Vert (H^Z_{k-1}-\mu_{k-1}^{\alpha_k}) \phi_{\veps} \Vert^2 < \frac{\veps}{9}.
	\end{equation}
Let $R_{\veps}$ be such that 
\begin{equation}
\supp(\phi_{\veps}) \subset \{ x=(x_1,\cdots ,x_{k-1})\in \R^{3(k-1)}\big| |x_i|\leq R_{\veps},\ i=1,\cdots, k-1\}.
\end{equation}
The spectrum of $T_k$ is the positive real axis and $C_0^\infty(\R^3)$ is dense in the domain of $T_k$. Thus for any $\veps>0$ there exists $f^{(\veps)}\in \C_0^\infty(\R^3)$ with $\Vert f^{(\veps)}\Vert=1$ such that
\begin{equation}\nonumber
\big\Vert \big[T_k-(\lambda-\mu^{\alpha_k}_{k-1})\big]f^{(\veps)}\big\Vert^2\leq \frac{\veps}{9}.
\end{equation}
Let us consider a decreasing sequence $\veps_m\rightarrow 0$ and the functions $\phi_{\veps_m},\   f^{(\veps_m)}$ chosen accordingly as described above. For each of the $\veps_m$ we will pick a vector $A_m\in \R^3$ and define the shifted function $$f_{A_m,\veps_m}(x_k):=f^{(\veps_m)}(x_k+A_m).$$ The sequence of shifts $A_m$ is chosen such that $\supp(f_{A_m,\veps_m})\cap B_{2R_{\veps_m}}=\emptyset$, and such that
$$\supp\big( f_{A_m,\veps_m}\big) \cap \Big(\bigcup _{l=1}^{m-1}\supp\big(f_{A_l,\veps_l}\big)\Big) =\emptyset.$$
Because the kinetic energy operator is translation invariant we get
\begin{equation}\label{eq_hvz_03}
\big\Vert \big[T_k-(\lambda-\mu^{\alpha_k}_{k-1})\big]f_{A_m,\veps_m}\big\Vert^2\leq \frac{\veps_m}{9}.
\end{equation}
We set
\begin{equation}
\varphi_m(x):= \phi_{\veps_m}(x_1,\cdots,x_{k-1}) f_{A_m,\veps_m}(x_k)
\end{equation}
and let
\begin{equation}
\psi_m(x):=P^{\alpha_k} \varphi_m(x).
\end{equation}
Similar to the proof in Section~\ref{sec_abo_01} we have
\begin{equation}
\Vert \psi_m\Vert^2 = \Vert P^{\alpha_k} \varphi_m\Vert^2 = \theta_{\alpha^*_{k-1}}\Vert \phi_{\veps_m} f_{A_m,\veps_m}\Vert^2,
\end{equation}
where $\theta_{\alpha^*_{k-1}}>0$ is a constant depending on $\alpha^*_{k-1}$ and $\alpha_k$ only (see Section~\ref{sec_abo_01}). By choice of $A_m$, the functions $\psi_m$ have disjoint support and thus $\psi_m\rightharpoonup 0$. \par
We will now estimate $\Vert (H_k^Z-\lambda)\psi_m\Vert$. The Hamiltonian $H_k ^Z$ commutes with the projection operator $P^{\alpha_k}$, and since $\Vert P^{\alpha_k}\Vert\leq 1$ we get
\begin{equation}\nonumber
\begin{split}
\Vert (H_k^Z-\lambda)P^{\alpha_k} \varphi_m\Vert^2 &=\Vert P^{\alpha_k }(H_k^Z-\lambda)\varphi_m\Vert^2\leq \Vert(H_k^Z-\lambda)\varphi_m\Vert^2.
\end{split}
\end{equation}
We split $(H_k^Z-\lambda)$ into three parts
$$
(H_k^Z-\lambda)=(H_{k-1}^Z-\mu^{\alpha_k}_{k-1})+\big(T_k-(\lambda-\mu^{\alpha_k}_{k-1})\big)+\big(\sum_{1\leq i<k}\frac{e^2}{|x_i-x_k|}-\frac{e^2Z}{|x_k|}\big).
$$
On the support of $\varphi_m$ we have
\begin{equation}\label{eq_hvz_04}
\left| \sum_{1\leq i <k}\frac{e^2}{|x_i-x_k|}-\frac{e^2Z}{|x_k|}\right|^2\leq \frac{\veps_m}{9}.
\end{equation}
Together with \eqref{eq_hvz_02} and \eqref{eq_hvz_03} this yields
\begin{equation}\label{eq_hvz_08}
\Vert (H_k^Z-\lambda)P^{\alpha_k} \varphi_m\Vert^2\leq \veps_m.
\end{equation}
This shows that $\lambda \in \sigma_{ess}(H_k^ZP^{\alpha_k})$, and since $\lambda \in [\mu^{\alpha_k}_{k-1} , +\infty)$ was chosen arbitrarily this proves the inclusion \eqref{eq_hvz_11}.
\subsubsection{"Hard part":}
We will show that
\begin{equation}\label{eq_hvz_01}
 \sigma_{ess}(H_k^ZP^{\alpha_k})\subseteq [\mu^{\alpha_k}_{k-1},+\infty).
\end{equation}
We prove this inclusion by induction in $k$.
For $k=1$, the hydrogen-like case, this is well-known. We fix an arbitrary $k\leq Z$ and assume that for any $k'<k$ \eqref{eq_hvz_01} is true. 
Take any $\lambda\in \sigma_{ess}(H_k^ZP^{\alpha_k})$ and a corresponding Weyl sequence $(\psi_l)_{l\in \N}\subset P^{\alpha_k} L^2(\R^{3k})$.
Our aim is to show that 
\begin{equation}\nonumber
\lim _{l\rightarrow \infty} \las \psi_l,H_k^Z\psi_l\ras\geq \mu_{k-1}^{\alpha_k}.
\end{equation}
By Weyl's criterion this implies \eqref{eq_hvz_01}.\par 
 Let $u_R\in C^\infty(\R^3;[0,1])$ such that
\begin{equation}
u_R(z):=\left\{
\begin{array}{ll}
   1 & \mbox{ if } |z| \leq R \\
   0 & \mbox{ if } |z| > 2R
\end{array}
\right. 
\end{equation}
and for any $\cc\subseteq\{1,\cdots,k\}$ we define
\begin{equation}
F_\cc(x):=\prod_{i\in \cc}u_{R}(x_i)\prod_{j\notin\cc}\sqrt{1-u_{R}^2(x_j)}.
\end{equation}
With this definition we have
\begin{equation}
\sum_{\cc\subseteq\{1,\cdots,k\}}F_\cc^2\equiv 1.
\end{equation}
Let $\cc^*:=\{1,\cdots,k\}$; observe that $$\supp\big( F_{\cc^*}\big)\subset\bigotimes_{i=1}^kB_{2R}^{(i)}.$$ 
We apply a weakened form of Theorem~\ref{thm_loc_01} to estimate the localization error and get
\begin{equation}\label{eq_hvz_05}
\begin{split}
\las \psi_l,H_k^Z \psi_l\ras &= \las F_{\cc^*} \psi_l,H_k^Z F_{\cc^*}\psi_l\ras + \sum_{\cc\neq \cc^*} \las F_\cc \psi_l,H_k^Z F_\cc \psi_l\ras - \mathcal{LE}\\
&=\las F_{\cc^*} \psi_l,H_k^Z F_{\cc^*}\psi_l\ras + \sum_{\cc\neq\cc^*} \las F_\cc \psi_l,H_k^Z F_\cc \psi_l\ras+\mathcal{O}(R^{-2}).
\end{split}
\end{equation}
For the first term on the r.h.s. of \eqref{eq_hvz_05} the definition of $\mu^{\alpha_k}$, see \eqref{eq_hvz_13}, implies
\begin{equation}\label{eq_hvz_06}
\begin{split}
\las F_{\cc^*}\psi_l,H_k^Z F_{\cc^*} \psi_l\ras &\geq \mu^{\alpha_k} \Vert F_{\cc^*} \psi_l\Vert^2\\
&\quad=\mu^{\alpha_k}_{k-1} \Vert F_{\cc^*}\psi_l\Vert^2 + (\mu^{\alpha_k}-\mu^{\alpha_k}_{k-1})\Vert F_{\cc^*} \psi_l\Vert^2.
\end{split}
\end{equation}
Let
\begin{equation}
H_\cc^Z=\sum_{i\in \cc} T_i -\sum_{i\in \cc}\frac{e^2Z}{|x_i|}+\sum_{\substack{i,j\in \cc\\i<j}}\frac{e^2}{|x_i-x_j|}.
\end{equation}
For each summand of the second term on the r.h.s. of \eqref{eq_hvz_05} we write
\begin{equation}\label{eq_hvz_07}
\begin{split}
\las F_\cc \psi_l,H_k^ZF_\cc \psi_l\ras& = \las F_\cc\psi_l,H_\cc^Z F_\cc \psi_l\ras + \sum_{j\notin \cc}\las F_\cc \psi_l,T_j F_\cc \psi_l\ras\\
&\quad+\sum_{j\notin \cc} \big\las F_\cc \psi_l,\big(-\frac{e^2Z}{|x_j|}+\sum_{i\neq j}\frac{e^2}{2|x_i-x_j|}\big) F_\cc \psi_l\big\ras .
\end{split}
\end{equation}
Each term in the second sum on the r.h.s. of \eqref{eq_hvz_07} is non-negative. For the summands in the third term on the r.h.s. of \eqref{eq_hvz_07}, by construction of $F_\cc$, there exists a constant $C>0$ such that 
\begin{equation}
\sum_{j\notin \cc} \big\las F_\cc \psi_l,\big(-\frac{e^2Z}{|x_j|}+\sum_{i\neq j}\frac{e^2}{2|x_i-x_j|}\big) F_\cc \psi_l\big\ras \geq -\frac{C}{R}\Vert F_\cc \psi_l\Vert^2.
\end{equation}
It is obvious that for any $\cc\subseteq\{1,\cdots,k\}$ the function $F_\cc$ is invariant under permutations in $S(\cc)$. This implies, that for $\psi\in P^{\alpha_k}L^2(\R^{3k})$ the function $F_{\cc}\psi$ necessarily has a symmetry corresponding to an induced representation $\alpha'_\cc\prec \alpha_k$ of $S(\cc)$. Thus for any $\cc\neq \cc^*$ we have
\begin{equation}\label{eq_hvz_12}
\las F_\cc \psi_l,H_\cc^Z F_\cc \psi_l\ras \geq \min_{\alpha'_\cc\prec \alpha_k} \inf \sigma(H_\cc^ZP^{\alpha'_\cc})\Vert F_\cc\psi_l\Vert^2\geq \mu_{k-1}^{\alpha_k} \Vert F_\cc \psi_l \Vert^2
\end{equation}
by the induction assumption, since $H_\cc^ZP^{\alpha'_\cc}$ is unitarily equivalent to $H_{k'}^ZP^{\alpha_{k'}}$ for $k'=\sharp \cc$ and some $\alpha_{k'}\prec\alpha_k$.
Gathering \eqref{eq_hvz_05}, \eqref{eq_hvz_06} and \eqref{eq_hvz_07}-\eqref{eq_hvz_12} we get that for some constant $C>0$ independent of $l\in \N$ we have
\begin{equation}\label{eq_hvz_10}
\las \psi_l,H_k^Z\psi_l\ras\geq \mu_{k-1}^{\alpha_k} \underbrace{\sum_{\cc}\Vert F_\cc \psi_l\Vert^2}_{=1} +(\mu^{\alpha_k}-\mu_{k-1}^{\alpha_k})\Vert F_{\cc^*}\psi_l\Vert^2 -\frac{C}{R}.
\end{equation}
It remains to show that $\Vert F_{\cc^*}\psi_l\Vert^2\xrightarrow{l\rightarrow \infty}0$. The operators $H_0:=\sum_{i=1}^k T_i$ and $H_k^Z$ are semi--bounded from below, thus there exists a constant $c>0$ such that $(H_0+c)$ and $(H_k^Z+c)$ are positive operators. We write
\begin{equation}
F_{\cc^*}\psi_l=F_{\cc^*}(H_k^Z+c)^{-1}(H_k^Z+c)\psi_l.
\end{equation}
Firstly we claim that the sequence $\big((H_k^Z+c)\psi_l\big)_{l\in \N}$ converges weakly to zero. Since $(\psi_l)_{l\in \N}$ is a Weyl sequence, $(H_k^Z-\lambda)\psi_l$ converges to zero in norm and
$$(H_k^Z+c)\psi_l= \underbrace{(H_k^Z-\lambda)\psi_l}_{\rightarrow 0}+\underbrace{(c+\lambda)\psi_l}_{\rightharpoonup 0}.$$
Our next goal is to show that the operator $F_{\cc^*}(H_k^Z+c)^{-1}$ is compact. We write
$$
F_{\cc^*}(H_k^Z+c)^{-1}=F_{\cc^*}(H_0+c)^{-\frac12}(H_0+c)^\frac12 (H_k^Z+c)^{-\frac12}(H_k^Z+c)^{-\frac12}.
$$
Since $(H_k^Z+c)^{-\frac12}$ is the inverse of a strictly positive operator, it is bounded. To obtain a bound of $(H_0+c)^\frac12 (H_k^Z+c)^{-\frac12}$. Let $V$ be the sum of Coulomb potentials in $H_k^Z$, such that
$$
H_k^Z=H_0+V.
$$
Since $V$ is relative $H_0$-bounded, there exist $1>a>0$ and $b>0$ such that for all $\varphi\in \mathcal{D}(H_0)\cap \mathcal{D}(V)$ we have
$$
|\las \varphi,V\varphi\ras| \leq a\las \varphi,H_0\varphi\ras + b\Vert \varphi\Vert^2.
$$
By this inequality, for all $\varphi \in \mathcal{D}(H_0)$ we get
\begin{equation}\nonumber
\begin{split}
\las \varphi,(H_0+c)\varphi\ras &= \las \varphi,(H_0+V+c)\varphi\ras - \las \varphi,V\varphi\ras \\
&\leq \las \varphi,(H_k^Z+c)\varphi\ras +a\las \varphi,H_0\varphi\ras +b\Vert \varphi\Vert^2.
\end{split}
\end{equation}
Since $a<1$, this is equivalent to
\begin{equation}\nonumber
\las \varphi,(H_0+c)\varphi)\ras\leq \frac{1}{1-a}\las \varphi,(H_k^Z+c)\varphi\ras +\frac{b-ac}{1-a}\Vert \varphi\Vert^2.
\end{equation}
In particular, setting $\varphi=(H_k^Z+c)^{-\frac12}\psi$ this yields
\begin{equation}\nonumber
\begin{split}
\Vert (H_0+c)^\frac12(H_k^Z+c)^{-\frac12}\psi\Vert^2&=\las (H_k^Z+c)^{-\frac12}\psi,(H_0+c)(H_k^Z+c)^{-\frac12}\psi\ras\\
&\leq\frac{1}{1-a}\Vert \psi\Vert^2 + \frac{b-ac}{1-a}\Vert (H_k^Z+c)^{-\frac12}\psi\Vert^2.
\end{split}
\end{equation}
Together with boundedness of $(H_k^Z+c)^{-\frac12}$ this implies that $(H_0+c)^\frac12(H_k^Z+c)^{-\frac12}$ is bounded.
%
Finally note that the operator $F_{\cc^*}(H_0+c)^{-\frac12}$ is compact, being a norm limit of Hilbert-Schmidt operators
\begin{equation}
B_n=F_{\cc^*}(H_0+c)^{-1} \chi(H_0<n).
\end{equation}
%
Thus
\begin{equation}
\Vert F_{\cc^*} \psi_l\Vert^2=\Vert F_{\cc^*}(H_0+c)^{-\frac12}(H_0+c)^\frac12 (H_k^Z+c)^{-\frac12}(H_k^Z+c)^{\frac12} \psi_l\Vert^2 \xrightarrow{l\rightarrow \infty} 0.
\end{equation}
Recall from inequality \eqref{eq_hvz_10} that 
\begin{equation}\nonumber
\las \psi_l,H_k^Z\psi_l\ras \geq \mu^{\alpha_k}_{k-1} +(\mu^{\alpha_k}-\mu^{\alpha_k}_{k-1})\Vert F_{\cc^*}\psi_l\Vert^2-\frac{C}{R}.
\end{equation}
Picking $R$ and $l$ large yields $\lambda\geq \mu^{\alpha_k}_{k-1}$,
where $\lambda$ was an arbitrary value in the essential spectrum of $H_k^ZP^{\alpha_k}$.
 \end{proof}
%
\section{Existence of a ground state for atoms and positive ions}\label{sec_exist}
Let $H_k^Z$, $S_k$, and $\alpha_k$ be the same as in Appendix~\ref{sec_hvz} and let $k\leq Z$.
\begin{theorem}\label{thm_app_01}
For any irreducible representation $\alpha_k$ of the group $S_k$, the operator $H_k^Z P^{\alpha_k}$  has a discrete eigenvalue at the bottom of its spectrum.
\end{theorem}
\begin{proof}[Proof of Theorem \ref{thm_app_01}]
We prove the theorem by induction in $k=1,\cdots,Z$. For $k=1$ we have
$$
H_1^Z=\sqrt{p^2+1}-1-\frac{Ze^2}{|x|}\leq \frac{p^2}{2}-\frac{Ze^2}{|x|}.
$$
The operator $\frac{p^2}{2}-\frac{Ze^2}{|x|}$ has an infinite number of negative eigenvalues, which yields the existence of a negative eigenvalue for $H_1^Z$. Note that for one electron we do not have restrictions regarding its symmetry.\par
For fixed but arbitrary $k\leq Z$, let us assume that for each irreducible representation $\alpha_{k-1}$ of the permutation group $S_{k-1}$, the operator  $H_{k-1}^ZP^{\alpha_{k-1}}$ has a ground state.\par 
We will construct a trial state $\psi_0\in P^{\alpha_k}H^{1/2}(\R^{3k})$ for arbitrary irreducible representation $\alpha_k$ of $S_k$ such that
$$
\Vert \psi_0\Vert^{-2}\las \psi_0,H_k^Z\psi_0\ras < \inf \sigma_{ess}(H_k^ZP^{\alpha_k }).
$$
Let $\alpha^*_{k-1}\prec \alpha_k$ be an irreducible representation of $S_{k-1}$ such that
\begin{equation}
\inf \sigma(H_{k-1}^ZP^{\alpha_{k-1}^*})=\min_{\alpha'_{k-1}\prec\alpha_k}\inf \sigma(H_{k-1}^ZP^{\alpha'_{k-1}})=:\mu_{k-1}^{\alpha_k}.
\end{equation}
By the induction assumption, there exists a state $\phi\in P^{\alpha^*_{k-1}}H^{1/2}(\R^{3(k-1)})$ with
\begin{equation}
\las \phi,H_{k-1}^Z \phi\ras = \mu_{k-1}^{\alpha_k}\Vert \phi\Vert^2.
\end{equation}
Let $f\in C^\infty_0(\R^3)$ with $\Vert f\Vert_{L^2}=1$ and $\supp(f)\subset\{x\in \R^3\big| 1\leq |x|\leq 2\}$, and let
\begin{equation}
f_R(z):=R^{-\frac32} f(zR^{-1}),
\end{equation}
so that $\Vert f_R\Vert =1$. For $u\in C^\infty(\R^3;[0,1])$ with
\begin{equation}
u(z):=\left\{
\begin{array}{ll}
   1 & \mbox{ if } |z| \leq \frac{1}{2} \\
   0 & \mbox{ if } |z| \geq 1
\end{array}
\right. 
\end{equation}
we define the cutoff function
\begin{equation}
\zerz(x_1,\cdots,x_{k-1}) :=\prod_{i=1}^{k-1} u\Big(x_i \cdot \frac{R}{Z+1}\Big).
\end{equation}
This cutoff function localizes each particle $i=1,\cdots,k-1$ in a ball of radius $\frac{R}{Z+1}$ and is invariant under permutations in $S_{k-1}$. We define 
\begin{equation}
\tpsi_0(x):= (\zerz \phi)(x_1,\cdots,x_{k-1}) f_R(x_k)
\end{equation}
and the trial state
\begin{equation}
\psi_0:=\frac{P^{\alpha_k} \tpsi_0}{\Vert P^{\alpha_k} \tpsi_0\Vert }.
\end{equation}
Following the same argument as in Section~\ref{sec_abo_01}, we have
\begin{equation}\label{eq_exist_02}
\frac{\las \tpsi_0,H_k^ZP^{\alpha_k}\tpsi_0\ras}{\Vert P^{\alpha_k}\tpsi_0\Vert^2}=\frac{\las \tpsi_0,H_k^Z\tpsi_0\ras}{\Vert \tpsi_0\Vert^2}.
\end{equation}
We split the Hamiltonian $H_k^Z$ into three parts
\begin{equation}
H_k^Z=H_{k-1}^Z+T_k+\Big(\sum_{1\leq i<k}\frac{e^2}{|x_i-x_k|}-\frac{e^2Z}{|x_k|}\Big).
\end{equation}
Using the exponential decay of the eigenfunction $\phi$, similar to \eqref{eq_abo_05}, we get
\begin{equation} \label{eq_exist_01}
\las \zerz \phi,H_{k-1}^Z\zerz \phi\ras =\mu^{\alpha_k}_{k-1} \Vert \phi\Vert^2+\mathcal{O} (e^{-cR})
\end{equation}
for some constant $c>0$. Note that for $x_k\in \supp (f_R)$ we have $|x_k|=(1+\theta)R$ for some $\theta\in [0,1]$ and by choice of $\zerz$, for $x\in \supp(\tpsi_0)$ we get
\begin{equation}\label{eq_exist_05}
\begin{split}
\sum_{1\leq i<k}\frac{e^2}{|x_i-x_k|}&\leq\sum_{1\leq i<k}\frac{e^2}{|x_k|-|x_i|} \\
&\leq \frac{e^2(k-1)(Z+1)}{(Z+Z\theta+\theta)R}
\end{split}
\end{equation}
and 
\begin{equation}\label{eq_exist_06}
-\frac{e^2Z}{|x_k|}= -\frac{e^2Z}{(1+\theta) R}.
\end{equation}
Using \eqref{eq_exist_05} and \eqref{eq_exist_06}, and $k\leq Z$, we arrive at
\begin{equation}
\begin{split}
\sum_{1\leq i<k}\frac{e^2}{|x_i-x_k|}-\frac{e^2Z}{|x_k|}&\leq -\frac{e^2(\theta+1+Z\theta)}{R(1+\theta)(Z+Z\theta+\theta)}\\
&\leq -\frac{e^2}{R(Z+Z\theta+\theta)}-\frac{e^2Z\theta}{R(1+\theta)(Z+Z\theta+\theta)}.
\end{split}
\end{equation}
The first term on the r.h.s. is increasing in $\theta$ and the second term is non-positive, which yields the bound
\begin{equation}\label{eq_exist_03}
\big\las \tpsi_0 , \big(\sum_{1\leq i<k}\frac{e^2}{|x_i-x_k|}-\frac{e^2Z}{|x_k|} \big)\tpsi_0\big\ras\leq -\frac{e^2}{(2Z+1)R}\Vert \tpsi_0\Vert^2.
\end{equation}
Furthermore, for the particle $k$ we have
\begin{equation}\label{eq_exist_04}
\begin{split}
\las \psi_0,T_k\psi_0\ras &= \Vert \zerz \phi \Vert^2 \las f_R,T_k f_R\ras\\
&\leq \Vert \zerz \phi \Vert^2 \las f_R,\frac{ p_k^2}{2} f_R\ras\\
& \leq \frac{C}{R^2}\Vert \tpsi_0\Vert^2.
\end{split}
\end{equation}
Collecting \eqref{eq_exist_02}, \eqref{eq_exist_01}, \eqref{eq_exist_03} and \eqref{eq_exist_04}, we get 
\begin{equation}\label{eq_exist_07}
\frac{\las \psi_0,H_k^ZP^{\alpha_k}\psi_0\ras}{\Vert P^{\alpha_k}\psi_0\Vert^2}\leq \mu^{\alpha_k}_{k-1}+\frac{C}{R^2}-\frac{e^2}{(2Z+1)R}<\mu_{k-1}^{\alpha_k}
\end{equation}
for sufficiently large $R$. By Theorem~\ref{thm_hvz_01} we have
\begin{equation}
\mu^{\alpha_k}_{k-1}=\inf \sigma_{ess}(H_k^ZP^{\alpha_k }).
\end{equation}
So \eqref{eq_exist_07} shows that the discrete spectrum of $H_k^Z P^{\alpha_k}$ below $\mu_{k-1}^{\alpha_k}$ is not empty, in particular a ground state of $H_k^ZP^{\alpha_k}$ exists.\color{black}
\end{proof}

%
\section{Commutator bounds via quadratic forms} \label{sec_app_04}
%
In this section we gather some auxiliary results, which are essential for 
the proof of exponential decay of eigenfunctions of pseudo--relativistic 
operators and also to give exponentially small error bounds for some of 
the error terms in the van der Waals--London asymptotic. For non-relativistic 
Schr\"odinger operators exponential bounds for the decay of eigenfunctions at 
infinity are well understood since the groundbreaking works 
of Slaggie and Wichmann for three--body systems \cite{slaggie-wichmann}, 
Ahlrichs for atoms \cite{Ahlrichs}, O'Connor \cite{connor1973}, 
Combes and Thomas \cite{combes-thomas},  Deift, Hunziker, Simon,  
and Vock \cite{deift-hunziker-simon-vock} for multi--particle systems,  
which culminated in the work of Agmon  \cite{agmonlecture}.  
Of course, O’Connors analytic method for proving exponential decay for 
eigenfunctions also works, neglecting symmetry issues, for non--local 
operators like $\sqrt{p^2+1}-1$ 
due to the analyticity of the corresponding symbol in a strip 
$\{z\in\C^{3N}: |\Im(z)|<\delta\}$ for suitable $\delta>0$. 
This was done by Nardini in \cite{nardini1988asymptotic}, but it 
does not allow to include the required 
symmetry of the eigenstates. 
Thus we develop a variant of Agmon's method, which is based on 
configuration space methods, for multi--particle pseudo--relativistic 
Schr\"odinger operators, because it is invariant under permutation of 
particles and easily allows to include particle symmetries. 
However, due to the non--locality of the pseudo--relativistic operator 
$\sqrt{p^2+1}-1$, this is considerably harder than in the 
non--relativistic case.

Our main tool is  an extension of the localization error formula of Loss, Lieb, and Yau in \cite{lieb1988} in the spirit of \cite{Griesemer}, see Lemma \ref{lem_ims_extended} below. Before we can state it, we need to first investigate the behavior of $H^{1/2}(\R^d)$ under multiplication with bounded Lipschitz continuous functions. 
\begin{lem}\label{lem_H-onehalf-invariance}
  Let $\xi:\R^d\to \C$ be a bounded Lipschitz continuous function. Then for any $f\in H^{1/2}(\R^d)$ the product $\xi f$ is also in $ H^{1/2}(\R^d)$. 
\end{lem}
\begin{rema}
	That $H^{1/2}(\R^d)$ is invariant under multiplication with bounded $C^\infty$ functions, whose derivative is also bounded, is well known, see \cite[Theorem 7.16]{lieb2001analysis}. That it is enough to have bounded Lipschitz functions, seems to be less appreciated. 
	As our proof shows, it is enough to assume that $\xi$ is bounded and 
	$\gamma$-H\"older continuous  with H\"older constant $1/2<\gamma\le 1$. 
\end{rema}
\begin{proof}
  Clearly, if $\xi $ is bounded, then $\Vert \xi f \Vert\le \Vert\xi   \Vert_\infty \Vert f \Vert$, so it is enough to show that $\xi f$ is in the domain of the fractional Laplacian $|p|^{1/2}= (-\Delta)^{1/4}$. According to \cite[Theorem 7.12]{lieb2001analysis} the quadratic form of $|p|$ is given by 
  \begin{align*}
  	\las f, |p| f \ras  = c_d \iint\limits_{\R^d\times\R^d} \frac{|f(x)-f(y)|^2}{|x-y|^{d+1}}\, dxdy
  \end{align*}
  with $c_d= \frac{\Gamma(\frac{d+1}{2})}{2\pi^{(d+1)/2}}$, and $\Gamma$ being the Gamma function. Hence 
  \begin{align}\label{eq_rep-p}
  	\||p|^{1/2} \xi f\|^2 =   	\las \xi f, |p| \xi f \ras  = c_d \iint\limits_{\R^d\times\R^d} \frac{|\xi(x)f(x)-\xi(y)f(y)|^2}{|x-y|^{d+1}}\, dxdy
  \end{align}
  Using 
  \begin{align*}
  	|\xi(x)f(x)- \xi(y)f(y)|^2 
  		&= |(\xi(x)-\xi(y)) f(x) + \xi(y)\big(f(x)-f(y)\big)|^2 \\
  		&\le 2 |\xi(x)-\xi(y)|^2 |f(y)|^2+ 2\|\xi\|_\infty^2 | f(x)-f(y) |^2 
  	\end{align*}
   in \eqref{eq_rep-p} one has 
   \begin{align*}
 	  	\||p|^{1/2} \xi f\|^2 
  	  		&\lesssim \iint\limits_{\R^d\times\R^d} \frac{|\xi(x)-\xi(y)|^2}{|x-y|^{d+1}}  |f(y)|^2 \, dxdy 
  	  			+  \|\xi\|_\infty^2 \iint\limits_{\R^d\times\R^d} \frac{|f(x)-f(y)|^2}{|x-y|^{d+1}}\, dxdy \\
  	  		&\lesssim  \sup_{x\in\R^d}\int\limits_{\R^d} \frac{|\xi(x)-\xi(y)|^2}{|x-y|^{d+1}}\, dy \|f\|^2 
  	  			+ \|\xi\|_\infty^2 \Vert|p|^{1/2}f\Vert^2   
  \end{align*}
 With $L$ the Lipschitz constant of $\xi$, we have 
 \begin{align*}
 	|\xi(x)-\xi(y)|\le \min(L|x-y|, 2\|\xi\|_\infty) \, .
 \end{align*}
 Hence 
 \begin{align*}
 	  \sup_{x\in\R^d}&\int\limits_{\R^d} \frac{|\xi(x)-\xi(y)|^2}{|x-y|^{d+1}}\, dy
 	  \le \int\limits_{\R^d} \frac{\min(L^2|y|^2, 4\|\xi\|_\infty^2)}{|y|^{d+1}}\, dy \, 
 	  		  \lesssim L^2 \|\xi\|_\infty^2
 \end{align*}
 by evaluating the integral in spherical coordinates.  This shows 
 \begin{align*}
 	\||p|^{1/2} \xi f\|^2 
		\lesssim L^2\|\xi\|_\infty^2 \|f\|^2 + \|\xi\|_\infty^2 	\||p|^{1/2}  f\|^2 <\infty 
 \end{align*}
 for all $f\in H^{1/2}(\R^d)$, thus $\xi f\in H^{1/2}(\R^d)$. 
\end{proof}

\begin{lem}[Commutation formula for the relativistic kinetic energy, one particle case]\label{lem_ims_extended} 
For a bounded real-valued Lipschitz function $\xi$, $T=\sqrt{p^2+1}-1$, and any function 
$\varphi\in H^{1/2}(\R^{3})$ we have 
\begin{equation}\label{eq_ims-extended}
\re \las \xi^2 \varphi,T \varphi\ras 
	= \las \xi \varphi,T \xi\varphi\ras 
		- \calL_\xi(\varphi,\varphi)
\end{equation}
where the quadratic form $\calL_\xi$ is given by 
\begin{equation}\label{eq_loc_error}
\calL_\xi(\varphi,\varphi) = 
	\frac{1}{4\pi^2} \iint\limits_{\R^3\times\R^3} \frac{K_2(|x-y|)}{|x-y|^2}(\xi(x)-\xi(y))^2 \ol{\varphi(x)} \varphi(y)\,  dx dy
\end{equation}
where $K_2$ is the modified Bessel function of order two.
\end{lem}
\begin{rema}
	An analogous formula, when $1=\sum_{j=1}^K \xi_j^2$ for a partition of unity, was found 
	by Michael Loss and used in \cite{lieb1988} (see formula (3.6) 
	in Theorem 9 in \cite{lieb1988}). For our applications 
	it is important to have \eqref{eq_ims-extended} in its full generality, however. 
\end{rema}
\begin{proof}
  Note that $\calL_\xi(\varphi,\varphi)$ is well-defined for all \
  $\varphi\in L^2(\R^3)$, since $\xi$ is a bounded  Lipschitz continuous function, so    
  $\xi^2\varphi$ and $\xi\varphi$ are in $H^{1/2}(\R^3)$, due 
  to Lemma  \ref{lem_H-onehalf-invariance}. So all terms in \eqref{eq_ims-extended} are well-defined. 
   According to \cite[Theorem 7.12]{lieb2001analysis} we have
 \begin{equation}\label{eq_exp_17}
  \las \varphi,T \varphi\ras=\frac{1}{4\pi^2}\iint\limits_{\R^3\times\R^3} \frac{K_2(|x-y|)}{|x-y|^2}\, |\varphi(x)-\varphi(y)|^2\, dxdy
 \end{equation}
 for $\varphi\in H^{1/2}(\R^3)$. By polarization, this extend to 
 \begin{equation}
 	\las f,T g\ras=\frac{1}{4\pi^2}\iint\limits_{\R^3\times\R^3} 
 	\frac{K_2(|x-y|)}{|x-y|^2}\ol{(f(x)-f(y))}(g(x)-g(y))dxdy 
 \end{equation}
 for two functions $f,g\in H^{1/2}(\R^3)$. Thus
 \begin{equation}\label{eq_ims-anfang}
 \begin{split}	
 	\re & \las \xi^2\varphi, T\varphi \ras 
 	=  \\
 	&\frac{1}{4\pi^2}\iint\limits_{\R^3\times\R^3} 
 	\frac{K_2(|x-y|)}{|x-y|^2}\re\big(\ol{(\xi^2(x)\varphi(x)-\xi^2(y)\varphi(y))}(\varphi(x)-\varphi(y))\big)dxdy 
 \end{split}
 \end{equation}
 For $a,b\in\R $ and $c,d\in\C$ one has the simple identity 
 \begin{equation}\label{eq_simple}
 	\re \big(\ol{(a^2c-b^2d)}(c-d) \big)
 	- |ac-bd|^2 
 	=- (a-b)^2\re(\ol{c}d) \, .
 \end{equation}
 Using \eqref{eq_simple} in \eqref{eq_ims-anfang} immediately yields \eqref{eq_ims-extended} and \eqref{eq_loc_error}, since, by symmetry, $\calL_\xi(\varphi,\varphi)$ is real. 
\end{proof}

\begin{lem}[Simple bound on the commutation error, one particle case]\label{lem_L_xi-bound}
	Assume that $\xi:\R^d\to\R $ is Lipschitz. Then 
	\begin{equation}\label{eq_L_xi-bound}
		|\calL_\xi(\varphi,\varphi)| 
			\le \frac{3L_\xi^2}{2}  \big\|\varphi\big\|^2
	\end{equation}
	where $L_\xi $ is the Lipschitz constant of $\xi$. 
\end{lem}
\begin{proof}
  Using the Lipschitz continuity, 
  $
  	|\xi(x)-\xi(y)|\le L_\xi|x-y|
  $ for all $x,y\in \R^3$, 
 in \eqref{eq_loc_error}, we get 
 \begin{align*}
 	|\calL_\xi&(\varphi,\varphi)| 
 		\le  \frac{L_\xi^2}{4\pi^2} \iint\limits_{\R^3\times\R^3} K_2(|x-y|)|\varphi(x)| |\varphi(y)|\,  dx dy \\
 		&\le   \frac{L_\xi^2}{4\pi^2} \iint\limits_{\R^3\times\R^3} K_2(|x-y|)|\frac{1}{2}\left(\varphi(x)|^2 + |\varphi(y)|^2\right)\,  dx dy 
 			=    \frac{L_\xi^2}{4\pi^2} \|K_2(|\cdot|)\|_{L^1(\R^3)}\|\varphi\|^2\, .
 \end{align*}
 Since $\int_0^\infty K_2(r)r^2\, dr= \frac{3\pi}{2}$, see \cite[Formula 11.4.22]{abramowitz1964handbook}, we have 
 \begin{align}
 	\Vert |K_2(|\cdot|)\Vert_{L^1(\R^3)} 
 		= 4\pi \int_0^\infty K_2(r) r^2\, dr 
 		=  6\pi^2 \, , 
 \end{align}
 this proves the lemma. 
\end{proof}

If the weight $\xi$ is of the form $\xi=\chi e^F$, with $\chi$ a cut--off 
function and $F$ bounded and Lipschitz, then the Lipschitz constant of 
$\xi$ is bounded by  $L_\xi\lesssim (L_\chi+L_F)e^{\Vert F \Vert_\infty}$, 
no better bound being available,  in general.  
Thus the simple commutation error bound from Lemma \ref{lem_L_xi-bound} is 
insufficient for the application to exponential bounds 
for eigenfunctions of pseudo--relativistic Schr\"odinger operators, where we have to use a sequence of bounded functions  
 $F_n$, whose Lipschitz constant is uniformly 
 bounded in $n\in\N$, but for which $\Vert F_n \Vert_\infty$ 
 diverges as $n$ grows. The next Lemma shows how such a refined bound can be achieved.  

\begin{lem}[Refined bound on the commutation error, one particle case]\label{lem_L_chi-expF-bound}
	Assume that $\xi=\chi e^F$ with $F$ Lipschitz and bounded and $\chi$ Lipschitz and $0\le \chi\le 1$. Then 
	\begin{equation}\label{eq_L_chi-expF-bound} 
		|\calL_\xi(\varphi,\varphi)| 
			\le \left(L_\chi +L_F\right)^2\frac{\big\|K_2e^{L_F|\cdot|} \big\|_{L^1(\R^3)}}{4\pi^2} \big\|e^{F}\varphi\big\|^2
	\end{equation}
	where $L_F$, respectively $L_\chi$, is the Lipschitz constant of $F$, respectively $\chi$. 
\end{lem}
\begin{rema}\label{rem:finiteness}
For the application to exponential bounds for eigenfunctions it is important that the exponential weight $e^F$ appears only in the form  $e^F\varphi$ in the  r.h.s.\ of \eqref{eq_L_chi-expF-bound} and the rest depends only on the Lipschitz constants of $F$ and $\chi$.
Using the known asymptotics of the modified Bessel--function, $K_2(r)\sim \sqrt{\frac{\pi}{2r}}e^{-r}$ for large $r$ and $K_2(r)\sim \frac{2}{r^2}$ for small $r>0$, {\rm(}e.g., \cite[(9.7.2)]{abramowitz1964handbook} and the remark after \cite[(9.7.4)]{abramowitz1964handbook}, for large $r\in \R$ or  \cite[\S 4.8]{andrews1999} and \cite[(4.12.6)]{andrews1999} for a more detailed discussion{\rm)} one sees that  
\begin{equation*}
  \big\|K_2e^{L_F|\cdot|} \big\|_{L^1(\R^3) }
  = 4\pi \int_0^\infty K_2(r) e^{L_F r}r^2\, dr <\infty 
\end{equation*}
iff $L_F<1$. 
	It is easy to see that any function $F$ of the form 
	\begin{equation*}
		F(x)=F_{\mu,\veps}(x)=\frac{\nu|x|}{1+\veps|x|}
	\end{equation*}
	with $\nu,\veps\ge0$, is subadditive, that is,  $F(x+y)\le F(x)+F(y)$ for all $x,y$. Hence, by the reverse triangle inequality 
	\begin{align*}
		|F(x)-F(y)|\le F(x-y) \le \nu|x-y|
	\end{align*}
	which shows that $F_{\nu,\veps}$ it is Lipschitz continuous with constant 
	$L_{F_{\nu,\veps}}\le \nu$. \\ 
	Furthermore, if 
	$\chi$ is Lipschitz, its scaled version 
	\begin{equation}
		\chi_R(x)= \chi(x/R)
	\end{equation}
	for $R>0$, is Lipschitz with constant $L_{\chi_{R}}= L_\chi/R$. Such a choice for $F$ and $\chi$ makes the factor $(L_F+L_\chi)^2$ as small as one wishes and taking the limit $\veps \to  0+$ allows to recover the exponentially growing weight $e^{F_{\nu,0}}=e^{\nu|\cdot|}$. 
\end{rema}
\begin{proof}
  Lemma \ref{lem_ims_extended} gives 
  \begin{align}\label{eq_awesome1}
  	|\calL_\xi(\varphi,\varphi)| 
  		\le  \frac{1}{4\pi^2} \int_{\R^{6}} \frac{K_2(|x-y|)}{|x-y|^2}(\xi(x)-\xi(y))^2 |\varphi(x)| |\varphi(y)| dy dx
  \end{align}
  Since 
  \begin{align}
  	\xi(x)-\xi(y) &= \chi(x)e^{F(x)}-\chi(y)e^{F(y)} \nonumber\\
  		&= (\chi(x)-\chi(y))e^{F(x)} + \chi(y)\big( e^{F(x)}- e^{F(y)} \big) \label{eq_difference1}\\
  		&=  \chi(x)\big( e^{F(x)}- e^{F(y)}  \big) 
  			+ \big( \chi(x)-\chi(y) \big) e^{F(y)}
  			\label{eq_difference2}
  \end{align}
  and averaging \eqref{eq_difference1} and \eqref{eq_difference2} one sees 
  \begin{align}\label{eq_awesome2}
  	\xi(x)-\xi(y)& = \frac{1}{2}\big( \chi(x)-\chi(y) \big) \big( e^{F(x)} + e^{F(y)} \big) 
  		+ \frac{1}{2}\big( \chi(x)+\chi(y) \big) \big( e^{F(x)} - e^{F(y)} \big) \nonumber \\
  		&= \big( \chi(x)-\chi(y) \big) \cosh\left( \frac{F(x)-F(y)}{2} \right) e^{\frac{1}{2}F(x)}e^{\frac{1}{2}F(y)} 
  		\nonumber\\ 
  		&\phantom{+~}+ \big( \chi(x)+\chi(y) \big) \sinh\left( \frac{F(x)-F(y)}{2} \right)e^{\frac{1}{2}F(x)}e^{\frac{1}{2}F(y)}
  \end{align}
  Now we note that due to the subadditivity of $F$ we have
  \begin{equation*}
  	F(x)-F(y)\le |F(x)-F(y)|\le F(x-y) \le L_F|x-y|
  \end{equation*} 
  and 
  \begin{equation*}
  	|\chi(x)-\chi(y)|\le L_\chi|x-y|. 
  \end{equation*}
  Moreover,  
  \begin{align*}
  	|\sinh(a)|= \sinh(|a|)=  \frac{1}{2}e^{|a|}\big(1-e^{-2|a|}  \big) \le |a|e^{|a|}, 
  \end{align*} 
  thus 
  \begin{equation*}
  	\sinh\left( \frac{F(x)-F(y)}{2} \right)\le \frac{L_F|x-y|}{2}e^{L_F|x-y|/2}
  \end{equation*} and 
  \begin{equation*}
  	  \cosh\left( \frac{F(x)-F(y)}{2} \right) \le e^{L_F|x-y|/2}\, .
  \end{equation*}
  Hence from \eqref{eq_awesome2} we get the bound 
  \begin{align*}
  	|\xi(x)-\xi(y)|\le 
  	\left(L_\chi|x-y| +L_F|x-y|\right) e^{L_F|x-y|/2}e^{\frac{1}{2}F(x)}e^{\frac{1}{2}F(y)}
  \end{align*}
  and using this in \eqref{eq_awesome1} yields  
    \begin{align}\label{eq_awesome3}
  	& |\calL_\xi(\varphi,\varphi)| 
  		\le  \frac{\left(L_\chi +L_F\right)^2}{4\pi^2} \int_{\R^{6}} K_2(|x-y|)e^{L_F|x-y|}
  			 |e^{F(x)}\varphi(x)| |e^{F(y)}\varphi(y)| dy dx 
  \end{align}
  Since the Bessel--function $K_2$ is positive  
  \begin{align*}
  	\int_{\R^{6}} &K_2(|x-y|)e^{L_F|x-y|}
  			 |e^{F(x)}\varphi(x)| |e^{F(y)}\varphi(y)| dy dx \\
  		&\le   \frac{1}{2}\int_{\R^{6}} K_2(|x-y|)e^{L_F|x-y|} 
  		\left(
  			 |e^{F(x)}\varphi(x)|^2+ |e^{F(y)}\varphi(y)|^2\right) dy dx \\
  		&= \big\|K_2e^{L_F|\cdot|} \big\|_{L^1(\R^3) }
  			\big\|e^{F}\varphi\big\|^2\, ,
  \end{align*}
  thus \eqref{eq_awesome3} yields \eqref{eq_L_chi-expF-bound}. 
\end{proof}
For our derivation of upper and lower bounds to the van 
der Waals energy, we also need an additional 
refinement, which is our main tool for showing 
that the localization error is exponentially small,
see Section \ref{sec_below_02}. 

\begin{lem}\label{lem_loc_01}
  Let $\chi:\R^3\to [0,1]$ be 
  Lipschitz continuous cut--off function which varies only on the transition region 
  $\Omega\subset \R^3$, i.e., $\chi(x)\in\{0,1\}$ for all $x\not\in \Omega$. Given $d>0$ let 
  $\Omega_d= \{x\in\R^3: \dist(x,\Omega)\le d\}$. 
Then 
  \begin{align}
  	|\calL_\chi(\varphi,\varphi)|\le   
  		C \left( L_\chi^2\Vert\Theta_d \varphi\Vert^2  + \frac{e^{-d/2}}{d^2}\Vert \varphi \Vert^2 \right)
  \end{align}
  for all $R>0$, where $\Theta_d=\id_{\Omega_d}$ is a cut--off function corresponding to a slightly enlarged region where $\chi$ varies and the constant $C$ depends only $ \Vert K_2(|\cdot|)e^{|\cdot|/2}\Vert_{L^1(\R^3)}$.
\end{lem} 
\begin{proof}
  To prove the Lemma, it is convenient to split the integral into two regions,  
  \begin{equation*}
  	A_d= \{ (x,y)\in \R^3\times\R^3:\, |x-y|< d \} 
  \end{equation*}
  and its complement 
  \begin{equation*}
  	A_d^c= \{(x,y)\in \R^3\times\R^3:\, |x-y|\ge  d\}\, .
  \end{equation*}
  Note that if $x\not\in\Omega_d$ and $|x-y|<d$, then $\chi(x)=0$ implies $\chi(y)=0$ and $\chi(x)=1$ implies $\chi(y)=1$. Thus  
  \begin{equation}
    \begin{split}
  	  	\big(\chi(x)-\chi(y)\big)^2 
  	  &= \big(\chi(x)-\chi(y)\big)^2 \id_{A_d}(x,y)\Theta_d(x) \Theta_d(y)  \\
  	  &\phantom{=~~}+  \big(\chi(x)-\chi(y)\big)^2 \id_{A_d^c}(x,y) \, .
    \end{split}
  \end{equation}
  By assumption,  $0\le \chi_R\le 1$ and $\chi$ is Lipschitz continuous with Lipschitz constant 
  $L_\chi$.  Thus 
  $ |\chi(x) -\chi(y)| \le \min\left(L_\chi|x-y|, 2 \right)$
  for all $x,y\in\R^3$. Hence  
  \begin{equation}
  \begin{split}
  	 &\frac{K_2(|x-y|)}{|x-y|^2}\big(\chi(x)-\chi(y)\big)^2 \\
  	 &\quad\le L_\chi^2 K_2(|x-y|) \id_{A_d}(x,y)\Theta_d(x) \Theta_d(y)  
  	   + \frac{4}{d^2}K_2(|x-y|) \id_{A_d^c}(x,y)  
  \end{split}
  \end{equation}
  Using this  bounds in the definition \eqref{eq_loc_error} of $L_\chi$ one sees 
  \begin{equation}
  	|\calL_{\chi}(\varphi,\varphi)|
  		\le \frac{L_\chi^2}{4\pi^2} I_1 
  			+ \frac{1}{\pi^2d^2} I_2
  \end{equation}
  with 
  \begin{align*}
  	I_1&=\iint\limits_{A_d} K_2(|x-y|) \Theta_d(x)|\varphi(x)|\Theta_d(y)|\varphi(y)| \, dydx\\
  	&\le \frac{1}{2}\iint_{\R^3\times\R^3} K_2(|x-y|) \left( |\Theta_d(x)\varphi(x)|^2 + |\Theta_d(y)\varphi(y)|^2 \right)\, dydx \\
  	&= \Vert K_2(|\cdot|)\Vert_{L^1(\R^3)} \Vert\Theta_d\, \varphi \Vert^2\, , 
  \end{align*}
  since $K_2$ is positive.  
  Similarly, using also $|x-y|\ge d$ on $A_d^c$,  
    \begin{align*}
  	I_2&=\iint\limits_{A_d^c} K_2(|x-y|) |\varphi(x)||\varphi(y)| \, dydx\\
  	&\le e^{-d/2}\iint\limits_{\R^3\times\R^3} K_2(|x-y|) e^{|x-y|/2}|\varphi(x)||\varphi(y)| \, dydx\\
  	&\le e^{-d/2} \Vert K_2(|\cdot|)e^{|\cdot|/2}\Vert_{L^1(\R^3)} \Vert\varphi \Vert^2
  		= C e^{-d/2}  \Vert\varphi \Vert^2 \, . \qedhere
  \end{align*} 
  and from Remark \ref{rem:finiteness} we know that 
  $C=  \Vert K_2(|\cdot|)e^{|\cdot|/2}\Vert_{L^1(\R^3)}<infty$.  
\end{proof}
We also have to extend the commutation error bound from Lemma \ref{lem_L_chi-expF-bound} to the many-body case, which needs some modifications, mainly in notation. 
Let $C$ be  cluster, i.e, $C\subset [N]=\{1,2,\ldots,N\}$.  Given $j\in C$ and any $y\in \R^3$ we denote 
by $y^j$ the coordinate in $\R(C)$ with $(y^j)_l= y\delta_{j,l}$, where $\delta_{j,l}$ is 
the Kronecker--delta. That is, if one relabels the coordinates in $\R(C)$ so that $\R(C)= \R^K$, with $K=$ number of particles in the cluster $C$, one has 
$y^j= (0,\ldots,0,y,0,\ldots,0)$ with $y$ in the $j^{\text{th}}$ slot. With this notation we have 

\begin{lem}[Commutation formula for the  multi--particle case]\label{lem_ims-extended-many-body}
	  Let $T=\sum_{k\in C}T_k$, with $T_k= \sqrt{p_k^2+1}-1$, the total kinetic energy operator of the particles in the cluster. For any bounded Lipschitz continuous function $\xi$ and 
	  any $\psi\in H^{1/2}(\R(C))$ we have  
	\begin{equation}\label{eq_ims-extended-many-body}
		\re\las \xi^2\psi, \sum_{k\in C} T_k \psi \ras 
			= \las  \xi \psi, \sum_{k\in C} T_k \xi \psi \ras  
				- L^C_\xi(\psi,\psi) 
	\end{equation} 
	 as quadratic forms, with 
	\begin{equation}\label{eq_commutation_error-many-body}
		L^C_\xi(\varphi,\varphi) = 
	\frac{1}{4\pi^2} \sum_{j\in C}\int_{\R(C)}\int_{\R^3} t(y)(\xi(x)-\xi(x+y^j))^2 \re(\ol{\varphi(x)} \varphi(x+y^j))) dy dx
\end{equation}
and $t(y)=\frac{K_2(|y|)}{|y|^2}$ for $y\in \R^3$. 
\end{lem}
\begin{proof}
  The proof is a straightforward adaptation of the arguments in the proof of Lemma \ref{lem_L_xi-bound}. 
\end{proof}

\begin{lem}[Refined bound on the commutation error, multi--particle case]\label{lem_L_xi-bound-many-body}
	Assume that $\xi=\chi e^F$ with $F$ bounded and Lipschitz and 
	$\chi$ Lipschitz. Then the commutation error $\calL^C_\xi$ from \eqref{eq_ims-extended-many-body} can be bounded by 
	\begin{equation}\label{eq_commutation-error-bound-many-body}
		|\calL^C_\xi(\psi,\psi)| 
			\le KC_{L_F}\left(L_\chi +L_F\right)^2 \big\|e^{F}\psi\big\|^2
	\end{equation}
	where $L_F$, respectively $L_\chi$, is the Lipschitz constant of $F$, respectively $\chi$, $K$ is the number of particles in the cluster $C$, and 
	\begin{equation}
		C_{L}= \frac{\big\|K_2e^{L|\cdot|} \big\|_{L^1}}{4\pi^2}
	\end{equation} 
\end{lem}
\begin{proof}
  As in the proof of Lemma \ref{lem_L_xi-bound}, we have 
  \begin{equation}
  	\begin{split}
  	  &(\xi(x)-\xi(x+y^j))^2 \\
  	  &= \Big[ (\chi(x)-\chi(x+y^j))\cosh\Big( \frac{F(x)-F(x+y^j)}{2} \Big) \\
  	  	&\phantom{+~~~}	 + (\chi(x)+\chi(x+y^j)\sinh\Big( \frac{F(x)-F(x+y^j)}{2} \Big)
  	  	 \Big]^2 e^{F(x)} e^{F(x+y^j)}
  	\end{split}
  \end{equation}
  Since $|\chi(x)-\chi(x+y^j)|\le L_\chi|y|$ and 
  $|F(x)-F(x+y^j)|\le L_F|y|$, we can conclude as in the proof of Lemma \ref{lem_L_xi-bound} to get \eqref{eq_commutation-error-bound-many-body}. 
\end{proof}

\section{Intercluster interaction in diatomic molecules}\label{sec_app_01}
%
In this part we estimate the term $\las \jbs \psi,I_\bmin \jbs\psi\ras$ which is an important part in the proof of Theorem~\ref{thm_pre_02}. For these estimates we will use  orthogonality relations, which will be proven in Appendix~\ref{sec_app2}. \par 
Denote by $P_n(z), n\in \N, z\in \R$ the $n$-th degree Legendre polynomial, these polynomials are generated by  $(1-2zt+t^2)^{-\frac{1}{2}}$ (see \cite[22.9.12]{abramowitz1964handbook}). More explicitly, for $-1<z<1$ and $|t|<1$ we have 
\begin{equation}
\frac{1}{\sqrt{1-2zt+t^2}}=\sum_{n=0}^\infty P_n (z)t^n.
\end{equation}
Consequently, for $D,h\in \R^3$ with $h<D$ we get
\begin{equation}\label{eq_app_01}
\frac{1}{|D-h|}=\sum_{n=0}^\infty P_n\left(\frac{h}{|h|}\cdot \frac{D}{|D|}\right)\frac{|h|^n}{|D|^{n+1}}.
\end{equation}
In particular for $n=2,3,4$ we have
\begin{equation}
P_2(z)=\frac{1}{2}(3z^2-1),\quad 
P_3(z)=\frac{1}{2}(5z^3-3z),\quad
P_4(z)=\frac{1}{8}(35z^4-30z^2+3).
\end{equation}
Let $\beta$ be a decomposition into two clusters $\cc_1$ and $\cc_2$ with $\sharp\cc_1=Z_1$ and $\sharp\cc_2=Z_2$. The intercluster interaction is given by
\begin{equation}\label{eq_app_11}
I_\beta(x)=-\sumi \frac{e^2Z_2}{|x_i-X_2|} -\sumj\frac{e^2Z_1}{|x_j-X_1|}+\sumij \frac{e^2}{|x_i-x_j|}+\frac{e^2 Z_1 Z_2}{|D|}.
\end{equation}
For $i_k\in \cc_1$ we define
\begin{equation}\label{eq_app_16}
\mathcal{F}^{(1)}_n(x):=\sumi |x_i|^n P_n\left(\frac{x_i}{|x_i|}\cdot \frac{D}{|D|}\right),
\end{equation}
\begin{equation}\label{eq_app_17}
\mathcal{F}^{(2)}_n(x):=\sumj |x_j|^n P_n\left(\frac{-x_j}{|x_j|}\cdot \frac{D}{|D|}\right)
\end{equation}
and
\begin{equation}\label{eq_app_18}
\mathcal{F}^{(3)}_n(x):=\sumij |x_i-x_j|^n P_n\left(\frac{x_i-x_j}{|x_i-x_j|}\cdot \frac{D}{|D|}\right).
\end{equation}
Let
\begin{equation}\label{eq_app_04}
f_n(x):= -e^2Z_2 \mathcal{F}^{(1)}_n(x)-e^2Z_1\mathcal{F}^{(2)}_n(x)+e^2\mathcal{F}^{(3)}_n(x).
\end{equation}
Note that  for $n=2,3$ the functions defined in \eqref{eq_app_04} are the same as $f_2$ and $f_3$ in \eqref{eq:def-fA} and \eqref{eq:def-fB}. Observe that $x\in \operatorname{supp}(J_\beta)$ implies $|x_i-X_1|<<\D$ for $i\in \cc_1$ and $|x_j-X_2|<<\D$ for $j\in \cc_2$ and the Taylor series of $I_\beta$ converges. This yields 
\begin{equation}\label{eq_app_09}
I_\beta(x)=\sum_{n=0}^\infty \frac{ \U^* f_n(x)}{\D^{n+1}}+\frac{e^2Z_1Z_2}{\D} \qquad \forall x\in \operatorname{supp} (J_\beta)
\end{equation}
where $\U$ is defined in \eqref{def_ubeta}.
%
%
\begin{lem}\label{lem_app_01}For any decomposition $\beta \in \Da$
\begin{equation}
\frac{ \U^* f_0(x)}{\D}+\frac{e^2Z_1Z_2}{\D}=0 \quad \mbox{and}\quad  f_1(x)=0.
\end{equation}
\end{lem}
\begin{proof}
For $\beta\in \Da$ we have $\sharp\cc_1=Z_1$ and $\sharp\cc_2=Z_2$. Since $P_0(z)=0$, by \eqref{eq_app_16} - \eqref{eq_app_18} we get
\begin{equation}
\mathcal{F}^{(1)}_0(x)=Z_1,\quad \mathcal{F}^{(2)}_0(x)=Z_2,\ \text{and } \mathcal{F}^{(3)}_0(x)=Z_1Z_2.
\end{equation}
By definition of $ f_0$ in \eqref{eq_app_04} this implies
\begin{equation}
\U^* f_0(x)=\U^* (-e^2Z_2Z_1)=-e^2Z_2Z_1
\end{equation}
which proves the first part of the lemma. Since $P_1(z)=z$, writing $e_D:=\frac{D}{\D}$ we have
\begin{equation}\begin{split}
&\mathcal{F}^{(1)}_1(x)=\sumi x_i\cdot e_D,\\
&\mathcal{F}^{(2)}_1(x)=\sumj -x_j\cdot e_D
\end{split}
\end{equation}
and
\begin{equation}
\mathcal{F}^{(3)}_1(x)=\sumij (x_i-x_j)\cdot e_D.
\end{equation}
By definition
\begin{equation}
\U^*  f_1(x)=\U^*\Big(-\sumi e^2 Z_2(x_i\cdot e_D)-\sumj e^2Z_1(-x_j\cdot e_D)+\sumij e^2[(x_i-x_j)\cdot e_D]\Big)=0.
\end{equation}
\end{proof}
In the next lemma we will establish a bound of the remainder in the Taylor expansion of $I_\beta$. Let us define the potential
\begin{equation}\label{eq_below_04}
I^o_\beta(x):= (I_\beta \chi_{J_\beta} )(x)
\end{equation}
where $\chi_{J_\beta}(x)$ is the characteristic function of the support of $J_\beta$.
%
%
\begin{lem}\label{lem_app_06}
Let $\beta\in \Da$, then for any $k\geq 2$ there exists a constant $0<C<\infty$ such that for $x\in \operatorname{supp}(I_\beta^o)$ we have
\begin{equation}\label{eq_app_06}
\Big| I_\beta^o (x)-\sum_{n=2}^{k-1}\frac{\U^* f_n(x)}{\D^{n+1}}\Big|\leq C\frac{\big(d_\beta( x) \big)^k}{\D^{k+1}}
\end{equation}
where
\begin{equation}
d_\beta( x) :=\Big(\sum_{l=1,2}\sum_{i\in \cc_l}|x_i-X_l|^2\big) ^\frac{1}{2}.
\end{equation}
\end{lem}
\begin{rema}\label{rema_app_01}
Notice that $d_\beta(\cdot )$ characterizes how far away the particles in $\cc_1$ and $\cc_2$ are from their respective nucleus. This norm does not depend on the distance $\D$ between the nuclei. In particular
\begin{equation}\label{eq_app_03}
\U d_\beta( \cdot )=\Vert \cdot \Vert.
\end{equation}
\end{rema}
\begin{proof}
Note that for $k=2$ the sum on the l.h.s of \eqref{eq_app_06} is the empty sum which, by convention, is zero.
The $k$-th summand of the Taylor expansion of $I^o_\beta$ is
\begin{equation}
\begin{split}
&\U^*\Big(-e^2Z_2\sumi P_k\left(\frac{x_i}{|x_i|}\cdot e_D\right) \frac{|x_i|^k}{|D|^{k+1}}-e^2Z_1\sumj P_k\left(\frac{-x_j}{|x_j|}\cdot e_D\right) \frac{|x_j|^k}{|D|^{k+1}}\\
&+e^2 \sumij P_k\left(\frac{x_i-x_j}{|x_i-x_j|}\cdot e_D\right) \frac{|x_i-x_j|^k}{|D|^{k+1}}\Big).
\end{split}
\end{equation}
We apply the Taylor theorem with a remainder in Lagrange form. Since the Legendre polynomials take values between $-1$ and $1$ on the interval $[-1,1]$, the Lagrange form remainders are bounded above by one. Consequently
\begin{equation}\nonumber
\begin{split}
\Big|I^o_\beta-\sum_{n=2}^{k-1}\frac{\U^* f_n}{\D^{n+1}}\Big|&\leq  \sumi \frac{e^2Z_2|x_i-X_1|^k}{\D^{k+1}}+\sumj\frac{e^2Z_1|x_j-X_2|^k}{\D^{k+1}}\\
&\quad+\sumij \frac{e^2|(x_i-X_1)-(x_j-X_2)|^k}{\D^{k+1}}
\end{split}
\end{equation}
and there exists a constant $C$ such that
\begin{equation}
\Big| I_\beta^o(x) -\sum_{n=2}^{k-1}\frac{ \U^*f_n(x)}{\D^{n+1}}\Big|\leq C\frac{\big(d_\beta( x) \big)^k}{\D^{k+1}} \quad \forall x\in \operatorname{supp} (I^o_\beta). 
\end{equation}
\end{proof}
%
%
\begin{corr}\label{cor_app_03}
Let $\beta\in \Da$ and $\varphi_1,\ \varphi_2\in \ltn$ such that there exists $b>0$ and $A_0$ with
\begin{equation}\label{eq_app_14}
\Vert e^{b|\cdot|}\varphi_2\Vert^2 \leq A_0\Vert \varphi_2\Vert^2.
\end{equation}
Then for any $k\geq 2$ there exists a constant $C_k(b,A_0)<\infty$  such that 
\begin{equation}\label{eq_app_13}
\Big|\Big\las \U^*\varphi_1, \Big(I^o_\beta -\sum_{n=2}^{k-1} \frac{\U^*f_n}{|D|^{n+1}} \Big)\U^*\varphi_2\Big\ras\Big| \leq C_k\D^{-(k+1)} \Vert \varphi_1\Vert \ \Vert \varphi_2 \Vert.
\end{equation}
\end{corr}
\begin{proof}To prove \eqref{eq_app_13} we apply Lemma~\ref{lem_app_06} to get
\begin{equation}
\Big|\Big\las \U^*\varphi_1, \Big(I^o_\beta -\sum_{n=2}^{k-1} \frac{\U^*f_n}{|D|^{n+1}} \Big)\U^*\varphi_2\Big\ras\Big|\leq C \big|\las \U^*\varphi_1 , \frac{\big(d_\beta( \cdot ) \big)^k }{\D^{k+1}} \U^* \varphi_2 \ras \big|
\end{equation}
and by \eqref{eq_app_03} we arrive at
\begin{equation}
C \big|\las \U^*\varphi_1 , \frac{\big(d_\beta( \cdot ) \big)^k }{\D^{k+1}} \U^* \varphi_2 \ras \big|=C \big|\las\varphi_1 , \frac{\Vert \cdot\Vert^k }{\D^{k+1}}\U \U^* \varphi_2 \ras \big|.
\end{equation}
Now \eqref{eq_app_13} follows, using the Cauchy-Schwarz inequality and the exponential decay of $\varphi_2$ from assumption \eqref{eq_app_14}.
\end{proof}
%
%
To simplify the notation in the remainder of the section, we set
\begin{equation}
\tilde \phi:= \U^* \phi,\quad \tilde \phi_2 := \U^*\phiB,\quad \tilde \phiC :=\U^* \phiC \text{ and } \tilde g := \U^* g.
\end{equation}
%
%
\begin{lem}\label{lem_app_04} 
Let $\phiA,\phiB,\phiC,g \in \Hfa$ and $\gamA,\gamB, \gamC\in \C$ be as defined in \eqref{eq_below_49}-\eqref{eq_below_50}. For any fixed $\delta>0$, there exist $C>0,D_0>0$ such that for $\D>D_0$
\begin{equation}
\begin{split}
\las \jbs\psi,I_\beta \jbs\psi\ras&\geq  2|D|^{-6}\re \gamA\overline{\gamB}\Vert \phiB\Vert_1^2+2|D|^{-8}\re \gamA \overline{\gamC}\Vert \phiC\Vert^2_1  \\
&\quad -C\frac{|\gamA|^2+|\gamB|^2+|\gamC|^2}{|D|^{10}}-\delta\Vert g\Vert^2.
\end{split}
\end{equation}
\end{lem}
\begin{proof}
Note that by definition of $I_\beta^o$ in \eqref{eq_below_04} one has  
\begin{equation}\label{eq_below_03}
\las J_\beta \psi,I_\beta J_\beta \psi\ras = \las J_\beta \psi,I_\beta^o J_\beta \psi\ras,
\end{equation}
and, according to \eqref{eq_below_50},
\begin{equation}
\begin{split}
J_\beta \psi&=\U^* \big(\gamA \phiA+\D^{-3}\gamB\phiB + \D^{-4}\gamC \phiC +g\big)\\
&= \gamA \tilde\phiA+\D^{-3}\gamB\tilde\phiB + \D^{-4}\gamC \tilde\phiC +\tilde g.
\end{split}
\end{equation}
Using this we can split the expression on the r.h.s of \eqref{eq_below_03} into the terms
\begin{equation}\label{eq_below_16}
\begin{split}
&\las \jbs\psi, I_\beta^o \jbs\psi\ras \\
&=  |\gamA|^2\las \tilde\phiA,I_\beta^o\tilde \phiA\ras +\frac{2\re \gamA \overline{\gamB}}{|D|^3}\las \tilde \phiB,I_\beta^o \tilde \phiA\ras+\frac{2\re \gamA\overline{\gamC}}{|D|^4}\las \tilde \phiC,I_\beta^o \tilde \phiA\ras   \\
&\quad +2\re \gamA\las \tilde g,I_\beta^o \tilde \phiA\ras+\frac{|\gamB|^2}{|D|^6}\las\tilde  \phiB,I_\beta^o \tilde \phiB\ras + \frac{2\re \gamB \overline{\gamC}}{|D|^7}\las \tilde \phiC,I_\beta^o \tilde \phiB\ras  \\
&\quad +\frac{2\re \gamB}{|D|^3}\las\tilde  g,I_\beta^o \tilde \phiB\ras+\frac{|\gamC|^2}{|D|^8}\las \tilde \phiC,I_\beta^o\tilde  \phiC\ras +\frac{2 \re \gamC}{|D|^4} \las\tilde  g,I_\beta^o \tilde \phiC\ras + \las \tilde g,I_\beta^o \tilde g\ras\\
&= B_1+B_2+B_3+\las\tilde  g, I_\beta^o\tilde  g \ras, 
\end{split}
\end{equation}
where
\begin{equation}\label{eq_below_37}
B_1:=|\gamA|^2\las \tilde \phiA,I_\beta^o \tilde \phiA\ras + \frac{2 \re \gamA \overline{\gamB}}{|D|^3}\las \tilde \phiB,I_\beta^o \tilde \phiA\ras + \frac{2 \re \gamA \overline{\gamC}}{|D|^4}\las \tilde \phiC,I_\beta^o \tilde \phiA\ras 
\end{equation}
\begin{equation}\label{eq_below_38}
B_2:= \frac{|\gamB|^2}{|D|^6}\las \tilde \phiB,I_\beta^o \tilde \phiB\ras+\frac{2 \re \gamB \overline{\gamC}}{|D|^7}\las \tilde \phiC,I_\beta^o \tilde \phiB\ras +\frac{|\gamC|^2}{|D|^8}\las \tilde \phiC,I_\beta^o \tilde \phiC\ras 
\end{equation}
and
\begin{equation}
B_3:=2\re \gamA\las \tilde g,I_\beta^o \tilde \phiA\ras + \frac{2\re \gamB}{\D^3}\las \tilde g,I_\beta^o \tilde \phiB\ras + \frac{2 \re \gamC}{\D^4} \las\tilde  g, I_\beta^o \tilde \phiC\ras.
\end{equation}
In Propositions \ref{prop_below_07}, \ref{prop_below_08} and \ref{prop_below_09} we bound these three terms separately. We obtain
\begin{equation}\label{eq_below_08}
B_1\geq \frac{2 \re \gamA \overline{\gamB}}{|D|^6}\Vert \phiB \Vert_1^2+\frac{2\re \gamA\overline{\gamC}}{|D|^8} \Vert\phiC\Vert_1^2-C\frac{|\gamA|^2+|\gamB|^2+|\gamC|^2}{|D|^{10}},
\end{equation}
\begin{equation}
B_2 \geq  -C\frac{|\gamB|^2+|\gamC|^2}{|D|^{10}}
\end{equation}
and we show that for any $\delta>0$ there exist $C,D_0>0$ such that for all $\D>D_0$ we have
\begin{equation}\label{eq_below_05}
B_3 \geq -C\frac{|\gamA|^2}{|D|^{10}}-C\frac{|\gamB|^2}{\D^{12}}-C\frac{|\gamC|^2}{\D^{14}}-\frac{\delta}{2}\Vert g\Vert^2.
\end{equation}
For the term $\las \tilde g,I_\beta^o \tilde g\ras$ in \eqref{eq_below_16} we use the fact that on the support of $J_\beta$, the distance between particles belonging to different subsystems grows proportionally to $\D$. Thus for any $\delta>0$ we can choose $D_0> 0$ such that for $\D>D_0$
\begin{equation}\label{eq_below_28}
\las\tilde  g,I_\beta^o \tilde g\ras \geq -\frac{\delta}{2}\Vert \tilde g \Vert^2 =-\frac{\delta}{2}\Vert g \Vert^2 .
\end{equation}
Collecting the estimates \eqref{eq_below_08} - \eqref{eq_below_28} proves the lemma.
\end{proof}
%
%
\begin{prop}[Estimate of $B_1$]\label{prop_below_07} We have
\begin{equation}
B_1\geq\frac{2 \re \gamA \overline{\gamB}}{|D|^6}\Vert \phiB \Vert_1^2+\frac{2\re \gamA\overline{\gamC}}{|D|^8} \Vert\phiC\Vert_1^2-C\frac{|\gamA|^2+|\gamB|^2+|\gamC|^2}{|D|^{10}}.
\end{equation}
\end{prop}
\begin{proof}
By Condition 2) of Theorem~\ref{thm_pre_02}, for each $\ia\prec\alpha$, the functions in $P^\ia\Wt$ transform according to the $\ell=0$ degree irreducible representation of $SO(3)$. In particular the one electron densities are spherically symmetric with respect to their associated nucleus, see \cite{ioannis}. Due to mutual orthogonality of the spaces $P^\ia\Wt$ for different $\ia$, all functions in $\Wt$ have this property. 
Applying Newton's theorem (\cite[Theorem 9.7]{lieb2001analysis}) we get
\begin{equation}
|\gamA|^2\las\tilde  \phiA,I_\beta^o \tilde \phiA\ras = 0.
\end{equation}
For the second term of \eqref{eq_below_37}, by Lemma \ref{lem_app_01} and Lemma~\ref{lem_app_06} we get
\begin{equation}
\begin{split}
\las \tilde \phiB,I_\beta^o \tilde \phiA\ras &\geq |D|^{-3}\las \tilde \phiB, (\U^*f_2) \tilde \phiA\ras +|D|^{-4}\las \tilde \phiB, (\U^* f_3)\tilde \phiA\ras+ |D|^{-5}\las  \tilde \phiB ,(\U^* f_4)\tilde \phiA\ras \\
&\quad+|D|^{-6} \las \tilde \phiB,(\U^*f_5) \tilde \phiA\ras-C\Big|\Big\las \tilde  \phiB, \frac{ \big(d_\beta( \cdot)\big)^6}{\D^7} \tilde \phi\Big\ras\Big|.
\end{split}
\end{equation}
Notice that for $l=2,3,4,5$ we have
\begin{equation}\label{eq_app_22}
\las \tilde \phiB , (\U^* f_l) \tilde \phi\ras = \las \U^* \phiB , (\U^* f_l) (\U^* \phi)\ras = \las \phiB,f_l \phi\ras.
\end{equation}
We will use the following orthogonality relations between $\phiB$ and $f_l\phi,\ l=3,4,5$ from Lemma~\ref{lem_app_03B}:
\begin{equation}\label{eq_below_17}
\las  \phiB, f_3\phiA\ras = \las  \phiB, f_4 \phiA\ras= \las  \phiB,  f_5 \phi\ras =0.
\end{equation}
This implies
\begin{equation}
\las \tilde \phiB,I_\beta^o\tilde \phiA\ras \geq \D^{-3} \las \phiB,f_2 \phi\ras-C\Big|\Big\las  \tilde \phiB, \frac{ \big(d_\beta( \cdot)\big)^6}{\D^7} \tilde \phi\Big\ras\Big|.
\end{equation}
Note that by Remark~\ref{rema_app_01} and due to exponential decay of the function $\phi$ we have
\begin{equation}
C \Big|\las  \tilde \phiB, \frac{\big(d_\beta( \cdot)\big)^6}{\D^7} \tilde \phi\ras\Big|\leq C |D|^{-7}\Vert \phi_2\Vert \Vert \phi\Vert.
\end{equation}
By definition of the semi--norm, see \eqref{eq_below_01}, $\las \phiB, f_2 \phi\ras =\Vert \phiB\Vert_1^2$ and since $ 2\re \gamA\overline \gamB\leq |\gamA|^2+|\gamB|^2$ we get
\begin{equation}\label{eq_below_23}
\frac{2\re \gamA \overline{\gamB}}{|D|^3}\las \tilde \phiB, I_\beta^o \tilde \phiA\ras\geq \frac{2 \re \gamA \overline{\gamB}}{|D|^6}\Vert \phiB \Vert_1^2-C\frac{|\gamA|^2+|\gamB|^2}{|D|^{10}}.
\end{equation}
Now we estimate the last term in \eqref{eq_below_37}. Since $\phi$ decays exponentially we can apply Corollary~\ref{cor_app_03} with $k=5$ and proceeding as in \eqref{eq_app_22} yields
\begin{equation}\label{eq_below_18}
\las\tilde  \phi_3,I_\beta^o\tilde  \phi\ras\geq |D|^{-3}\las \phiC,   f_2\phiA\ras +|D|^{-4}\las\phiC,  f_3 \phiA\ras + |D|^{-5}\las \phiC,   f_4\phiA\ras-C\D^{-6}\Vert \phi_2\Vert \ \Vert \phi\Vert.
\end{equation}
According to Lemma~\ref{lem_app_03B} the first and third summand of \eqref{eq_below_18} are zero and we get
\begin{equation}\label{eq_below_19}
\frac{\gagas}{|D|^4}\las \tilde \phiC,I_\beta^o \tilde \phiA\ras \geq  \frac{\gagas}{|D|^8}\Vert \phiC \Vert ^2_1-C\frac{|\gamA|^2+|\gamC|^2}{|D|^{10}}.
\end{equation}
\end{proof}
%
%
\begin{prop}[Estimate of $B_2$]\label{prop_below_08} There exists a constant $C>0$ such that
\begin{equation}
B_2\geq -C\frac{|\gamB|^2+|\gamC|^2}{|D|^{10}}
\end{equation}
\end{prop}
\begin{proof}
Recall that
\begin{equation}\label{eq_app_08}
B_2= \frac{|\gamB|^2}{|D|^6}\las \tilde \phiB,I_\beta^o \tilde \phiB\ras+\frac{2 \re \gamB \overline{\gamC}}{|D|^7}\las \tilde \phiC,I_\beta^o \tilde \phiB\ras +\frac{|\gamC|^2}{|D|^8}\las \tilde \phiC,I_\beta^o \tilde \phiC\ras. 
\end{equation}
For the first term on the r.h.s. of \eqref{eq_app_08}, since $\phiB$ decays exponentially (see Corollary~\ref{cor_exp_01}), we can use Corollary~\ref{cor_app_03} with $k=3$ and the analogous to \eqref{eq_app_22} we get
\begin{equation}\label{eq_app_101}
\las \tilde \phiB, I_\beta^o \tilde \phiB \ras \geq \D^{-3} \las \phiB, f_2 \phiB\ras-C\D^{-4}\Vert \phiB\Vert^2 .
\end{equation}
By Lemma~\ref{lem_app_02} we have $\las \phiB, f_2 \phiB\ras =0$ which implies
\begin{equation}
\frac{|\gamma_2|^2}{\D^6} \las\tilde  \phiB, I_\beta^o \tilde \phiB \ras\geq -C\frac{|\gamma_2|^2}{\D^{10}}.
\end{equation}
To bound the second and third term on the r.h.s. of \eqref{eq_app_08} we apply Corollary~\ref{cor_app_03} with $k=2$ to get
\begin{equation}\label{eq_below_24}
\frac{2 \re \gamB \overline{\gamC}}{|D|^7}\las \tilde \phiC,I_\beta^o \tilde \phiB\ras \geq -C\frac{|\gamB|^2+|\gamC|^2}{|D|^{10}}
\end{equation}
and
\begin{equation}\label{eq_below_27}
\begin{split}
\frac{|\gamC|^2}{|D|^8}\las \tilde \phiC, I_\beta^o \tilde \phiC\ras 
&\geq -C\frac{|\gamC|^2}{|D|^{11}}.
\end{split}
\end{equation}
\end{proof}
%
%
\begin{prop}[Estimate of $B_3$]\label{prop_below_09}For any fixed $\delta>0$ there exist $C>0$ and $D_0>0$ such that for $\D>D_0$ we have
\begin{equation}\label{eq_below_39}
 B_3  \geq -C\frac{|\gamA|^2}{|D|^{10}}-C\frac{|\gamB|^2}{\D^{12}}-C\frac{|\gamC|^2}{\D^{14}}-\frac{\delta}{2}\Vert g\Vert^2.
\end{equation}
\end{prop}
\begin{proof}
Recall 
\begin{equation}\label{eq_app_15}
B_3=2\re \gamA\las \tilde g,I_\beta^o \tilde \phiA\ras + \frac{2\re \gamB}{\D^3} \las \tilde g,I_\beta^o \tilde \phiB\ras + \frac{2 \re \gamC}{\D^4} \las \tilde g, I_\beta^o \tilde \phiC\ras.
\end{equation}
For the first term, by Corollary~\ref{cor_app_03} with $k=4$ and the analogous to \eqref{eq_app_22} we get
\begin{equation}\label{eq_below_20}
2\re \gamA\las \tilde g,I_\beta^o \tilde \phiA\ras \geq  2\re \gamA|D|^{-3} \las g,  f_2 \phiA\ras +2\re \gamA|D|^{-4}\las g, f_3\phiA\ras -C|\gamA| \D^{-5}\Vert g \Vert  \Vert \phi\Vert
\end{equation}
where by definition of $g$ we have
\begin{equation}
\las g, f_2 \phiA\ras = \las g,\phiB\ras _1=0
\end{equation}
and
\begin{equation}
\las g, f_3\phiA\ras=\las g,\phi_3\ras_1 =0.
\end{equation}
This implies
\begin{equation}\label{eq_below_25}
2\re \gamA\las \tilde g,I_\beta^o \tilde \phiA\ras\geq -C|\gamA| \D^{-5}\Vert g \Vert  \Vert\phi\Vert.
\end{equation}
By Corollary~\ref{cor_app_03} with $k=2$ we get
\begin{equation}
\frac{2\re \gamB}{|D|^3} \las \tilde g,I_\beta^o \tilde \phiB\ras \geq -C|\gamB|\D^{-6}\Vert g \Vert \Vert \phiB\Vert
\end{equation}
and
\begin{equation}\label{eq_below_26}
\frac{2\re\gamC }{|D|^4}\las \tilde g,I_\beta^o \tilde \phiC\ras \geq -C|\gamC| \D^{-7}\Vert g \Vert  \Vert\phiC\Vert.
\end{equation}
Applying Young's inequality for products in \eqref{eq_below_25}-\eqref{eq_below_26} yields the result.
\end{proof}
%
%
%
%
\section{Orthogonality relations}\label{sec_app2} 
In this section we prove several orthogonality relations, which follow from the symmetry properties of functions in $\Wt$. Let $\mathcal P ^{(i)}:\ltn\rightarrow\ltn$ such that
\begin{equation}\label{eq_app_20}
(\mathcal P ^{(i)}\varphi)(x):=\varphi(x_1,\cdots,x_{i-1},-x_i,x_{i+1},\cdots,x_N)
\end{equation}
and define $$\pca:=\prod_{i\in \cc_1}\mathcal P ^{(i)}.$$ 
As usual we say that a function $\varphi\in \ltn$ is $\pca$-\textit{even} iff $\pca\varphi=\varphi$. A function $\varphi\in \ltn$ is called $\pca$-\textit{odd} iff $\pca \varphi=-\varphi$. Similarly, we define the operator $\pcb$ and set $$\pcn:=\pca\pcb.$$
%
%
\begin{lem}\label{lem_app_11}
Let $\ia\prec\alpha$ an irreducible representation of $S_\beta$ be such that $P^\ia \Wt\neq \emptyset$. For $\mathcal{P}_{\bullet}=\pca,\pcb,\pcn$ we have
\begin{center}
either: all functions $\phi\in P^\ia \Wt$ are $\mathcal{P}_{\bullet}$--even
\\\hspace{12mm}or: all functions $\phi\in P^\ia \Wt$ are $\mathcal{P}_{\bullet}$--odd.
\end{center}
\end{lem}
\begin{proof}
From the definition of $\Ht_\beta$ it is apparent that $\mathcal{P}_{\bullet}\Ht_\beta\mathcal{P}_{\bullet}=\Ht_\beta$. Consequently the $\mathcal{P}_{\bullet}$--even and the $\mathcal{P}_{\bullet}$--odd  functions are invariant subspaces of $\Ht_\beta$. By Condition 2) we have $\dim(P^\ia \Wt)=\dim \ia$ and since $\ia$ is irreducible it can not contain nontrivial invariant subspaces, so either all functions in $P^\ia \Wt$ are $\mathcal{P}_{\bullet}$--even or all functions in $P^\ia \Wt$ are $\mathcal{P}_{\bullet}$--odd.
\end{proof}
%
%
\begin{lem}\label{lem_app_03A}
For any $\phi\in \Wt$ we have

\begin{equation}\label{eq_app_23}
 \las \phi,f_2 \phi\ras =\las \phi,f_3\phi\ras=0
\end{equation}
\end{lem}
\begin{proof}
Recall the definitions
\begin{equation}
f_2(x)=\sumij -e^2\big(3(x_i \cdot e_D)(x_j\cdot e_D)-x_i\cdot x_j\big)
\end{equation}

and
\begin{equation}\label{eq_app_07}
\begin{split}
f_3(x) =\sumij & \frac{e^2}{2} \Big(  3(x_i-x_j)\cdot e_D\big[2(x_i\cdot x_j)-5(x_i\cdot e_D)(x_j\cdot e_D)\big]\\
 & \qquad+ 3|x_i|^2(x_j\cdot e_D)-3|x_j|^2(x_i\cdot e_D)\Big).
 \end{split}
\end{equation}
It is easy to see that $f_2$ is $\pca$--even and $\pcb$--odd. Note that $f_2$ is invariant under permutations  in $S_\beta$ which preserve the cluster decomposition $\beta$. Hence multiplication by $f_2$ commutes with the projection $P^\ia$. Since the spaces $P^\ia \Wt$ are mutually orthogonal for different $\ia$, for all $\phi\in \Wt$ we have
\begin{equation}\label{eq_app_05}
\las \phi,f_2 \phi\ras = \sum_{\ia\prec\alpha}\las P^\ia \phi,f_2 P^\ia\phi\ras.
\end{equation}
Since $|P^\ia \phi|^2$ is $\pca$--even and $f_2$ is $\pca$--odd we get
\begin{equation}
\las P^\ia \phi,f_2 P^\ia \phi\ras =0.
\end{equation}
Similarly, from the explicit expression of $f_3$ in \eqref{eq_app_07} follows
\begin{equation}
\pcn f_3=-f_3
\end{equation}
which yields 
\begin{equation}
\las \phi,f_3\phi\ras=\sum_{\ia\prec\alpha}\las P^\ia\phi,f_3 P^\ia\phi\ras = \sum_{\ia\prec\alpha}\las P^\ia\phi,(\pcn f_3 )P^\ia\phi\ras =0.
\end{equation}
\end{proof}
%
%
\begin{corr}\label{cor_app_01}
For any $\phi\in \Wt$ the functions
\begin{equation}
\phi_k:=(\tilde H_\beta^\alpha-\mu^\alpha)^{-1} f_k\phi, \quad k=2,3
\end{equation}
are well defined.
\end{corr}
\begin{proof} This is an immediate consequence of Lemma~\ref{lem_app_03A}, since it states that $f_k\phi$ is orthogonal to $\phi$.
\end{proof}
%
%
\begin{corr} \label{cor_app_02}
For any $\phi\in \Wt$ we have
\begin{equation}
\quad \las \phi,\phi_2\ras_1=\las \phi,\phi_3\ras_1 = 0.
\end{equation}
\end{corr}
\begin{lem}\label{lem_app_08}
For any $\phi\in \Wt$ we have
\begin{equation}
\las \phiB,  f_2 \phiB\ras =0 .
\end{equation}
\end{lem}
\begin{proof}
By the same argument used in Lemma~\ref{lem_app_03A}, since $f_2$ appears three times in the expression
\begin{equation}
\begin{split}
\las \phiB, f_2 \phiB\ras=\las \resbt f_2 \phi,f_2 \resbt f_2 \phi\ras
\end{split}
\end{equation}
applying $\pca$ results in a change of sign which yields the result.
\end{proof}
%
%
\begin{lem}\label{lem_app_03B}
For any $\phi\in \Wt$ and with $\phi_2,\ \phi_3$ defined in Corollary~\ref{cor_app_01} we have
\begin{equation}
\begin{split} 
& i) \quad\las \phi_2,\phi_3\ras=\las \phi_2,f_3\phi\ras = \las \phi_2,f_5\phi\ras = 0 \\
&ii) \quad \las \phi_3,f_2 \phi\ras=\las \phi_3,f_4 \phi\ras = 0.
 \end{split}
\end{equation}
\end{lem}
\begin{proof}
Notice that the Legendre polynomials fulfill
\begin{equation}
P_n(-z)=(-1)^nP_n(z).
\end{equation}
In particular for $h,D\in \R^3$ we get
\begin{equation}
P_n \left(\frac{-h}{|h|}\cdot \frac{D}{|D|}\right)=(-1)^n P_n \left(\frac{h}{|h|}\cdot \frac{D}{|D|}\right)
\end{equation}
and thus
\begin{equation}
\big(\pcn f_n\big)(x)=(-1)^n f_n(x).
\end{equation}
Hence
\begin{equation}\label{eq_app_10}
\las f_2\phi,f_3\phi\ras = \big\las\big( \pcn f_2\big) \phi,\big(\pcn f_3\big)\phi\big\ras=-\las f_2\phi,f_3\phi\ras=0.
\end{equation}
Analogously
\begin{equation}\label{eq_app_12}
\las f_2 \phi,f_5\phi\ras =\las f_3 \phi,f_4\phi\ras=0.
\end{equation}
Since $\tilde H_\bmin$ commutes with $\pcn$, so do $(\tilde H_\bmin-\mu^\alpha)$ and $(\tilde H_\bmin-\mu^\alpha)^{-1}$.  Hence by the same argument we also get 
\begin{equation}
\las \phi_2,\phi_3\ras=\las \phi_2, f_3\phi\ras = \las \phi_2, f_5\phi\ras=\las \phi_3, f_2 \phi\ras=\las \phi_3,f_4 \phi\ras = 0.
\end{equation}
\end{proof}
%
%
In the next lemma we will use the $SO(3)$ symmetry of the system. 
%
%
\begin{lem}\label{lem_app_02}
For any $\phi\in \Wt$ and $\phi_2=( \tilde H_\bmin-\mu^\alpha)^{-1}f_2\phi$ we have
\begin{equation}\label{eq_app_02}
\las \phi_2, f_4 \phi\ras =0.
\end{equation}
\end{lem}
\begin{proof}
As a first step we notice that the functions $f_2$ and $f_4$ are the sums of Legendre polynomials of degrees $2$ and $4$ respectively. For the Legendre polynomials $P_k$ of order $k$ and the spherical harmonics $Y^m_\ell$ we have
$$
P_\ell(\cos \theta )=\sqrt{\frac{4\pi}{(2\ell+1)}}Y^0_\ell(\theta, \varphi).
$$
Note that in \eqref{eq_app_16}-\eqref{eq_app_18}, leading to the definition of $f_n$ in \eqref{eq_app_04}, for $\mathcal F^{(1)}_n$ we have $\cos \theta = \frac{x_i}{|x_i|}\cdot \frac{D}{\D}$, for $\mathcal F^{(2)}_n$ we have $\cos \theta = \frac{-x_j}{|x_j|}\cdot \frac{D}{\D}$ and in $\mathcal F^{(3)}_n$ we have $\cos \theta = \frac{x_i-x_j}{|x_i-x_j|}\cdot \frac{D}{\D}$ respectively.
Consequently the Legendre polynomials of order $\ell$ are transformed according to the irreducible representation of degree $\ell$ under the actions of the $SO(3)$ group, see \cite{hamermesh1962}.\par
By Condition~2) of Theorem~\ref{thm_pre_02}, the state $\phi$ belongs to the irreducible representation of degree $\ell=0$ of the group $SO(3)$. Thus the products $f_2\phi$ and $f_4\phi$ are transformed according to the representations of degree $\ell=2$ and $\ell=4$ respectively.\par 
By rotational invariance of the operator $\tilde H_\beta$, the function $\resbt f_2 \phi$ has the same symmetry a $f_2 \phi$, namely it transforms according to the irreducible representation of degree $\ell=2$.
\par 
But functions belonging to two different irreducible representations are orthogonal. This proves the lemma.
\end{proof}
%
%
%
\section{Remark on actions of the permutation group}\label{sec_sym}
Let $g\in L^2(\R^{3(m+n)})$ be a function depending on position vectors of $(m+n)$ particles. Let $A$ be an operator on $L^2(\R^{3m})$ and $g\in \mathcal{D}(A\otimes \mathbbm{1}^{3n})$, so that $A\otimes \mathbbm{1}^{3n}$ acts on $g$ as a function of the first $m$ position vectors.\par 
\begin{lem}\label{lem_sym} Assume that for some $R>0$ we have $\supp (g) \subset \{\xi \in \R^{3(m+n)}, |\xi_i|<R \ i=1,\cdots,m, |\xi_j|\geq 2R \ j\geq m+1\}$. Let $S_{m+n}$ be the permutation group of $(m+n)$ particles and $\pi\in S_{m+n}$ such that $\pi\notin S_m\otimes S_n$. In other words $\pi$ exchanges at least one of the first $m$ particles with a particle labelled  by $j\geq m+1$. Then
\begin{equation}\label{eq_app_21}
\supp\big((A\otimes \mathbbm{1}^{3n} )g\big) \cap \supp\big(\mathcal T _\pi g \big)=\emptyset
\end{equation}
where $\T _\pi g(\xi)=g(\xi_{\pi^{-1}(1)},\cdots ,\xi_{\pi^{-1}(m+n)})$.
\end{lem}
\begin{proof}
For local operators $A$ this relation was first used by Sigalov and Zhislin to prove existence of an eigenvalue
of atoms with arbitrary types of rotational and permutational symmetry \cite{SigalovZhislin}. If the operator is local, \eqref{eq_app_21} can be rewritten as
\begin{equation}
\supp(g) \cap \supp \big(\mathcal T _\pi g\big) = \emptyset.
\end{equation}
If $A$ is a non-local operator, \eqref{eq_app_21} is still true, because for at least one particle $i_0\geq m+1$ we have
\begin{equation}
|\xi_{i_0}|>2R\quad  \text{ on } \supp\big((A\otimes \mathbbm{1}^{3n})g\big)
\end{equation}
and
\begin{equation}
|\xi_{i_0}|<R\quad  \text{ on } \supp\big(\mathcal T_\pi g\big).
\end{equation}
\end{proof}
\vspace{1cm}
\textbf{Acknowledgements:} Semjon Vugalter thanks the University of Toulon and Jean-Marie Barbaroux thanks the Karlsruhe Institute of Technology for their hospitality and financial supports for their  visits, during which a part of this work was done. 

Dirk Hundertmark and Semjon Vugalter are funded by the Deutsche Forschungsgemeinschaft (DFG, German Research Foundation) -- Project-ID 258734477 -- SFB 1173. 

We would also like to thank the Centre International de Rencontres Math\'ematiques (CIRM) in Luminy, whose \textit{REB program}  
enabled a scientific exchange at an important stage of the work and Ioannis Anapolitanos for discussions on the van der Waals--London asymptotic. 
\vspace{1cm}
%
%
\bibliographystyle{plain}
\bibliography{bibfile}
\end{document}